%% file: main.tex
\documentclass[graybox,envcountsame]{svmult}

\usepackage[paperwidth=155mm,
            paperheight=235mm,
            left=16mm,
            right=16mm,
            top=17mm,
            bottom=25mm
]{geometry}

\usepackage[T1]{fontenc}
\usepackage[utf8]{inputenc}
\usepackage{microtype}
\usepackage[tt=false, type1=true]{libertine}

\definecolor[named]{ACMBlue}{cmyk}{1,0.1,0,0.1}
\definecolor[named]{ACMYellow}{cmyk}{0,0.16,1,0}
\definecolor[named]{ACMOrange}{cmyk}{0,0.42,1,0.01}
\definecolor[named]{ACMRed}{cmyk}{0,0.90,0.86,0}
\definecolor[named]{ACMLightBlue}{cmyk}{0.49,0.01,0,0}
\definecolor[named]{ACMGreen}{cmyk}{0.20,0,1,0.19}
\definecolor[named]{ACMPurple}{cmyk}{0.55,1,0,0.15}
\definecolor[named]{ACMDarkBlue}{cmyk}{1,0.58,0,0.21}
\usepackage{hyperref}
\hypersetup{
    colorlinks,
    linkcolor=ACMPurple,
    citecolor=ACMPurple,
    urlcolor=ACMDarkBlue,
    filecolor=ACMDarkBlue}

\input{packages}

\input{commands}
\input{environments}

\input{tikz_commands}

\usepackage[capitalize,nosort,nameinlink]{cleveref}

\crefformat{footnote}{#2\footnotemark[#1]#3}
\crefname{section}{Sect.}{Sects.}
\Crefname{section}{Section}{Sections}
\crefname{algorithm}{Alg.}{Algs.}
\Crefname{algorithm}{Algorithm}{Algorithms}
\crefname{theorem}{Thm.}{Thms.}
\crefname{line}{line}{lines}
\crefalias{ALC@unique}{line}
\crefalias{ALC@line}{line}

\usepackage{makeidx}         
\usepackage{graphicx}        
\usepackage{multicol}        
\usepackage[bottom]{footmisc}

\makeindex             

\renewcommand{\orcidID}[1]{{\href{https://orcid.org/#1}{\protect\raisebox{3.25pt}{\protect\includegraphics[alt={orcid}]{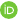}}}}}

\newcounter{auxctr}
\newcounter{cor-memory-safety}
\newcounter{thm-soundness-graph-construction-main-text}
\newcounter{thm-termination}

\begin{document}

\title*{\tool{AProVE}: Modular Termination Analysis of Memory-Manipulating \tool{C}
  Programs}
\author{Frank Emrich,\orcidID{0000-0002-8591-2740}\\ Jera
  Hensel,\orcidID{0000-0003-2852-9830} and\\J\"urgen Giesl\orcidID{0000-0003-0283-8520}}
\institute{Frank Emrich \at  The University of Edinburgh, Edinburgh, UK, \email{frank.emrich@ed.ac.uk}
\and Jera Hensel \at  RWTH Aachen University, Aachen, Germany, \email{hensel@cs.rwth-aachen.de}
\and J\"urgen Giesl \at  RWTH Aachen University, Aachen, Germany, \email{giesl@cs.rwth-aachen.de}}
%
%
\maketitle

\abstract{Termination
  analysis
  of \tool{C} programs is a challenging task. On the one
hand, the analysis needs to be precise enough to draw meaningful conclusions.
On the other hand, relevant programs in practice
are
large and require substantial abstraction. It is this inherent trade-off that
is the crux of the problem. In this work, we present \tool{AProVE}, a tool
that uses symbolic execution to analyze termination of memory-manipulating
\tool{C} programs. 
While traditionally, \tool{AProVE}'s focus was on the preciseness of the
analysis, we describe how we adapted our approach
towards a modular analysis. Due to this  adaption, our approach can now also handle recursive
programs. Moreover, we present
further performance improvements which we developed to make \tool{AProVE} scale to
large programs.}


\input{introduction}

\input{domain}
\input{graphconstruction}

\input{its}

\input{evaluation}

\begin{acknowledgement}
This research was partly funded by the Deutsche
  Forschungsgemeinschaft (DFG, German Research Foundation) - 
  235950644 (Project GI 274/6-2).
\end{acknowledgement}

\ethics{Competing Interests}{The authors have no conflicts of interest to declare that are relevant to the content of this chapter.}

\input{appendix}

\input{proofs}

\bibliography{references}

\end{document}

%% file: packages.tex
\usepackage{amssymb} 

\usepackage{varioref}
\usepackage{listings} 
\usepackage{lipsum}	
\usepackage[pangram]{blindtext}
\usepackage[normalem]{ulem}
\usepackage{placeins}
\usepackage{cprotect}
\usepackage{wrapfig}
\usepackage{subcaption}
\usepackage{fancyvrb}
\usepackage{boxedminipage}
\usepackage{array}
\usepackage{stmaryrd}
\usepackage[shortlabels]{enumitem}
	
\usepackage{graphicx}
\usepackage{ wasysym }
\usepackage{tikz}
\usepackage{pgfplots}
\usetikzlibrary{arrows}
\usetikzlibrary{patterns}
\usetikzlibrary{shapes,arrows,calc,arrows}
\usetikzlibrary{automata,arrows,decorations.pathmorphing,positioning,fit}
\usetikzlibrary{trees,shapes}
\usetikzlibrary{decorations.pathreplacing}
\usetikzlibrary{calc}
\usetikzlibrary{intersections}
\usetikzlibrary{backgrounds}
\usetikzlibrary{fadings}
\usetikzlibrary{tikzmark}

\usepackage{calc}
\usepackage{tkz-euclide}
\usepackage{catchfilebetweentags}
\usepackage{afterpage}
\usepackage{url}
\usepackage{todonotes}
\usepackage{xstring}
\usepackage{nameref}
\usepackage{alphalph}
\usepackage{relsize}
\usepackage{alltt}

\usepackage[utf8]{inputenc} 
\usepackage{algorithm2e} 
\usepackage{scrextend}
\usepackage{refcount} 



%% file: commands.tex
\renewcommand{\emptyset}{\varnothing}

\newcommand{\code}[1]{\texttt{#1}}
\newcommand{\tool}[1]{\textsf{#1}}
\newcommand{\sfC}{\tool{C}}
\newcommand{\LLVM}{\tool{LLVM}\xspace}

\newcommand{\aprove}{\tool{AProVE}}

\newcommand{\clang}{\tool{Clang}\xspace} 

\newcommand{\PP}{\mathcal{P}}
\newcommand{\Blks}{\mathit{Blks}}
\newcommand{\Ids}{\mathcal{V}_{\PP}}
\newcommand{\Idsi}{\mathcal{V}_{\PP}^{\mathit{fr}}}
\newcommand{\Pos}{p}
\newcommand{\GG}{\mathcal{G}}
\newcommand{\PPos}{\mathit{Pos}}

\renewcommand{\I}{\mathcal{I}}

\newcommand{\bigast}{\mathop{\mathlarger{\mathlarger{\mathlarger{*}}}}}

\newcommand{\pointstosl}{\hookrightarrow}

\newcommand{\SL}{\mathit{SL}}
\newcommand{\FOL}{\mathit{FO}}
\newcommand{\CON}{\mathit{CON}}

\newcommand{\N}{\mathbb{N}}

\newcommand{\Z}{\mathbb{Z}}

\newcommand{\corrvar}{\leadsto}

\newcommand{\Prog}{\mathcal{P}}

\newcommand{\cto}{\to_{\mbox{\scriptsize{\LLVM}}}}
\newcommand{\ccto}{\to_{\mbox{\scriptsize{\emph{\LLVM}}}}}

\newcommand{\ERROR}{\mathit{ERR}}

\newcommand{\nxt}[1]{\overline{#1}}

\newcommand{\renamed}{\delta}

\newcommand{\conttop}[1]{\ensuremath{\widehat{#1}}} 

\newcommand{\codevar}[1]{\ensuremath{\mathtt{#1}}}
\newcommand{\codevarx}[1]{\protect{\ensuremath{\mathtt{#1}}}}

\newcommand{\rfunc}{\codevar{func}}
\newcommand{\sfunc}{\codevar{func}}

\newcommand{\repbyR}[1][]{\ensuremath{\in^{\mathrm{w}}_{#1}}} 

\newcommand{\repby}{\repbyR[]}

\newcommand{\pointsto}[1][]{\hookrightarrow_{#1}} 

\newcommand{\domain}{\mathop{domain}} 

\newcommand{\PT}{\mathit{PT}}
\newcommand{\MT}{\PT}
\newcommand{\KB}{\mathit{KB}}
\newcommand{\CS}{\mathit{CS}}
\newcommand{\LV}{\mathit{LV}}
\newcommand{\FR}{\mathit{FR}}
\newcommand{\LVi}{\mathit{LV}_{\!i}}
\newcommand{\CV}{\mathit{VI}}

\newcommand{\llangle}{\langle\!\langle}
\newcommand{\rrangle}{\rangle\!\rangle}

\newcommand{\stateformula}[1]{\langle{#1}\rangle}
\newcommand{\extendedformula}[1]{\llangle{#1}\rrangle}
\newcommand{\stateformulaSL}[1]{\langle{#1}\rangle_{\mathit{SL}}}

\newcommand{\alloc}[2]{\ensuremath{\llbracket{}#1,\,#2\rrbracket}}

\newcommand{\QFIA}{\mathit{QF\_IA}}
\newcommand{\AL}{\mathit{AL}}
\newcommand{\ALL}[1][]{\ensuremath{\mathit{AL}_{#1}^\ast}}

\newcommand{\Vsym}{\mathcal{V}_{\mathit{sym}}}
\newcommand{\functionmap}[0]{\rightarrow}
\newcommand{\partialfunctionmap}[0]{\rightharpoonup}

\newcommand{\sizeOf}{\mathit{size}}

\newcommand{\slassignment}{\mathit{as}}
\newcommand{\slmemory}{\mathit{mem}}

\newcommand{\updatefct}[3]{#1[#2 := #3]}

\newcommand{\GraphFNodeOne}{A}
\newcommand{\GraphFNodeTwo}{B}
\newcommand{\GraphFNodeSix}{C}
\newcommand{\GraphFNodeSeven}{D}
\newcommand{\GraphFNodeTen}{E}
\newcommand{\GraphFNodeEleven}{F}
\newcommand{\GraphFNodeTwelve}{G}
\newcommand{\GraphFNodeThirteen}{H}
\newcommand{\GraphFNodeFourteen}{I}
\newcommand{\GraphFNodeFifteen}{J}
\newcommand{\GraphFNodeEighteen}{K}
\newcommand{\GraphFNodeNineteen}{L}
\newcommand{\GraphFNodeTwenty}{M}

\newcommand{\GraphMainNodeCall}{V}
\newcommand{\GraphMainNodeCallAbstraction}{W}
\newcommand{\GraphMainNodeIntersectionUnsat}{X}
\newcommand{\GraphMainNodeIntersection}{Y}
\newcommand{\GraphMainNodeAfterStore}{Z}

\newcommand{\lostAL}{\mathit{removedAL}}

\newcommand{\seqs}

\newcommand{\implies}{\Rightarrow}
\newcommand{\dotsc}{\ldots}
\newcommand{\text}[1]{\mbox{#1}}
\newcommand{\eqref}[1]{(\ref{#1})}

%% file: environments.tex
\DeclareMathAlphabet
{\mathttit}{OT1}{lmtt}{l}{it}

\newlist{labeledenumerate}{enumerate}{10}
\setlist[labeledenumerate]{leftmargin=*,labelindent=1em,topsep=5pt}

%% file: tikz_commands.tex
\tikzstyle{eval-edge}=[-stealth]
\tikzstyle{gen-edge}=[-stealth]
\tikzstyle{int-gen-edge}=[-stealth]
\tikzstyle{ca-edge}=[-stealth]
\tikzstyle{fs-edge}=[-stealth]
\tikzstyle{omit-edge}=[-stealth, dotted]

\tikzset{
    invisible/.style={opacity=0},
    visible on/.style={alt={#1{}{invisible}}},
    alt/.code args={<#1>#2#3}{%
      \alt<#1>{\pgfkeysalso{#2}}{\pgfkeysalso{#3}} 
    },
  }
\tikzset{onslide/.code args={<#1>#2}{%
  \only<#1>{\pgfkeysalso{#2}} 
}}

\pgfdeclaremetadecoration{middlezigzag}{straight}{
    \state{straight}[switch if less than=\pgfmetadecorationsegmentlength to final,
                  width=\pgfmetadecoratedpathlength/2 - \pgfmetadecorationsegmentlength/2 ,
                  next state=zigzag] {
        \decoration{curveto}
    }
    \state{zigzag}[width=\pgfmetadecorationsegmentlength,
                   next state=final] {
        \decoration{zigzag}
    }
    \state{final}{
        \decoration{curveto}
        \beforedecoration{\pgfpathmoveto{\pgfpointmetadecoratedpathfirst}}
    }
}

\tikzset{
    middle zigzag/.style={
        decorate,
        decoration={
            middlezigzag,
            meta-segment length=#1,
            segment length=0.5cm,
			amplitude=1.5pt
        }
    },
    middle zigzag/.default=1cm
}

%% file: introduction.tex
\section{Introduction}
\label{sect:introduction}

\aprove{} \cite{AProVE-JAR} is a tool for termination and complexity analysis of many programming
languages including \sfC{}. Its approach for termination analysis of \sfC{}
programs focuses in particular on the connection between memory addresses and their contents.
To
avoid handling all intricacies of \sfC, we use the \tool{Clang} compiler \cite{Clang} to transform
programs into the platform-independent
intermediate representation of the \LLVM{} Compilation Framework
\cite{LLVM}.
As we presented in~\cite{LLVM-JAR}, in the first step,
our technique constructs a \emph{symbolic execution graph} (SEG)  which
over-approximates all possible program runs and models memory addresses and contents
explicitly. As a prerequisite for termination, \aprove{} shows the absence of
undefined behavior during the construction of the SEG. In this way, our approach also
proves memory safety of the program.
Afterwards, the strongly connected components (SCCs) of the graph are transformed into
integer transition systems (ITSs) whose termination implies termination of the
original \sfC{} program.
To analyze termination of the ITSs,
we apply
standard techniques which are implemented in a 
back-end that \tool{AProVE} 
also uses for termination analysis of
other programming languages. Here, the satisfiability checkers \tool{Z3} \cite{Z3}, \tool{Yices} \cite{Yices}, and
\tool{MiniSAT} \cite{MiniSAT} are applied 
to solve the search problems that arise
during the termination proofs.
Moreover, we also use the tool \tool{KoAT} \cite{TOPLAS16,KoAT-IJCAR22} in the back-end,
which can analyze both termination and complexity of ITSs, see \cite{TACAS25}.

Sometimes, the SEG does not contain over-approximating steps but it
models the program precisely. Then, non-termination of the ITS resulting
from an
SCC of the graph together with a path from the root of the graph to
the respective SCC implies non-termination of the program. In this case,
our approach can also prove non-termination of \sfC{} programs \cite{TACAS17,TACAS22} by
using the tools \tool{LoAT} \cite{IJCAR22,LoAT-CADE23} and 
\tool{T2} \cite{T2} to show non-termination of the corresponding ITS. (\aprove's
own back-end does not support the analysis of ITSs where runs may only begin with designated start terms.)
While integers were considered to be unbounded in \cite{LLVM-JAR}, we extended our
approach to handle bitvector arithmetic and also discussed the use of our approach for
complexity analysis of \sfC{} programs in \cite{JLAMP}.

\smallskip

 \begin{center}
   \begin{tikzpicture}[font=\scriptsize]
     \node[shape=rectangle,draw=black, minimum width=5mm]
       (C) {\tool{C}};
     \node[shape=rectangle,draw=black, minimum width=5mm]
       (LLVM) [right=1.1 of C] {\tool{LLVM}};
     \node[shape=ellipse,  draw=black, minimum height=4mm, minimum width=5mm]
       (SEG) [right=0.6 of LLVM]
       {\parbox{.8cm}{\centerline{Symbolic}\centerline{Execution}\centerline{Graph}}};
     \node[shape=rectangle,  draw=black, minimum height=4mm, minimum width=5mm]
       (ITS) [right=0.8 of SEG, yshift=3mm]
       {ITS};
     \node[shape=rectangle,rounded corners, draw=black, minimum height=4mm, minimum width=27mm]
       (Complexity) [right=0.8 of ITS]
       {Complexity};
     \node[shape=rectangle,rounded corners, draw=black, minimum height=4mm, minimum width=27mm]
       (Termination) [above=0.15 of Complexity]
       {Termination};
     \node[shape=rectangle,rounded corners, draw=black, minimum height=4mm, minimum width=27mm]
       (Non-Termination) [below=0.15 of Complexity]
       {Non-Termination};
     \node[shape=rectangle,rounded corners, draw=black, minimum height=4mm, minimum width=27mm]
       (Safety) [below=0.15 of Non-Termination]
       {Memory Safety};
     \path[thick,->] (C.east) edge node[above=-0.05]{\tool{Clang}} (LLVM.west);
     \path[thick,->] (LLVM.east) edge (SEG);
     \path[thick,->,bend left=4] (SEG) edge (ITS.west);
     \path[thick,->,bend left=18] (ITS) edge (Termination.west);
     \path[thick,->] (ITS) edge node[above right=.07 and
       -.25]{\mbox{\tiny \tool{KoAT}}}
     (Complexity.west);
     \path[thick,->,bend right=18] (ITS) edge node[above right=-.05 and
       -.25]{\mbox{\tiny \tool{LoAT},\tool{T2}}} ($(Non-Termination.west)+(0,0.1)$); 
     \path[thick,->,bend right=5] (SEG) edge ($(Non-Termination.west)-(0,0.1)$);
     \path[thick,->,bend right=10] (SEG) edge (Safety.west);
   \end{tikzpicture}
 \end{center}

\medskip
 
We showed how our approach  
supports programs with several functions in \cite{LLVM-JAR}, but up to now it could not analyze
functions in a modular way and it could not deal with recursion.\footnote{A paragraph with a
preliminary announcement of an extension of our approach to recursion was given in our
report for \emph{SV-COMP 2017} \cite{TACAS17}.}
For symbolic execution, the approach of~\cite{LLVM-JAR} used an abstraction that only
considered the values of program variables and the memory. 

In this work, we extend this approach to also support the abstraction of call
stacks, which allows us to re-use previous analyses of auxiliary functions in a
modular way. Moreover, in this way we can
analyze recursive programs as well.
Our technique of abstracting from the exact shape of the call stack in the
symbolic execution graph is based on our earlier approach for termination analysis of
\tool{Java Bytecode} (\tool{JBC})
in \cite{RTA11}. However, \cite{RTA11}
 is tailored to \tool{JBC} and thus has to support \tool{Java}'s object
 orientation and memory model. In contrast, the analysis in the current paper
supports features that are not
present in \tool{JBC}, like explicit allocation and deallocation of memory, as well
as pointer arithmetic. So the challenge for the extension of our approach for \sfC{}
termination analysis is to combine the byte-accurate representation of
the memory with the modular handling of (possibly
recursive) functions.

We recapitulate the abstract states of our symbolic execution in \cref{sect:domain}
and introduce our new approach to construct SEGs that handle functions in a modular way in
\cref{sect:seg}.
As mentioned before,  we also prove the absence of undefined behavior during this construction.
Afterwards, we
present the transformation into
ITSs whose termination implies termination of the \sfC{} program (\cref{sect:its}).
\cref{sect:eval} discusses our implementation and points out \aprove's
strengths and weaknesses, gives an overview on related work, and evaluates our
contributions empirically in comparison to other tools.
 App.\ \ref{app:semantics-sl} discusses details on the semantics of abstract states that we omitted from
the main part of the paper. Finally, App.\  \ref{app:proofs} contains all
proofs.

\paragraph{\aprove{} at \emph{SV-COMP}} In 2014, the \emph{Termination} category was added
to the demonstration track of
the \emph{International Competition on Software Verification
  (SV-COMP)}.\footnote{See
    \url{https://sv-comp.sosy-lab.org/}.} 
Back then, our tool was only able to prove termination for non-recursive programs.
One year later, \emph{Termination} became an official category. We implemented first support
to handle recursion, which already led to many successful termination proofs of small
recursive programs at \emph{SV-COMP} 2015.
In 2015 and 2016, we integrated the treatment of bitvector arithmetic
and overflows into our tool. Moreover, we developed
two different approaches to prove non-termination, where the first approach is reflected by \aprove{}'s first non-termination proofs
at \emph{SV-COMP} 2016, and the second by more powerful non-termination results at
\emph{SV-COMP} 2017.
In the following year, we generalized the techniques that \aprove{} uses for recursive
functions
in order to modularize the analysis also for 
non-recursive functions.
Furthermore, we integrated heuristics for the analysis of large programs.
Both extensions are described in the current paper and led to a significant number of new termination
proofs for recursive programs and for large programs with several functions.
Since \emph{SV-COMP} 2019, \aprove{} is able to produce non-termination witnesses and
to analyze termination of simple programs with recursive data structures. In \cite{CADE23}, we
extended this approach to the handling of more complex programs where termination
depends on the shape and the contents of recursive data structures.

Due to personal reasons, we were not able to submit our tool to \emph{SV-COMP} 2020
and \emph{SV-COMP} 2021, but we participated in \emph{SV-COMP} 2022 and \emph{SV-COMP} 2025 again.
In all these years, \aprove{} was always among the top three
(and often first or second)
in the ranking of the
\emph{Termination} category.

\paragraph{Limitations} As discussed in~\cite{LLVM-JAR}, some features of \LLVM are not
yet supported by our approach (e.g., we do not handle \code{undef}, floating point numbers,
or vectors).
Moreover, to ease the presentation, we do not regard
\code{struct}  types and
 we
again
disregard integer overflows and treat integer
types as unbounded in this paper.
For simplicity, we assume a 1 byte data alignment (i.e., values may be stored
at any address).
However, the handling of arbitrary alignment
is implemented in \aprove{}  and we refer to~\cite{LLVM-JAR} for
details. Finally, we do not consider \emph{disproving}
properties 
like memory safety or termination in this paper.

%% file: domain.tex
\section{Abstract Domain for Symbolic Execution}\label{sect:domain}

We use the following program from the \emph{Termination} category of \emph{SV-COMP} to
demonstrate our approach. Here, we  assume
\code{nondet\_int} to return a random integer.
The function \code{f} gets an integer pointer \code{p} as input. If the integer \code{*p}
is already negative, then the memory allocated by \code{p} is released and the integer is
returned. Otherwise, \code{f} recursively decrements the integer until it is negative
(i.e., until one reaches -1). The function \code{main} uses a 
non-deterministically chosen integer \code{i}. As long as this integer is positive, it is
copied to a new address \code{op}, and \code{f(op)} is added to the integer. Since
\code{f} always returns a negative number as its result, the \code{while}-loop of the
function \code{main} terminates. To ease readability, we use these two functions
as a minimal example which illustrates how our technique handles side effects and
explicit memory management in the context of recursion, and how it allows the re-use of
previous analyses. See \cref{sect:eval} for an evaluation of our approach
on more realistic (and more complex) functions.

\bigskip

\hspace*{-.6cm}
\begin{minipage}{\textwidth}
\begin{verbatim}
int f(int* p) {          int main() {
    if (*p < 0) {            int i = nondet_int(); 
       int pv = *p;          while (i > 0) {
       free(p);                int* op = malloc(sizeof(int));
       return pv; }            *op = i;
    (*p)--;                    i += f(op);
    return f(p);             }
}                        }
\end{verbatim}
\end{minipage}

\bigskip

Fig.\ \ref{fig:llvm-code-f} gives the \LLVM{} code corresponding\footnote{The  \LLVM{}
  code in Fig.\ \ref{fig:llvm-code-f} is equivalent to the 
        code produced by the \clang compiler~\cite{Clang}. 
	However, to simplify the presentation, we modified the \LLVM code by using
        \code{i8} instead of \code{i32} integers. 
	\aprove{} can also prove termination of the original \LLVM program that results
        from compiling our example \sfC{} program with \clang.} 
to the function \code{f}.
It consists of the \emph{basic blocks} \code{entry}, \code{rec}, and \code{term}.
We removed the leading \code{\%} from variable names and numbered the instructions in each
block to increase readability. 
The execution of \code{f} starts in the block \code{entry}.
The semantics of the  \LLVM{} code will be discussed  in
\cref{sect:seg} when we construct the SEG.

\begin{figure}[t]
  \sidecaption[t]
  \begin{minipage}{7.2cm}
\vspace*{-1.5cm}
    \begin{boxedminipage}{7.1cm}
      \begin{alltt}
        \begin{tabbing}
          ent\=\kill
        define i8 @f(i8* p) \{\\
        entry: \\
        \>0: \=pval = load i8* p\\
        \>1: \>ricmp = icmp slt i8 pval, 0\\
        \>2: \>br i1 ricmp, label term, label rec\\
         rec:\\
   \>0: \>dec = add i8 pval, -1\\
   \>1: \>store i8 dec, i8* p\\
   \>2: \>rrec = call i8 @f(i8* p)\\
   \>3: \>ret i8 rrec\\
    term:\\
    \>0: \>call void @free(i8* p)\\
    \>1: \>ret i8 pval \}
      \end{tabbing}
      \end{alltt}
\end{boxedminipage}
  \end{minipage}
\caption{\LLVM{} code for the function \code{f}}
\label{fig:llvm-code-f}       
\end{figure}

We now recapitulate the notion of \emph{abstract states} from \cite{LLVM-JAR}, which we use for symbolic
execution. Abstract states represent sets of \emph{concrete} states, i.e., 
of configurations during an actual execution of the program.
In these abstract states, the values of the program variables are represented by \emph{symbolic variables} instead
of concrete integers. 
In our abstract domain, a state consists of a call stack $\CS$, a knowledge base $\KB$ with
information about the symbolic variables,
a set $\AL$ describing memory allocations
by \code{malloc}, and a set $\PT$ describing the content of the heap. 
A call stack $\CS = [\FR_1, \dotsc, \FR_n]$ consists of $n$ stack frames $\FR_i$, where
$\FR_1$ is the topmost
and $\FR_2, \dotsc, \FR_n$ are the \emph{lower} stack frames.
We use ``$\cdot$'' to decompose call stacks,
i.e., $[\FR_1, \ldots, \FR_n] = \FR_1 \cdot [\FR_2, \ldots, \FR_n]$.
Given a state
$s$ with call stack $\CS$, its \emph{size} is defined as  $|s| = n$.
The first component of a stack frame $\FR_i$ is a \emph{program position} (\codevar{b},
$k$), indicating that instruction $k$ of block \codevar{b}
  is to be executed next.
To ease the formalization, we assume that different functions do not have basic blocks
with the same names.
Let $\PPos = (\Blks \times \N)$ be the set of all program positions, where  $\Blks$
is  the set of all basic blocks.
As the second component,
each stack frame $\FR_i$ has a partial injective function $\LVi: \Ids
\partialfunctionmap \Vsym$, where ``$\partialfunctionmap$'' indicates partial functions.
Each function $\LVi$ maps \emph{local program variables} $\Ids$
(e.g., $\Ids = \{\code{p}, \code{pval} \dots\}$)
to symbolic variables from an infinite
set $\Vsym$ with  $\Vsym \cap \Ids = \emptyset$.
We require all $\LVi$ in a state to have pairwise disjoint ranges.
\label{def:cs-as-formula}  
We  often extend $\LVi$ to a function from $\Ids \uplus \Z$ to $\Vsym 
\uplus \Z$ by defining $\LVi(n) = n$ for all $n \in \Z$.   
Moreover, we identify  $\CS$ with the set of equations
$\bigcup_{i=1}^n \{ \code{x}_i = \LV_i(\code{x}) \mid 
\code{x} \in \domain(\LVi)\}$, where 
$\domain(\LVi)$ denotes the set of all program variables $\code{x} \in \Ids$ where $\LVi(\code{x})$ is defined.
As a third and last component, each stack frame $\FR_i$ has a set $\AL_i$ of allocations.
It consists of expressions of the form $\alloc{v_1}{v_2}$ for $v_1, v_2 \in \Vsym$,
which indicate that $v_1 \le v_2$ and that all addresses between $v_1$ and $v_2$ have
been allocated by \code{alloca} in the $i$th stack frame.

While the call stack $\CS$ is the first component of an \LLVM\ state,
the second component is a \emph{\underline{k}nowledge
  \underline{b}ase} 	$\KB \subseteq \QFIA(\Vsym)$ of
\underline{q}uantifier-\underline{f}ree first-order formulas that express
\underline{i}nteger \underline{a}rithmetic properties of $\Vsym$. 
For concrete states, the knowledge base constrains the state's symbolic variables such
that their values are uniquely determined, whereas for abstract states several values are
possible. We identify \emph{sets} of first-order formulas $\{\varphi_1, \ldots, \varphi_m\}$ with their
conjunction $\varphi_1 \wedge \ldots \wedge \varphi_m$. 

The third component of a state
is the allocation list $\AL$.
It consists  of  expressions of the form $\alloc{v_1}{v_2}$ for $v_1, v_2 \in \Vsym$,
which mean that $v_1 \le v_2$ and that all addresses between $v_1$ and $v_2$ have
been allocated by \code{malloc}.
In contrast to \code{alloca}, such allocated memory needs to be released explicitly by the programmer.
Let $\AL^*(s) := \bigcup_{i=1}^n \AL_i \cup AL$ denote the set of all allocations of a
state $s$.
We require any two entries $\alloc{v_1}{v_2}$ and $\alloc{w_1}{w_2}$ from $\AL^*(s)$ with
$(v_1,v_2) \neq (w_1,w_2)$ to be disjoint.

The fourth  component $\PT$ is a set of ``\underline{p}oints-\underline{t}o'' atoms $v_1
\pointsto[\codevar{ty}] v_2$ where $v_1,v_2 \in \Vsym$ and  $\codevar{ty}$ is an \LLVM{}
type. This means that the value $v_2$ of type $\codevar{ty}$ is stored at the address $v_1$.
For example, as each memory cell stores one byte,
$v_1 \pointsto[\codevar{i32}] v_2$ states that $v_2$ 
is stored in the four cells  $v_1, \ldots, v_1+3$. 

Finally, we use a special state $\ERROR$ to be reached if we cannot prove absence of
undefined behavior (e.g., if a violation of memory safety by accessing non-allocated
memory might take place). 
\begin{definition}[States]\label{LLVM states}
\LLVM{}  \emph{states} have the form $(\CS, \KB, \AL, \PT)$  where 
  \begin{itemize}
	\item $\CS \in (\PPos \times (\Ids \partialfunctionmap \Vsym) \times \{\alloc{v_1}{v_2} \: | \: v_1, v_2 \in \Vsym\})^\ast$, 
	\item $\KB \subseteq \QFIA(\Vsym)$,
	\item $\AL \subseteq \{\alloc{v_1}{v_2} \: | \: v_1, v_2 \in \Vsym\}$, and
	\item $\PT \subseteq \{ (v_1 \pointsto[\codevar{ty}] v_2) \: | \: v_1, v_2 \in \Vsym,
\text{\codevar{ty} is an \LLVM{} type} \}$. 
  \end{itemize}
In addition, there is a state $\ERROR$ for undefined behavior.
For any state $s$, let $\Vsym(s)$ consist of  all symbolic variables  occurring in $s$.
\end{definition}

As an example, we consider the following  state $A$:
\label{pageref:state-A-in-text}
$$([((\code{entry},0), \{\code{p}_1 = v_\codevar{p}\}, \emptyset)],\emptyset,
\{\alloc{v_\codevar{p}}{v_\codevar{p}}\}, \{ v_\codevar{p} \pointsto[\codevar{i8}]  v_\codevar{*p}\} ) $$ 
It represents concrete states at the beginning of \code{f}'s \code{entry} block, where the value of
the program variable \code{p} in the first and only stack frame is represented by the
symbolic variable $v_\codevar{p}$.
There is an allocation $\alloc{v_\codevar{p}}{v_\codevar{p}}$, consisting of only a single byte,
where the value $v_\codevar{*p}$ is stored. As the knowledge base is empty, we have no
further knowledge about $v_\codevar{*p}$.  We often refer to the components of states by
using superscripts, e.g.,  $\AL^{s}$ refers to the allocation list of a state $s$.

In order to construct the symbolic execution graph, for any state $s$ we define a
first-order formula $\stateformula{s}$, which contains $\KB$ and expresses relations
resulting from the entries in $\AL$ and $\PT$. By
representing states with first-order formulas, we can use standard SMT solving for all
reasoning required in our approach. We also use the first-order formulas
$\stateformula{s}$ for the subsequent  generation of integer transition systems from
symbolic execution graphs. 

\begin{definition}[Representing States by $\FOL$ Formulas]
\label{def:StateFOLFormula} Given a state $s = (\CS, \KB,\allowbreak \AL,\allowbreak \PT)$, the set
$\stateformula{s}$ is the smallest set with
\[
\mbox{ $\begin{array}{rcl}
\stateformula{s} &=& \KB \; \cup \;
\{1 \leq v_1 
\wedge v_1 \leq v_2 \mid \alloc{v_1}{v_2} \in \AL^*(s)\} \; \cup\\
&& \{ v_2 < w_1 \vee w_2 < v_1 \mid
   \alloc{v_1}{v_2},\alloc{w_1}{w_2} \in \AL^*(s), \; (v_1,v_2) \neq (w_1,w_2)\} \; 
   \cup\\
&& \{1 \leq v_1 
   \mid (v_1 \hookrightarrow_{\codevar{ty}} v_2) \in \PT \} \; \cup \\
&& \{v_2 = w_2 \mid (v_1 \hookrightarrow_{\codevar{ty}} v_2), 
   (w_1 \hookrightarrow_{\codevar{ty}} w_2) \in \PT \mbox{ and } \models \, 
   \stateformula{s} \implies v_1 = w_1\} \; \cup\\
&& \{v_1
  \neq w_1 \mid (v_1 \hookrightarrow_{\codevar{ty}} v_2),  
   (w_1 \hookrightarrow_{\codevar{ty}} w_2) \in \PT \mbox{ and } \models \, 
  \stateformula{s} \implies v_2 \neq w_2\}.   
\end{array}$}
\]
\end{definition}

We now formally introduce \emph{concrete} states as states of a particular form.
They determine the values of variables and the contents of the memory \emph{uniquely}.
To enforce a uniform representation, in concrete states we only allow statements of the
form $w_1 \pointsto[\codevar{i8}] w_2$ in $\PT$.
So here we represent memory data byte-wise,
and since $\LLVM{}$ represents values in two's complement, each byte stores a
value from $[-2^7,2^7-1]$.
Moreover, since concrete states represent
actual executions of programs on a machine, we require that their  
set $\PT$   only contains information about addresses that are known to be allocated.

\begin{definition}[Concrete States]
\label{def:concrete_state}
An \LLVM{} state $c$ is \emph{concrete} iff $c = \ERROR$ or $c  =
(\CS,\KB,\AL,\PT)$
such that the following holds:
\begin{itemize}
\item[$\bullet$] $\stateformula{c}$ is satisfiable
\item[$\bullet$]
for all $v \in \Vsym(c)$  there exists an $n \in \Z$ such that $\models  \stateformula{c} \implies v = n$
\item[$\bullet$]
there is no $(w_1 \hookrightarrow_{\codevar{ty}} w_2) \in \PT$ for $\codevar{ty} \neq
\code{i8}$,
\item[$\bullet$]
for all $\alloc{v_1}{v_2} \in \AL^*(c)$ and for all integers
$n$ with
$\models \stateformula{c}  \implies v_1 \leq n \land n \leq v_2$,
there exists
$(w_1 \hookrightarrow_{\codevar{i8}} w_2) \in \PT$ for some $w_1,w_2 \in \Vsym$
such that $\models\stateformula{c} \implies  w_1 = n$ and
 $\models\stateformula{c} \implies w_2 = k$ for some $k \in [-2^7,2^7-1]$
\item[$\bullet$] for every  $(w_1 \pointsto[\codevar{i8}] w_2) \in \PT$, there is a $\alloc{v_1}{v_2} \in
  \AL^*$ such that
    $\models \stateformula{c} \implies v_1 \leq w_1 \leq v_2$.
\end{itemize}
\end{definition}

In \cite{LLVM-JAR}, for every abstract state $s$, we also introduced a \emph{separation logic}
formula $\stateformula{s}_{\SL}$ which extends $\stateformula{s}$ by further
information about the memory.  The semantics of these formulas are defined using
\emph{interpretations} $(\slassignment,\slmemory)$. 
The function $\slassignment$ assigns integer values to the program
variables. 
The function $\slmemory$ describes the
memory contents at allocated addresses.
We recapitulate $\stateformula{s}_{\SL}$, formal definitions of $\slassignment$ and
$\slmemory$, and the semantics
of separation logic in App.\ \ref{app:semantics-sl}.
For any abstract state $s$ we have $\models \stateformula{s}_\SL \implies
\stateformula{s}$, i.e., $\stateformula{s}$ is a weakened version of $\stateformula{s}_\SL$. 
As mentioned, we use $\stateformula{s}$ for the construction of the symbolic execution
graph, enabling standard first-order SMT solving to be used for all reasoning required in
this construction. 

Finally, we recapitulate which concrete states $c \neq \ERROR$ are represented by an abstract state
$s$ according to \cite{LLVM-JAR}. 
\label{def-matching-stacks}Here, we require that the stacks of $c$ and $s$ have the same size, i.e., $|c| = |s|$,
and at each stack index $1 \leq i \leq |s|$ we have $\FR_i^c = (p_i, \LV_i^c, \AL_i^c)$ and
$\FR_i^s = (p_i,\LVi^s,\AL_i^c)$ with $\domain(\LVi^c) =  \domain(\LVi^s)$.
In the next section, we will present a variant of Def.\ \ref{def:representation-same-stack-size} for states of different
stack sizes.

In order to define the representation relation between states with stacks of the same size, we
extract an interpretation $(\slassignment^c,\slmemory^c)$ from concrete states $c$.
Furthermore, we use \emph{concrete instantiations} $\sigma : \Vsym \to \Z$ which map
symbolic variables to integers. An abstract
state $s$ then \emph{represents} a concrete state $c$ if there exists a concrete
instantiation $\sigma$ such that  $(\slassignment^c,\slmemory^c)$ is a model of
$\sigma(\stateformula{s}_\SL)$ and if for each allocation of $s$ there exists a corresponding allocation
in $c$ of the same size.
Here, we extend the concrete instantiation $\sigma$ to
formulas as usual, i.e., $\sigma(\varphi)$ instantiates all free occurrences of $v \in
\Vsym$ in $\varphi$ by $\sigma(v)$.

\begin{definition}[Representing Concrete by Abstract States]
\label{def:representation-same-stack-size}
Let $c = ([(p_{1}, \LV_{1}^c,\allowbreak \AL_{1}^c),\allowbreak \ldots, (p_n, \LV_{n}^c, \AL_{n}^c)],
\KB^c, \AL_{0}^c,\allowbreak \PT^c)$ be a concrete state.
We say that $c$ is \emph{represented} by 
a state $s = ([(p_{1}, \LV_{1}^s, \AL_{1}^s), \ldots, (p_{n},
  \LV_{n}^s, \AL_{n}^s)],\allowbreak \KB^s,\allowbreak \AL_{0}^s,\allowbreak \PT^s)$
iff
\begin{enumerate}
    \item $\domain(\LVi^c) =  \domain(\LVi^s)$ for all $1 \leq i \leq n$,
    \item $(\slassignment^c,\slmemory^c)$ is a \emph{model} of
        $\sigma(\stateformula{s}_\SL)$ for some concrete instantiation $\sigma : \Vsym \to
      \Z$, and
    \item for all $\alloc{v_1}{v_2} \in \AL_{i}^s$ with $0 \leq i \leq n$, there exists $\alloc{w_1}{w_2} \in
	\AL_{i}^c$ such that $\models \stateformula{c} \Rightarrow w_1 = \sigma(v_1) 
	\wedge w_2 = \sigma(v_2)$.\footnote{Note that this condition is new as compared to \cite{LLVM-JAR}.
	However, this additional condition is needed in order to achieve soundness.
	The reason is that if $s$ contains an allocation in stack frame $i$ and $c$ contains
	the corresponding allocation in stack frame $j$ with $j < i$, then after returning
	from stack frame $j$, there would be an allocation in a successor state $\overline{s}$ of $s$ that is
	not represented in the corresponding successor $\overline{c}$ of $c$. Therefore,
        $\overline{c}$ would not be
	represented by $\overline{s}$, which would violate the soundness of our approach.
	}
\end{enumerate}
The error state $\ERROR$ is only represented by $\ERROR$ itself.
\end{definition}

%% file: graphconstruction.tex
\section{Construction of Symbolic Execution Graphs}\label{sect:seg}

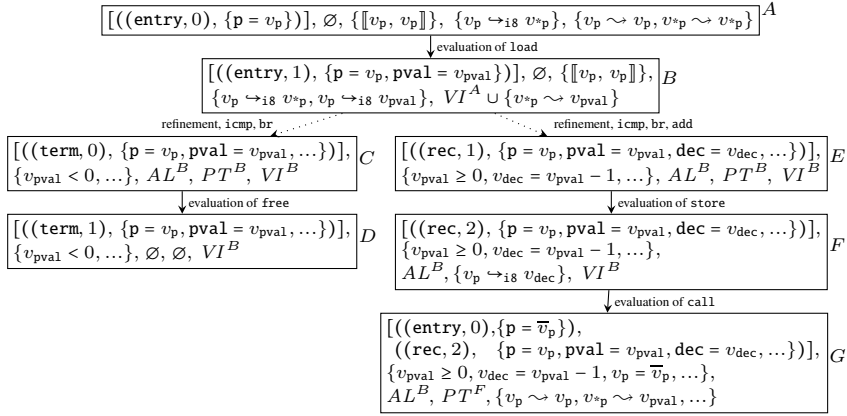
\begin{figure}[t]
\center
\input{graph_f_initial}
\caption{\label{fig:initial-graph-function-f}Initial states of the symbolic execution graph of function \code{f} }
\end{figure}

In
\cref{fig:initial-graph-function-f},
we start constructing the symbolic execution graph for the function \code{f} from
\cref{fig:llvm-code-f}, independently of \code{main}.\footnote{In principle one could
  analyze some functions of the program
  in a modular way and use our previous non-modular approach from \cite{LLVM-JAR} for other functions. However, to
  ease the presentation, in this paper we assume that our new modular treatment is used for all
  functions. In our implementation in \aprove, we indeed apply our new modular approach
  for all functions except those that only consist of straightline code, i.e., that do not
  have any branching.}  
Here, we omit the index of the program variables in
stack frames, i.e., we write ``$\code{p} = v_{\codevar{p}}$'' instead of ``$\code{p}_1 =
v_{\codevar{p}}$''.
Moreover, to ease readability, some parts of the states are abbreviated by ``$\dots$'',
and allocations in the individual stack frames are omitted since they are empty throughout
this graph.
The last state component $\CV$ will be introduced later and can be ignored for
now.
The initial state for our analysis is $\GraphFNodeOne{}$, which we already considered
after Def.\ \ref{LLVM states}.
It is at the first program position in \code{f}. Therefore the next instruction
\code{load}s the value stored at \code{p} to \code{pval}. 
We re-use the symbolic execution rules from \cite{LLVM-JAR} for all steps not 
involving function calls.
As an example, we briefly recapitulate the \code{load} rule to give an
idea of the general graph construction.
For the formal definition of the remaining rules, we refer to \cite{LLVM-JAR}.

The following rule is used to symbolically evaluate a state $s$ to
a state $\overline{s}$ by
loading the value of type
\code{ty} stored at some address \code{ad} into the variable \code{x}.
For any type \codevar{ty}, let $\sizeOf(\codevar{ty})$ denote the size of \codevar{ty} in
bytes. For example, $\sizeOf(\codevar{i32}) = 4$.
As each memory cell stores one byte, we first have to check whether the 
addresses $\texttt{ad}, \ldots, \code{ad} + \mathit{size}(\code{ty})-1$ are 
allocated, i.e., whether there is a $\alloc{v_1}{v_2} \in \ALL$
such that $\stateformula{s} \Rightarrow (v_1 \leq \LV_1(\code{ad}) \; \wedge \; 
\LV_1(\code{ad})+\mathit{size}(\code{ty})-1 \leq v_2)$ is valid.
Then, we reach a new state where the previous position 
$\Pos = (\code{b}, k)$ is updated to the position $\Pos^+ = (\code{b}, k+1)$ of 
the next instruction in the same basic block, and we set $\LV_1(\code{x}) = w$ 
for a fresh $w \in \Vsym$.
Here we write $\updatefct{\LV_1}{\code{x}}{w}$ for the function
where $(\updatefct{\LV_1}{\code{x}}{w})(\code{x}) = w$ and for $\code{y}
\neq \code{x}$, we have $(\updatefct{\LV_1}{\code{x}}{w})(\code{y}) = \LV_1(\code{y})$.
Moreover, we add $\LV_1(\code{ad}) \hookrightarrow_{\codevar{ty}} w$ to $\MT$. 
Thus, if $\MT$ already contained a formula $\LV_1(\code{ad}) \hookrightarrow_{\codevar{ty}} u$,
then $\stateformula{s}$ implies $w = u$.

\begin{center}
\fbox{ 
\begin{minipage}{11.3cm}
\small
\mbox{\small \textbf{\hspace*{-.15cm}\texttt{load} from allocated memory 
($\Pos:$ ``\texttt{x = load ty* ad}'' with $\code{x},\texttt{ad}\in\Ids$)}}\\

\centerline{$\frac{\parbox{9.2cm}{\centerline{
$s = ((\Pos, \; \LV_1, \; \AL_1) \cdot \CS, \; \KB, \; \AL, \; \MT)$}\vspace*{.1cm}}}
{\parbox{9.2cm}{\vspace*{.1cm}
$\overline{s} = ((\Pos^+\!\!, \; \updatefct{\LV_1}{\code{x}}{w}, \; \AL_1) \cdot \CS, \;\KB, \; \AL, \; \MT \cup 
\{\LV_1(\code{ad}) \hookrightarrow_{\codevar{ty}} w \})$}}$ \quad \mbox{if}}

\vspace*{.2cm}

\hspace*{-.3cm}\begin{minipage}{14cm}
\begin{itemize}
\item[$\bullet$] \mbox{\small there is $\alloc{v_1}{\!v_2}\!\in\!\AL^*$ with 
                 $\models \stateformula{s} \Rightarrow 
                 (v_1 \leq \LV_1(\code{ad}) \, \wedge \, 
                 \LV_1(\code{ad}) + \mathit{size}(\code{ty}) -1 \leq v_2)$,}
\item[$\bullet$] \mbox{\small $w \in \Vsym$ is fresh}
\end{itemize}
\end{minipage}
\end{minipage}} 
\end{center}

State $\GraphFNodeTwo$ arises from applying this rule,
i.e., from evaluating the \code{load} instruction and thus, there is an \emph{evaluation
  edge} from
$\GraphFNodeOne$ to $\GraphFNodeTwo$. In $\GraphFNodeTwo$,
a new variable $v_\codevar{pval}$
is introduced for the value of the program variable $\code{pval}$.
If we could not prove memory safety of the operation, we would create an edge to $\ERROR$ instead.
The new  entry $(v_\codevar{p} \pointsto[\codevar{i8}]  v_\codevar{pval})$ in
$\PT^{\GraphFNodeTwo}$ denotes that $v_\codevar{pval}$ is the value at the address $v_\codevar{p}$. 
Thus, we have $(v_\codevar{*p} = v_\codevar{pval}) \in \stateformula{B}$.

The next instruction sets the variable \code{ricmp} to the result of an integer comparison
(\code{icmp}),
 based on whether \code{pval} is negative or not (i.e., \code{slt} stands for
``\underline{s}igned \underline{l}ess \underline{t}han'').
The instruction cannot be evaluated directly as there is no knowledge about the value of
$v_\codevar{pval}$ in $\KB^{\GraphFNodeTwo}$. 
Therefore, we perform a case analysis by creating outgoing \emph{refinement edges} to two
successors of the state ${\GraphFNodeTwo}$ where the knowledge base is extended by
$v_\codevar{pval} < 0$ and $v_\codevar{pval} \ge 0$, respectively.
For the sake of brevity we directly evaluate some subsequent instructions in both branches
and omit the intermediate states in \cref{fig:initial-graph-function-f}.

In the case with $v_\codevar{pval} < 0$, this yields the state $\GraphFNodeSix$ after the
execution of \code{icmp} and the \code{br} instruction, which \underline{br}anches to the
block \code{term}.
Analogously, for $v_\codevar{pval} \ge 0$, this yields the state $\GraphFNodeTen$ after the
execution of the \code{icmp}, \code{br}, and \code{add}  instructions. 

State $\GraphFNodeSix$ is at the \code{call} of the \code{free} instruction in the block
\code{term}, corresponding to the base case of the recursive function \code{f}.
Evaluation
of the \code{free} instruction
yields $\GraphFNodeSeven$, where the entries for
the pointer \code{p} have been removed from $\AL$ and $\PT$.
We refer to states like $\GraphFNodeSeven$, whose only stack frame is at a
\underline{ret}urn instruction of a function $\rfunc$, as \emph{return
  states} of $\rfunc$.

In State $\GraphFNodeTen$, 
one has to \code{store} the value of \code{dec} at the address
\code{p}, where 
$v_{\codevar{dec}} = v_\codevar{pval} - 1$ holds due to the previous \code{add} instruction.
Thus, in the resulting\footnote{%
	The symbolic execution rule for \code{store} in \cite{LLVM-JAR} always creates
        a fresh variable and an equality constraint for the value to be stored. 
	When storing a program variable instead of a numerical literal (i.e., a number), one can however
        re-use the existing symbolic variable  without influencing the analysis further,
        which we did here to ease readability.
	}
state $\GraphFNodeEleven$,
the new value at \code{p} is denoted by $(v_\codevar{p} \pointsto[\codevar{i8}] v_\codevar{dec}) \in \PT^{\GraphFNodeEleven}$.
Evaluation of the \code{call} instruction 
in $\GraphFNodeEleven$ yields
$\GraphFNodeTwelve$, whose topmost stack frame is at the beginning of the recursive
execution of \code{f}. 

In the remainder of the section, we
present our new modular approach for symbolic execution. To this end, we
first show in \cref{Abstracting the Call Stack}
how to abstract the call stack in order to obtain
a separate finite SEG for every (possibly recursive) function.
In \cref{Intersecting States} we explain how to continue the symbolic
execution after returning from a function call.
\cref{Complete SEGs} discusses how to obtain finite complete SEGs for every function.
Finally, \cref{Re-Use} shows how SEGs of
(possibly recursive) auxiliary functions can be re-used in a modular way.

\input{abstractingcallstack}

\input{intersecting}
\input{reusingfunctiongraphs}

%% file: graph_f_initial.tex
\footnotesize
\newcommand{\ydist}{0.3cm}
\newcommand{\xdist}{0.5cm}
\newcommand{\setfont}{\scriptsize}
\newcommand{\labelxshift}{-2pt}

\begin{tikzpicture}[node distance = \ydist and \xdist]
\scriptsize
\def\widetwidth{6.2cm}
\def\smalltwidth{4.75cm}
\def\fulllinewidth{11cm}
\def\edgenodedist{0.09cm}
\def\FirstIndentwidth{3cm}
\def\SecondIndentwidth{2cm}
\tikzstyle{state}=[
                           inner sep=2pt,
                           font=\scriptsize,
                           draw]

\node[state, label={[xshift=\labelxshift]2:$\GraphFNodeOne$}] (1)
{$[((\code{entry}, 0), \, \{\code{p} = v_{\codevar{p}}\})],
       \, \emptyset, \,
    \{ \alloc{v_{\codevar{p}}}{v_{\codevar{p}}} \}, \,
    \, \{ v_{\codevar{p}} \pointsto[\codevar{i8}] v_{\codevar{*p}} \},
    \,  \{v_{\codevar{p}} \corrvar v_{\codevar{p}}, v_{\codevar{*p}} \corrvar v_{\codevar{*p}}\}$
};

\node[state, below=of 1, label={[xshift=\labelxshift]2:$\GraphFNodeTwo$}, align=left] (2)
{$[((\code{entry}, 1), \, \{\code{p} = v_{\codevar{p}}, \code{pval} = v_{\codevar{pval}}\})],
\, \emptyset, \,
    \{ \alloc{v_{\codevar{p}}}{v_{\codevar{p}}} \},$\\$
    \, \{ v_{\codevar{p}} \pointsto[\codevar{i8}] v_{\codevar{*p}},
    v_{\codevar{p}} \pointsto[\codevar{i8}] v_{\codevar{pval}} \},
    \, \CV^\GraphFNodeOne \cup \{v_{\codevar{*p}} \corrvar v_{\codevar{pval}}\}$
};

\path let \p1 = (\widetwidth,0), \p2 =  (\smalltwidth,0), \p3 = (\xdist,0), \p4 = ($0.5*(\p1)-0.5*(\p2)+0.5*(\p3)$)  in coordinate (middle) at ($(2.south)-(\p4)$);

\node[state,label={[xshift=\labelxshift]2:$\GraphFNodeSix$},anchor=north east, text
width=4.9cm,yshift=-\ydist] (6) at (2.south -| middle)
{   $[((\code{term}, 0), \, \{\code{p} = v_{\codevar{p}}, \code{pval} = v_{\codevar{pval}}, ...\})],$\\
      $\{ v_\codevar{pval} < 0, ...\},$
    $\AL^{\GraphFNodeTwo}\!,$
    $\PT^{\GraphFNodeTwo}\!,$
      $\CV^\GraphFNodeTwo$
};

\node[state,label={[xshift=\labelxshift]2:$\GraphFNodeSeven$},below =of 6, text width=4.9cm] (7)
{   $[((\code{term}, 1), \, \{\code{p} = v_{\codevar{p}}, \code{pval} = v_{\codevar{pval}}, ...\})],$\\
       $\{ v_\codevar{pval} < 0, ...\},$
    $\emptyset, \,
    \emptyset,$
     $\CV^\GraphFNodeTwo$
};

\node[state,label={[xshift=\labelxshift]2:$\GraphFNodeTen$}, right =.5cm of 6.north east, anchor=north west,text width=6.2cm] (10)
{   $[((\code{rec}, 1), \, \{\code{p} = v_{\codevar{p}}, \code{pval} = v_{\codevar{pval}}, \code{dec} = v_\codevar{dec}, ...\})], $\\$
   \{ v_\codevar{pval} \ge 0, v_\codevar{dec} = v_\codevar{pval} - 1,  ...\}, \,
    \AL^{\GraphFNodeTwo}\!, \,
    \PT^{\GraphFNodeTwo}\!,
  \,\CV^\GraphFNodeTwo$
};

\node[state,label={[xshift=\labelxshift]2:$\GraphFNodeEleven$},below =of 10, text width=6.2cm] (11)
{   $[((\code{rec}, 2), \, \{\code{p} = v_{\codevar{p}}, \code{pval} = v_{\codevar{pval}}, \code{dec} = v_\codevar{dec}, ...\})],$\\$
     \{v_\codevar{pval} \ge 0, v_\codevar{dec} = v_\codevar{pval} - 1,  ...\},$\\$
    \AL^{\GraphFNodeTwo}\!,
\{ v_\codevar{p} \pointsto[\codevar{i8}] v_{\codevar{dec}} \},\,
     \CV^\GraphFNodeTwo$
};

\node[state,label={[xshift=\labelxshift]2:$\GraphFNodeTwelve$},below left=.3cm and -6.45cm of 11, text width=6.4cm] (12)
{\begin{tabbing}$[$\=$((\code{entry}, 0),$ \=$\{ \code{p}
= \overline{v}_{\codevar{p}} \}),$\\
\>$((\code{rec}, 2),$\>$\{\code{p} = v_{\codevar{p}}, \code{pval} =
v_{\codevar{pval}}, \code{dec} = v_\codevar{dec}, ...\})],$\\
$\{v_\codevar{pval} \ge 0, v_\codevar{dec} = v_\codevar{pval} - 1, v_\codevar{p} = \overline{v}_{\codevar{p}},
	... \},$\\$
    \AL^{\GraphFNodeTwo}\!, \,
    \PT^{\GraphFNodeEleven}\!,
\{v_{\codevar{p}} \corrvar v_{\codevar{p}}, v_{\codevar{*p}} \corrvar v_{\codevar{pval}}, ...\}$
\end{tabbing}
};

\draw[eval-edge] (1) -- node[anchor=west] {\tiny evaluation of \code{load}} (2);
\draw[omit-edge] (2) --  node[anchor=east,xshift=-10pt] {\tiny  refinement, \code{icmp}, \code{br}}(6);
\draw[omit-edge] (2) --  node[anchor=west,xshift=10pt] {\tiny  refinement, \code{icmp}, \code{br}, \code{add}} (10);
\draw[eval-edge] (6) -- node[anchor=west] {\tiny evaluation of  \code{free}} (7);
\draw[eval-edge] (10) -- node[anchor=west] {\tiny evaluation of \code{store}} (11);
\draw[eval-edge] (11) -- node[anchor=west] {\tiny evaluation of \code{call}} (12);

\end{tikzpicture}

%% file: abstractingcallstack.tex
\subsection{Abstracting the Call Stack}\label{Abstracting the Call Stack}

\begin{figure}[t]
\center
\input{graph_f}
\vspace{-0.1cm}
\caption{\label{fig:full-graph-function-f}Symbolic execution graph of function \code{f},
  with states $\GraphFNodeOne{}$ to $\GraphFNodeTwelve{}$ as in \cref{fig:initial-graph-function-f}}
\end{figure}
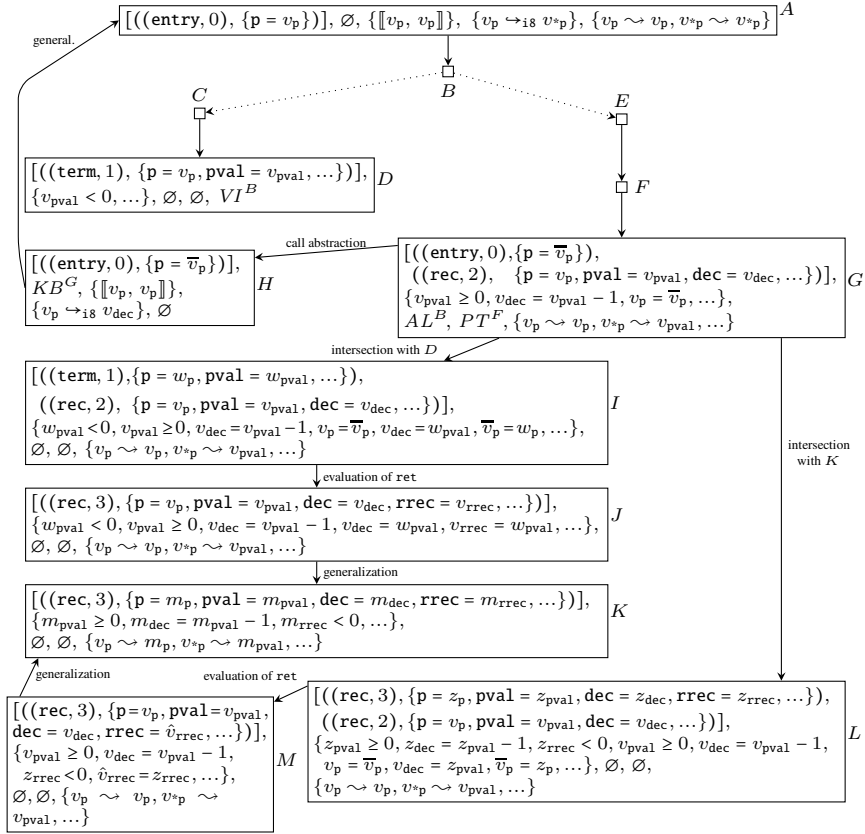

\cref{fig:full-graph-function-f}
continues
the construction of the SEG for the function \code{f} from
\cref{fig:initial-graph-function-f}.
So its states $\GraphFNodeOne{}$ to $\GraphFNodeTwelve{}$ are the same ones as in Fig.~\ref{fig:initial-graph-function-f}.
In particular, $\GraphFNodeTwelve{}$ corresponds to the start of the execution of the function \code{f} after the recursive call.

Any abstract state $s$ with $|s| > 1$ whose topmost stack frame is at the initial program
position of a function $\rfunc$ is a \emph{call state} of $\rfunc$. 
Note that our SEG already depicts the execution of the function \code{f}, starting in $\GraphFNodeOne$.
To re-use an already existing analysis of a function, we use
\emph{context abstractions}, where lower stack frames of a state are removed.

\begin{definition}[Context Abstraction and Call Abstraction]
\label{def:context-abstraction}
	Let $s = ([(p_1,\allowbreak\LV_1,\allowbreak\AL_1),\allowbreak\dotsc,(p_n,\LV_n,\AL_n)],\KB,\AL,\PT )$ be a state.
	Then for any $1 \leq k  \leq n$, the state $\conttop{s} = ([(p_1,\LV_1,\AL_1),
          \dotsc,(p_{k-1},\LV_{k-1},\AL_{k-1}),(p_k,\LV_k,\widehat{\AL_k})],\KB,\allowbreak\AL,\allowbreak \PT )$ is the \emph{context abstraction}
        of $s$ of size $k$, where $\widehat{\AL_k} = \bigcup_{i=k}^n \AL_i$.
    The \emph{call abstraction} of a state is its context abstraction of size 1.	 
\end{definition}

Note that the bottommost stack frame of the context abstraction contains
  the stack allocations of all removed frames. In this way, the information that these
  parts of the memory have been allocated is still available in the context abstraction.
These stack allocations will be re-assigned to their corresponding stack frames
at a
later stage of the graph construction (see \cref{Memory Information in the Intersection}).

We now extend Def.~\ref{def:representation-same-stack-size} about the representation of
concrete by abstract states, which was limited to states of same stack size.
An abstract state $s$ \emph{weakly represents} a concrete state $c$ if the $|s|$ topmost stack
frames of $c$ are represented by $s$, but $c$ may have further stack frames 
below.

\begin{definition}[Weakly Representing Concrete by Abstract States]
\label{def:weak-representation}
A concrete state  $c$ is \emph{weakly represented} by an abstract state  $s$, denoted $c
\repby s$,
iff $c = s = \ERROR$ holds or there exists a context abstraction $\conttop{c}$ of $c$ 
such 
that 
 $\conttop{c}$ is represented by  $s$ according to
Def.~\ref{def:representation-same-stack-size}.
\end{definition}

To re-use previous states  in the
symbolic execution graph that already analyzed the behavior of a function, 
each call state like $\GraphFNodeTwelve$, which results from  calling a function,
must have an outgoing \emph{call abstraction edge} to its call abstraction (i.e., to its context abstraction of size 1).  
In our example graph, this yields the call abstraction $\GraphFNodeThirteen$, whose only
stack frame is at the beginning of \code{f}.

Note that such a call abstraction step is ``sound'' w.r.t.\ the weak representation
relation $\repby$, since any concrete state that is weakly represented by
$\GraphFNodeTwelve$ is also weakly represented by $\GraphFNodeThirteen$. Indeed, whenever 
 $c \repby s$ holds for some
abstract state $s$ with $|s| > 1$ stack frames, we have $c \repby \conttop{s}$ for all
context abstractions of $s$ of size $1 \leq k \leq |s|$.

The call stacks of $\GraphFNodeThirteen$ and $\GraphFNodeOne$ have the same size and every
concrete state represented by $\GraphFNodeThirteen$ is also represented by
$\GraphFNodeOne$, i.e., $\GraphFNodeOne$ ``\emph{covers}'' $\GraphFNodeThirteen$.
Thus, $\GraphFNodeOne$ is a \emph{generalization} of $\GraphFNodeThirteen$.
Formally, we use the following rule from \cite{LLVM-JAR} 
to determine when to create a \emph{generalization edge} from some abstract state $s$ to
its generalization $\overline{s}$.
It ensures that whenever a concrete state is represented by $s$, then it
is also represented by $\overline{s}$.

\vspace*{.3cm}
\noindent
\fbox{
\begin{minipage}{11.3cm}
\label{rule:generalization}
\small 
\mbox{\small \textbf{\hspace*{-.15cm}generalization with instantiation $\mu$}}\\ 
\vspace*{1mm}\\
\centerline{$\frac{\parbox{8.8cm}{\centerline{
$s = ([(\Pos_1, \; \LV_1, \; \AL_1),
   \ldots,
   (\Pos_n, \; \LV_n, \; \AL_n) ],
\; \KB, \; \AL_0, \; \PT)$}\vspace*{.1cm}}}
{\parbox{6.2cm}{\vspace*{.1cm} \centerline{
$\overline{s} = ([(\Pos_1, \; \overline{\LV}_1, \; \overline{\AL_1}),
   \ldots,
   (\Pos_n, \; \overline{\LV}_n), \; \overline{\AL}_n) ],
\; \overline{\KB}, \; \overline{\AL_0}, \;  \overline{\PT})$}}}\;\;\;\;$
\mbox{if}} \vspace*{-.05cm}
{\small 
\begin{itemize}
\item[(a)] $s$ has no incoming refinement or generalization edge
\item[(b)] $\domain(\LV_i) = \domain(\overline{\LV}_i)$  and
$\LV_i(\code{x}) = \mu(\overline{\LV}_i(\code{x}))$ for all 
           $1 \leq i \leq n$ and all $\code{x} \in \Ids$ where $\LV_i$ and
  $\overline{\LV}_i$ are defined
\item[(c)] $\models \stateformula{s} \implies \mu(\overline{\KB})$
\item[(d)] \label{bullet-point-generalization-allocations} if $\alloc{v_1}{v_2} \in \overline{\AL_i}$, then 
           $\alloc{w_1}{w_2} \in \AL_i$ with $\models \stateformula{s} \implies  w_1 =
  \mu(v_1) \wedge w_2 = \mu(v_2)$ for all $0 \leq i \leq n$  
\item[(e)] \label{bullet-point-generalization-points-to} if $(v_1 \hookrightarrow_{\codevar{ty}} v_2) \in \overline{\PT}$,\\then 
           $(w_1 \hookrightarrow_{\codevar{ty}} w_2) \in \PT$ with $\models
  \stateformula{s} \implies w_1 = \mu(v_1) \wedge w_2 = \mu(v_2)$  
\end{itemize}}
\end{minipage}}
\vspace*{.2cm}

The instantiation
  $\mu: \Vsym(\overline{s}) \functionmap \Vsym(s)$ maps variables from the more general state (e.g.,
$\GraphFNodeOne$) to the more specific state (e.g., $\GraphFNodeThirteen$).
In our example, we use an instantiation $\mu^\GraphFNodeThirteen$ such that $\mu^\GraphFNodeThirteen(v_\codevar{p}) = \overline{v}_{\codevar{p}}$
and $\mu^\GraphFNodeThirteen(v_\codevar{*p}) = v_{\codevar{dec}}$.
Condition (a) prevents cycles of refinement and generalization edges in the graph, which
would not correspond to an actual computation.
Compared to the 
corresponding generalization rule in \cite{LLVM-JAR}, we
slightly weakened the conditions (d) and
(e).
In \cite{LLVM-JAR}, conditions (d) and (e) are more strict w.r.t.\ the variables used.
For instance, condition (d) would require $\alloc{\mu(v_1)}{\mu(v_2)} \in \AL_i$ whereas our version
allows variables $w$ to be used that are provably equal to such variables $\mu(v)$.
This extends the applicability of the rules in many cases where equivalent variables occur.

Our construction of 
symbolic execution graphs ensures
that for any call state (like $\GraphFNodeTwelve$) which
denotes the start of the 
execution of a 
function, there exists a path from the call state to its call abstraction
which continues via a generalization edge to the \emph{entry state} of the function.
An entry state has a single stack frame that is at the initial program position of a
function
 and has no
outgoing generalization edge, i.e., $\GraphFNodeOne$ is the entry state of \code{f},
where
the function's symbolic execution starts.

%% file: graph_f.tex
\footnotesize
\newcommand{\ydist}{0.3cm}
\newcommand{\xdist}{0.5cm}
\newcommand{\setfont}{\scriptsize}
\newcommand{\labelxshift}{-2pt}

\begin{tikzpicture}[node distance = \ydist and \xdist]
\scriptsize
\def\widetwidth{6.2cm}
\def\smalltwidth{4.75cm}
\def\fulllinewidhth{11cm}
\def\edgenodedist{0.09cm}
\def\FirstIndentwidth{3cm}
\def\SecondIndentwidth{2cm}
\tikzset{invisible/.style={opacity=0}}
\tikzstyle{state}=[
                           inner sep=2pt,
                           font=\scriptsize,
                           draw]

\node[state, label={[xshift=\labelxshift]2:$\GraphFNodeOne$}] (1) 
{$[((\code{entry}, 0), \, \{\code{p} = v_{\codevar{p}}\})],
       \, \emptyset, \, 
    \{ \alloc{v_{\codevar{p}}}{v_{\codevar{p}}} \}, \,
    \, \{ v_{\codevar{p}} \pointsto[\codevar{i8}] v_{\codevar{*p}} \},
    \,  \{v_{\codevar{p}} \corrvar v_{\codevar{p}}, v_{\codevar{*p}} \corrvar v_{\codevar{*p}}\}$
};

\node[state, below=of 1, align=left, invisible] (2) 
{$[((\code{entry}, 1), \, \{\code{p} = v_{\codevar{p}}, \code{pval} = v_{\codevar{pval}}\})],
\, \emptyset, \, 
    \{ \alloc{v_{\codevar{p}}}{v_{\codevar{p}}} \},$\\$
    \, \{ v_{\codevar{p}} \pointsto[\codevar{i8}] v_{\codevar{*p}},
    v_{\codevar{p}} \pointsto[\codevar{i8}] v_{\codevar{pval}} \},
    \, \CV^\GraphFNodeOne \cup \{v_{\codevar{*p}} \corrvar v_{\codevar{pval}}\}$
};

\path let \p1 = (\widetwidth,0), \p2 =  (\smalltwidth,0), \p3 = (\xdist,0), \p4 = ($0.5*(\p1)-0.5*(\p2)+0.5*(\p3)$)  in coordinate (middle) at ($(2.south)-(\p4)$);

\node[state,anchor=north east, text
width=4.5cm,yshift=-\ydist, invisible] (6) at (2.south -| middle) 
{   $[((\code{term}, 0), \, \{\code{p} = v_{\codevar{p}}, \code{pval} = v_{\codevar{pval}}, ...\})],$\\
      $\{ v_\codevar{pval} < 0, ...\},$ 
    $\AL^{\GraphFNodeTwo}\!,$
    $\PT^{\GraphFNodeTwo}\!,$
      $\CV^\GraphFNodeTwo$
};

\node[state,label={[xshift=\labelxshift]2:$\GraphFNodeSeven$},below =of 6, text width=4.9cm, yshift=0.75cm,xshift=.5cm] (7) 
{   $[((\code{term}, 1), \, \{\code{p} = v_{\codevar{p}}, \code{pval} = v_{\codevar{pval}}, ...\})],$\\
       $\{ v_\codevar{pval} < 0, ...\},$ 
    $\emptyset, \,
    \emptyset,$
     $\CV^\GraphFNodeTwo$
};

\node[state, right =.5cm of 6.north east, anchor=north west,text width=5.6cm, invisible] (10) 
{   $[((\code{rec}, 1), \, \{\code{p} = v_{\codevar{p}}, \code{pval} = v_{\codevar{pval}}, \code{dec} = v_\codevar{dec}, ...\})], $\\$
   \{ v_\codevar{pval} \ge 0, v_\codevar{dec} = v_\codevar{pval} - 1,  ...\}, \, 
    \AL^{\GraphFNodeTwo}\!, \,
    \PT^{\GraphFNodeTwo}\!,
  \,\CV^\GraphFNodeTwo$
};

\node[state,below =of 10, text width=5.6cm, invisible] (11) 
{ $[((\code{rec}, 2), \, \{\code{p} = v_{\codevar{p}}, \code{pval} = v_{\codevar{pval}}, \code{dec} = v_\codevar{dec}, ...\})],$\\$
     \{v_\codevar{pval} \ge 0, v_\codevar{dec} = v_\codevar{pval} - 1,  ...\},$\\$
    \AL^{\GraphFNodeTwo}\!,
\{ v_\codevar{p} \pointsto[\codevar{i8}] v_{\codevar{dec}} \},\,
     \CV^\GraphFNodeTwo$
};

\node[state,label={[xshift=\labelxshift]2:$\GraphFNodeTwelve$},below left=.3cm and -5.76cm
of 11, text width=6.3cm, yshift=1cm] (12) 
{ \begin{tabbing} $[$\=$((\code{entry}, 0), $\=$\{ \code{p}
= \overline{v}_{\codevar{p}} \}),$\\
\>$((\code{rec}, 2),$\>$\{\code{p} = v_{\codevar{p}}, \code{pval} =
v_{\codevar{pval}}, \code{dec} = v_\codevar{dec}, ...\})],$\\
$\{v_\codevar{pval} \ge 0, v_\codevar{dec} = v_\codevar{pval} - 1, v_\codevar{p} = \overline{v}_{\codevar{p}},
	... \},$\\$
    \AL^{\GraphFNodeTwo}\!, \,
    \PT^{\GraphFNodeEleven}\!,
\{v_{\codevar{p}} \corrvar v_{\codevar{p}}, v_{\codevar{*p}} \corrvar v_{\codevar{pval}}, ...\}$
\end{tabbing}
};

\node[state, label=270:$\GraphFNodeTwo$, align=left, below=of 1, yshift=-0.1cm] (2-mini) 
{
};

\node[state, label=90:$\GraphFNodeSix$, align=left, yshift=0.3cm,xshift=.5cm] at (6.north) (6-mini) 
{
};

\node[state, label=90:$\GraphFNodeTen$, align=left, anchor=north] at (6-mini -| 12.north) (10-mini) 
{
};

\coordinate (11-mini-coord) at ($(10-mini.south)!0.5!(12.north)$);

\node[state, label=0:$\GraphFNodeEleven$, align=left, anchor=north] at (11-mini-coord) (11-mini) 
{
};

\node[state,label={[xshift=\labelxshift]2:$\GraphFNodeThirteen$}, left=1.9cm of 12.west,
	anchor=east, text width=2.9cm] (13) 
{   $[((\code{entry}, 0), \{ \code{p} = \overline{v}_{\codevar{p}} \})],\allowbreak
   \KB^{\GraphFNodeTwelve}\!,  \, 
    \{ \alloc{v_{\codevar{p}}}{v_{\codevar{p}}} \},$\\
$\{ v_\codevar{p} \pointsto[\codevar{i8}] v_{\codevar{dec}} \},   \, \emptyset$
};

\coordinate (centered-south-call-state) at (12.south -| 2.south);

\node[state,label={[xshift=\labelxshift]2:$\GraphFNodeFourteen$},below =of centered-south-call-state -|13.south west, anchor=north west, text width=9.4cm] (14) 
{ \vspace*{-.23cm} \begin{tabbing} $[$\=$((\code{term}, 1),$ \=$\{ \code{p} = w_\codevar{p},
	\code{pval} = w_{\codevar{pval}},
	  ... \}),$\\
	\>$((\code{rec}, 2),$\>$\{\code{p} = v_\codevar{p},
	\code{pval} = v_{\codevar{pval}},
	\code{dec} = v_\codevar{dec}, ...\})],$\\
      $\{ w_\codevar{pval}\!<\!0 , v_\codevar{pval}\!\ge\!0,
      v_\codevar{dec}\!=\!v_\codevar{pval}\!-\!1,
      v_\codevar{p}\!=\!\overline{v}_\codevar{p},
      v_\codevar{dec}\!=\!w_\codevar{pval},   \overline{v}_\codevar{p}\!=\!w_\codevar{p}, ... \},$\\$
          \emptyset, \,
    \emptyset, \,
\{v_{\codevar{p}} \corrvar v_{\codevar{p}}, v_{\codevar{*p}} \corrvar v_{\codevar{pval}}, ...\}$
\end{tabbing}
};

\node[state,label={[xshift=\labelxshift]2:$\GraphFNodeFifteen$},below =of 14.south west,
  anchor=north west, text width=9.4cm] (15)  
{ \vspace*{-.23cm} \begin{tabbing} $[$\=$((\code{rec}, 3),$\=$\{\code{p} = v_\codevar{p},
	\code{pval} = v_{\codevar{pval}},
	\code{dec} = v_\codevar{dec},
	\code{rrec} = v_\codevar{rrec}, ...\})],$\\
        $\{ w_\codevar{pval} < 0, v_\codevar{pval} \ge 0,
v_\codevar{dec} = v_\codevar{pval} - 1,  v_\codevar{dec}=w_\codevar{pval},
v_\codevar{rrec} = w_\codevar{pval}, ...\},$\\$
    \emptyset, \,
    \emptyset, \,
\{v_{\codevar{p}} \corrvar v_{\codevar{p}}, v_{\codevar{*p}} \corrvar v_{\codevar{pval}}, ...\}$
\end{tabbing}
};

\node[state,label={[xshift=\labelxshift]2:$\GraphFNodeEighteen$},below =of 15.south west,
	anchor=north west, text width=9.4cm] (18)  
{  \vspace*{-.23cm} \begin{tabbing} $[$\=$((\code{rec}, 3),$ \=$\{\code{p} = m_\codevar{p},
	\code{pval} = m_{\codevar{pval}},
	\code{dec} = m_\codevar{dec},
	\code{rrec} = m_\codevar{rrec}, ...\})],$\\
    $\{ m_\codevar{pval} \ge 0, m_\codevar{dec} = m_\codevar{pval} - 1, m_\codevar{rrec} < 0, ...\},$\\$
    \emptyset, \,
    \emptyset, \,
    \{v_{\codevar{p}} \corrvar m_{\codevar{p}}, v_{\codevar{*p}} \corrvar m_{\codevar{pval}}, ...\}$
    \end{tabbing}
};

\node[state,label={[xshift=\labelxshift]22:$\GraphFNodeNineteen$},below = of 18.south -|
  12.east, anchor=north east, text width=7.7cm,xshift=.5cm] (19)  
{  \vspace*{-.23cm} \begin{tabbing} $[$\=$((\code{rec}, 3),$ \=$\{\code{p} = z_\codevar{p},
	\code{pval} = z_{\codevar{pval}},
	\code{dec} = z_\codevar{dec},
	\code{rrec} = z_\codevar{rrec}, ... \}),$\\
	\>$((\code{rec}, 2),$ \>$\{\code{p} = v_\codevar{p},
	\code{pval} = v_{\codevar{pval}},
	\code{dec} = v_\codevar{dec}, ...\})],$\\
$\{{z_\codevar{pval} \ge 0},\allowbreak {z_\codevar{dec} = z_\codevar{pval} - 1},\allowbreak
    {z_\codevar{rrec} < 0},\allowbreak  {v_\codevar{pval} \ge 0},$\\$
 \phantom{\{}{v_\codevar{dec} =
      v_\codevar{pval} - 1},\allowbreak
{v_\codevar{p} =  \overline{v}_\codevar{p}},\allowbreak
    {v_\codevar{dec} = z_\codevar{pval}},\allowbreak
    {\overline{v}_\codevar{p} = z_\codevar{p}},  ...\},  \, 
    \emptyset, \,
    \emptyset,$\\$
  \{v_{\codevar{p}} \corrvar v_{\codevar{p}}, v_{\codevar{*p}} \corrvar v_{\codevar{pval}}, ...\}$
  \end{tabbing}
};

\node[state,label={[xshift=\labelxshift]90:$\GraphFNodeTwenty$}, text width=4.1cm,
    anchor=south west] (20) at ($(19.south west)+(-4.4cm,-0.4cm)$)
{ $
	[((\code{rec}, 3), {\{\code{p}\!=\!v_\codevar{p}},
	{\code{pval}\!=\!v_{\codevar{pval}}}, \allowbreak
	{\code{dec} = v_\codevar{dec}}, \allowbreak
	\code{rrec} = \hat{v}_\codevar{rrec}, ...\})], \linebreak
    \{   
    {v_\codevar{pval} \ge 0},\allowbreak {v_\codevar{dec} = v_\codevar{pval} -
	1},\linebreak {\phantom{\{}z_\codevar{rrec}\!<\!0},
  {\hat{v}_\codevar{rrec} \!=\! z_\codevar{rrec}}, ...\},$\\$  
    \emptyset, 
    \emptyset,
    \{v_{\codevar{p}} \corrvar v_{\codevar{p}}, v_{\codevar{*p}} \corrvar v_{\codevar{pval}}, ...\}$
};

\draw[eval-edge] (1) --  (2-mini);
\draw[omit-edge] (2-mini) --  (6-mini);
\draw[omit-edge] (2-mini) --  (10-mini);
\draw[eval-edge] (6-mini) --  (7);
\draw[eval-edge] (10-mini) -- (11-mini);
\draw[eval-edge] (11-mini) -- (12);

\coordinate (middle-7-13) at ($(7)!0.5!(13)$);
\coordinate[yshift=-1.5pt] (ca-exit-coord) at (middle-7-13 -| 12.west) ;
\coordinate[yshift=0pt] (ca-entry-coord) at (ca-exit-coord -| 13.north) ;
\draw[ca-edge, rounded corners=2pt] ($(12.north west)+(0,-3pt)$) --
node[yshift=3.5pt] {\tiny call abstraction}($(13.north east)+(0,-3pt)$);

\coordinate (call-abstraction-gen-edge) at ($(13.north west)+(-\edgenodedist,0)$);
\draw[gen-edge,rounded corners=5pt] (13.west) -- (call-abstraction-gen-edge) -- (call-abstraction-gen-edge |- 2.west) -- node[anchor=south east,xshift=5pt,yshift=0pt] {\tiny general.} (1.west);

\draw[fs-edge] (12) -- node[xshift=-10pt,anchor=east] {\tiny intersection with
$\GraphFNodeSeven$
} (14);

\draw[gen-edge] (15) -- node[anchor=west] {\tiny generalization} (18);

\draw[eval-edge] (14) -- node[anchor=west] {\tiny evaluation of \code{ret}} (15);
\draw[eval-edge] ($(19.north west)+(0,-2pt)$) -- node[anchor=east,xshift=5pt,yshift=7pt] {\tiny evaluation of \code{ret}} (20);

\coordinate (third-int-return-gen-edge-corner1) at ($(20.north west)+(-\edgenodedist,0pt)$);
\coordinate (third-int-return-gen-edge-corner2) at ($(18.south west)+(-\edgenodedist,-0pt)$);
\draw[eval-edge,rounded corners=5pt] ($(20.north west)+(5pt,0)$) --  node[anchor=west,yshift=1pt] { \tiny generalization} ($(18.south west)+(5pt,0)$);

\coordinate (left-of-identical-intersection) at ($(14.east)+(1.3cm,0)$);
\draw[fs-edge]  (12.south -| left-of-identical-intersection) node[anchor=north west,yshift=-36pt,align=center,font=\tiny] 
{\tiny intersection \\with $\GraphFNodeEighteen$
} -- (19.north -| left-of-identical-intersection) ;

\end{tikzpicture}

%% file: intersecting.tex
\subsection{Intersecting Call and Return States}\label{Intersecting States}

In our example, the return state $\GraphFNodeSeven$ weakly
represents all concrete states whose topmost stack frame is at the \code{ret} instruction
in the base case of \code{f}. Therefore, the execution of those concrete states may
continue after returning to a lower stack frame that is not depicted in the abstract
state $\GraphFNodeSeven$. In those concrete states, the stack frames below the topmost frame
must correspond to the lower stack frames of a call state. Recall that when creating the call
abstraction of a call state (e.g., in the step from $\GraphFNodeTwelve$ to
$\GraphFNodeThirteen$), we removed its lower stack frames. Therefore, this process must be
reversed in order to continue the execution with the former lower stack frames after
reaching a return state like $\GraphFNodeSeven$. Hence, for
 a  call
state $s_c$ and a return state $s_r$ of the same function $\sfunc$,
  we create an abstract state
$s_i$ that represents the case that  the execution of the topmost stack frame
  of $s_c$ ended in $s_r$ and should now return to the lower stack frames of $s_c$.
We call $s_i$ the \emph{intersection} of  $s_c$ and $s_r$,
and each 
 call state  $s_c$ has   \emph{intersection edges} to all its intersections.
The stack of $s_i$ 
 is constructed from the only stack frame of $s_r$ and the 
stack frames of $s_c$, except its first one.
Note that by this construction, intersected states always have more than one stack frame and
the topmost frame is at a \code{ret} instruction.

For example, the intersection $\GraphFNodeFourteen$ of $\GraphFNodeTwelve$ and
$\GraphFNodeSeven$ weakly represents those concrete states 
 $c_\GraphFNodeFourteen$ that arise from some concrete state $c_\GraphFNodeTwelve \repby
\GraphFNodeTwelve$ where the further execution of $c_\GraphFNodeTwelve$'s topmost frame
ends in a state represented by $\GraphFNodeSeven$.
All intermediate concrete states in the execution from $c_\GraphFNodeTwelve$ to
$c_\GraphFNodeFourteen$ are weakly represented by the abstract states on the path
from $\GraphFNodeTwelve$ via its call abstraction  $\GraphFNodeThirteen$ to the
state $\GraphFNodeOne$ and from there on to $\GraphFNodeSeven$.

In general, 
when traversing an SEG to simulate a program's execution, then the two 
types of outgoing edges of a call state $s_c$ (i.e., the intersection edge and the call
abstraction edge) serve different purposes.
The path from $s_c$  via the call abstraction to the entry state
and subsequently to the return state
can only be used to simulate the execution of the function in the topmost stack frame, but
not the subsequent execution of the lower stack frames, because 
return states only have a single stack frame at a return instruction.
For this reason, traversing this path is only justified if
the execution of the topmost frame does not terminate.
Symbolic execution then never reaches the return state, from where it would not be able to continue.
In contrast, if the traversal of the SEG reaches a call state $s_c$
and the execution of the function in the topmost stack frame does terminate,
then the traversal can continue by using the intersection edge. From there on, symbolic
execution continues by returning from the topmost stack frame.

In the following, we discuss which information can be included in the intersected
states. To this end, one has to take into account how the variables are renamed on the
path from the call state to the return state (\cref{Tracking Symbolic Variable
  Renamings}).
Afterwards, we show in \cref{Memory Information in the Intersection}
how to obtain the components $\AL$ and $\PT$ for the intersected
state.
Finally, the formal definition of state intersections is presented in \cref{Definition of
  State Intersections}. 


\subsubsection{Tracking Symbolic Variable Renamings}\label{Tracking Symbolic Variable Renamings}

As for all other edges except generalization edges, symbolic variables occurring in two
states connected by an intersection edge represent the same values.
Therefore, in our example graph, all information in  $\KB^\GraphFNodeTwelve$ is still valid in $\GraphFNodeFourteen$.
Of course, we would also like to include information of the return state
$\GraphFNodeSeven$ in the intersected state $\GraphFNodeFourteen$, but one
has to take into account that symbolic variables in $\GraphFNodeSeven$ do not necessarily
represent the same value as symbolic variables of the same name in $\GraphFNodeTwelve$.

For example, consider a concrete state $c_\GraphFNodeTwelve \repby \GraphFNodeTwelve$
where $v_\codevar{pval}$ is 0 and $v_{\codevar{dec}}$ is $-1$. Here, $v_\codevar{pval}$ and
$v_{\codevar{dec}}$ are the values of \code{pval} and \code{dec}, respectively, 
in the second stack frame.
Further execution of $c_\GraphFNodeTwelve$ then yields a state $c_\GraphFNodeSeven \repby
\GraphFNodeSeven$ where $v_\codevar{pval}$ is $-1$.
In this state $c_\GraphFNodeSeven$, $v_\codevar{pval}$ is the value of 
\code{pval} in the topmost and only stack frame.
That the values of  $v_\codevar{pval}$ differ in $\GraphFNodeTwelve$ and $\GraphFNodeSeven$ 
is due to the fact that a generalization edge with instantiation $\mu^\GraphFNodeThirteen$
is part of the path from $\GraphFNodeTwelve$ to $\GraphFNodeSeven$.
There, $\mu^\GraphFNodeThirteen(v_\codevar{*p}) = v_{\codevar{dec}}$ indicates that the variable
$v_{\codevar{dec}}$ of
$\GraphFNodeTwelve$ and $\GraphFNodeThirteen$ corresponds to the variable $v_\codevar{*p}$ of
$\GraphFNodeOne$. 
In the states on the path from $\GraphFNodeOne$ to 
$\GraphFNodeSeven$, $v_\codevar{pval} = v_\codevar{*p}$ holds. So
$v_{\codevar{dec}}$ is the value that is stored at the address  \code{p} before the recursive
call, and when executing the recursive call, this value is represented 
by $v_\codevar{*p}$ and $v_\codevar{pval}$ in the newly created stack frame.

In the following, let $s_c$ again be a call state of some function $\sfunc$, let
$s_\mathit{ca}$ be its call abstraction,
let $s_e$ be the subsequent entry state, and let $s_r$ be a return state of $\sfunc$.
Moreover, let $s_i$ be the intersection of $s_c$ and $s_r$, i.e.,
the stack of $s_i$ contains the topmost stack frame of $s_r$ and the lower frames of $s_c$. 
To take into account that variables of the same name in $s_c$ and $s_r$ may have different
values, a mapping
$\renamed$  from symbolic variables to pairwise different
fresh variables is applied to all components of $s_r$.	
Thus, the knowledge base of the intersection contains $\KB^{s_c}$ and
$\renamed(\KB^{s_r})$.

Moreover, $\KB^{s_i}$ should contain the information which variables from  $s_c$
and from $\renamed(s_r)$ correspond to each other. More precisely,
we would like to find  variables $v \in \Vsym(s_c)$
and $w \in \Vsym(s_r)$, such that in every  possible execution of $\sfunc$'s call
starting in $s_c$ and ending in
$s_r$, the value of $v$ in $s_c$ is equal to the value of $w$ in $s_r$.

The possible executions of $\sfunc$ starting in $s_c$ and ending in $s_r$ are represented in the SEG by the paths 
from $s_c$ to its call abstraction $s_\mathit{ca}$ and further to the entry state $s_e$ via a generalization edge.
From there onwards, one has to regard the paths from $s_e$ to $s_r$.
However, we only need to consider paths from  $s_e$ to $s_r$ that do not include call abstraction edges.
To see this, regard a path of the form $s_e, \dotsc, \overline{s_c},\overline{s_\mathit{ca}}, \overline{s_e}, \dotsc$, where $\overline{s_c}$ is a call state and 
$\overline{s_\mathit{ca}}$ is its call abstraction with subsequent entry state $\overline{s_e}$.
As described before, the states from $\overline{s_e}$ onwards only simulate an execution
of $\overline{s_c}$'s topmost stack frame that does not return to $\overline{s_c}$'s lower stack frames.
In particular, reaching $s_r$ from $\overline{s_e}$ onwards would mean that the return
statement of $s_r$ is in
a stack frame created by
subsequent calls of $\sfunc$ from $\overline{s_e}$ onwards, but it would not correspond to
the return from the stack frame
of $s_e$.
Note that this reasoning is independent from whether or not $\overline{s_c}$, $\overline{s_\mathit{ca}}$, and $\overline{s_e}$
are actually identical to $s_c$, $s_{\mathit{ca}}$, and $s_e$, which would indicate a recursive function call. 

Therefore, we are only interested in the renaming of symbolic variables $v
\in \Vsym(s_c)$
along paths of the form
$s_c, s_\mathit{ca}, s_e, \dotsc, s_r$, where  the fragment $s_e, \dotsc, s_r$ is an \emph{execution path}.
This means that $s_e$ is an entry state and $s_r$ is a return state of the same function.
Furthermore, an execution path must not contain call abstraction edges.
However, execution paths may contain cycles.

To integrate renaming information into the abstract states, we augment the states with an additional
component
$\CV$ to track \underline{v}ariable \underline{i}dentities.
$\CV$  contains entries 
of the form $v \corrvar w$ indicating that the variable $v$ of the preceding entry state 
corresponds to the variable $w$ in the current state.

More precisely, an entry $v \corrvar w$ in a state $s$ has the following semantics:
For all execution paths of the shape $s_e, \dotsc, s, \dotsc, s_r$, the value of
$v$ in $s_e$ is the same as the value of $w$ in $s$.
Note that in general, an execution path may contain $s$ several times.
This would indicate that $s$ is part of a loop that results from
executing the function in $s_e$.
Our semantics of $v \corrvar w$ then implies that $w$ must have the same value in $s$
every time that $s$ occurs in the execution path.

For all rules that evaluate \LLVM{} instructions or that result in refinement edges,
  the component $\CV$ does not have any impact
on the components of the new resulting state except for its
$\CV$ component.
Therefore, we do not have to adapt the formula representations
or the representation relation introduced in Def.\ \ref{def:StateFOLFormula},
\ref{def:representation-same-stack-size}, and \ref{def:weak-representation}.
There are only two graph construction steps that consider $\CV$, namely generalization and
intersection.

For each entry state $s_e$, we add an entry $v \corrvar v$ to $\CV$ for each symbolic
variable $v \in \Vsym(s_e)$. So for State $\GraphFNodeOne$ in
\cref{fig:initial-graph-function-f} and \ref{fig:full-graph-function-f},  we have $\CV^{\GraphFNodeOne} = \{ v_\codevar{p} \corrvar
v_\codevar{p}, v_\codevar{*p} \corrvar v_\codevar{*p} \}$. 

To compute $\CV$ in the other states, we adapt the symbolic execution rules:
In the call abstraction, all entries in $\CV$ are removed.
In all other rules except for the generalization rule, $\overline{\CV}$ in the resulting
state $\overline{s}$ is obtained from $\CV$ in the previous state $s$ as follows:
\begin{eqnarray*}
  \overline{\CV} &=& \{v \corrvar w \mid v \corrvar w \in \CV \land w \in \Vsym(\overline{s})\}\;\; \cup\\
  &&\{v \corrvar \overline{w} \mid v \corrvar w \in \CV \land
              \models \stateformula{\overline{s}} \implies w = \overline{w}\}
\end{eqnarray*}
So we preserve all entries $v \corrvar w$ from $\CV$ if $w$ still exists in $\overline{s}$.
Furthermore, if in $\overline{s}$ there is a variable $\overline{w}$ and we have $w = \overline{w}$ in
$\overline{s}$, then we also add an entry $v \corrvar \overline{w}$ to track which variables are
equivalent.
So in our example, since $\models \stateformula{\GraphFNodeTwo} \implies 
v_{*p} = v_{\codevar{pval}}$
and $v_{*p} \corrvar v_{*p} \in \CV^\GraphFNodeOne$
hold,  $v_\codevar{*p} \corrvar v_\codevar{pval}$ is
added to $\CV^{\GraphFNodeTwo}$ during the symbolic execution of the \code{load} instruction.

Finally, we extend the generalization rule from \cref{Abstracting the Call Stack} by the following condition:
\begin{itemize}
    \item[(f)] If $\Pos_1 \not= (\sfunc\code{.entry},0)$ for a function
        $\sfunc$ with entry block $\sfunc\code{.entry}$, we have for each
        $v \corrvar w \in \overline{\CV}$ that $v \corrvar \mu(w) \in \CV$.
\end{itemize}
This condition ensures that in order for an entry $v \corrvar w$ to be valid in a generalized state $\overline{s}$,
all states $s$ that have a generalization edge to $\overline{s}$ using an instantiation
$\mu$ must have a correspondingly renamed entry $v \corrvar \mu(w)$. 
In particular, this ensures that variable correspondence entries are consistent with respect to all cycles\footnote{
  As we are only interested in variable correspondences along execution paths,
  we only consider cycles here that do not contain call abstraction edges.
}
that the state may be part of. (Note that $v$ is a variable from the entry state $s_e$,
i.e., it is not renamed.)

However, the condition (f) is not required for generalization edges from call
abstractions to entry states (e.g., for the edge from $\GraphFNodeThirteen$ to  $\GraphFNodeOne$).
For the path between a call state to an entry state via its call abstraction,
we instead take possible renamings
into account during the computation of the intersection.

Recall that for the construction of the intersection of $s_c$ and $s_r$
we would like to identify variable correspondences between $s_c$ and $s_r$.
However, the $\CV$ entries of $s_r$ denote correspondences between variables of $s_r$ and variables of $s_e$,
rather than variables of $s_c$.
This allows us to determine the renaming information independently from call states.
By only tracking variable correspondences from the entry state onwards, we are able to add call states
to an existing entry state later on. In contrast, if
we tracked variable correspondences of call states directly, this would require the modification of 
the entry state and its successors.

To extend the knowledge base of the intersected state $s_i$ by the information on which
variables in $s_r$ and $s_c$ correspond to each other, we now 
need to combine each entry $v \corrvar w$ of $s_r$ with the renaming of variables possibly performed by
the generalization edge between $s_\mathit{ca}$ and $s_e$ using the instantiation $\mu$.
Hence, the entry $v \corrvar w$ of $s_r$ indicates that the variable $\mu(v)$ of $s_c$ has the same value as
the variable $w$ of $s_r$ for all possible executions of the function in $s_c$'s topmost
frame that  end in $s_r$. 
Thus, we extend $\KB^{s_i}$ by an equality
between the variables $\mu(v) \in \Vsym(s_c)$ and $\renamed(w)$ for $w \in \Vsym(s_r)$
whenever $v \corrvar w$ holds in $s_r$.

In our example, the intersected state $\GraphFNodeFourteen$ therefore has the \code{ret}
instruction at program position $(\code{term}, 1)$ in its topmost stack frame, where
$\renamed$ renamed all variables $v_\codevar{x} \in \Vsym(\GraphFNodeSeven)$ to
$w_\codevar{x}$. 
The lower stack frame of  $\GraphFNodeFourteen$ is taken from $\GraphFNodeTwelve$.
In the knowledge base we have $w_\codevar{pval} < 0$ (from $\GraphFNodeSeven$, where the
renaming $\renamed$ was applied), $v_\codevar{pval} \ge 0$, $v_\codevar{dec} = v_\codevar{pval} -
1$, $v_\codevar{p} = \overline{v}_\codevar{p}$, etc.\
(from $\GraphFNodeTwelve$), as well as $v_\codevar{dec} = w_\codevar{pval}$ (since
$v_\codevar{*p} \corrvar v_\codevar{pval} \in \CV^{\GraphFNodeSeven}$,
$\mu^\GraphFNodeThirteen(v_\codevar{*p}) = v_\codevar{dec}$, and
$\delta(v_\codevar{pval}) = w_\codevar{pval}$)
and $\overline{v}_\codevar{p} = w_\codevar{p}$ (since 
$v_\codevar{p} \corrvar v_\codevar{p} \in \CV^{\GraphFNodeSeven}$,
$\mu^\GraphFNodeThirteen(v_\codevar{p}) = \overline{v}_\codevar{p}$,
and $\delta(v_\codevar{p}) = w_\codevar{p}$).
Thus, $\GraphFNodeFourteen$ represents concrete states where the value
$v_\codevar{pval}$ at \code{p} was
originally 0 (since $v_\codevar{pval} \geq 0$ and $v_\codevar{pval} - 1 = v_\codevar{dec} =
w_\codevar{pval} < 0$). Hence,
the first recursive call immediately triggers the base case.


\subsubsection{Memory Information in the Intersection}\label{Memory Information in the Intersection}

Now we describe how to compute the components $\AL$ and $\PT$ for intersected states. 
Let the states $s_c, s_\mathit{ca}$, $s_e$, $s_r$, and $s_i$ be  as before.
In general, the memory information $\renamed(\AL^{s_r})$ and $\renamed(\PT^{s_r})$ from the
return state can always be added to the intersected state $s_i$.
This is because intuitively, the intersected state is a refinement of the return state,
where no additional instructions have been evaluated.
However, it is more challenging to determine which memory information of the call state
can be added to the intersected state.

\paragraph{Heap Allocations}

Entries from $\AL^{s_c}$ can only be added to $AL^{s_i}$ if they have not been deallocated during
the execution of $s_c$'s topmost frame that ended in $s_r$.
In addition, allocations of the call state may only be added to the intersected state if they can be proven to be
disjoint from any entries in $\renamed(\AL^{s_r})$. This is needed to guarantee that the intersected state does not violate
the invariant of all allocations in a state being disjoint.

To ensure these two conditions, we only add an allocation
$\alloc{v_1}{v_2}$ from $\AL^{s_c}$ to $\AL^{s_i}$ if it has been removed during
the generalization from $s_\mathit{ca}$ to $s_e$
(i.e., if there exists no allocation corresponding to $\alloc{v_1}{v_2}$ in $s_e$).
Formally, this means that there exists no $\alloc{w_1}{w_2} \in \ALL(s_e)$ such that
  $\models \stateformula{s_{\mathit{ca}}} \implies  v_1 =
  \mu(w_1) \wedge v_2 = \mu(w_2)$, where $\mu$ is the instantiation used for the
  generalization from $s_{ca}$ to $s_e$. 

It is easy to see that $\alloc{v_1}{v_2}$ satisfies both conditions that have to be
imposed on allocations in order to add them
to the intersection:
The allocation $\alloc{v_1}{v_2}$ was removed during the generalization without being changed otherwise.
This means that it is present in all concrete states represented by  $s_c$, $s_\mathit{ca}$, $s_e$,
and $s_e$'s successors.
However, any access to this allocation by any of $s_e$'s successors would yield the $\ERROR$ state during symbolic execution,
as the allocation is not available in those abstract states.
This means that the allocation cannot be deallocated during subsequent execution.
In addition, any newly allocated memory is guaranteed to  be disjoint from
$\alloc{w_1}{w_2}$.

In contrast, if the allocation $\alloc{v_1}{v_2}$ had a counterpart $\alloc{w_1}{w_2}$ in the entry
state, then there are several
possibilities:
\begin{itemize}
\item The allocation is deallocated at some point prior to reaching the return state.
This means that it must not be added to the intersected state.
\item The allocation is not deallocated and has a counterpart in the return state.
This means that the allocation is in $\renamed(\AL^{s_r})$ and therefore already
part of the intersection.
\item The allocation is not deallocated, but it also does not have a counterpart in the return
  state. There are two possible reasons for this. The first possibility is that the
allocation is removed along an intersection
edge on an execution path from $s_e$ to $s_r$. In this case we cannot ensure that it was not freed during the function execution represented by the 
intersection edge. Hence, it must not be added to the intersected state $s_i$ that is currently being
constructed.

The other possibility is that 
the allocation has been removed along a generalization edge in the path from $s_e$ to $s_r$ (i.e.,
this is not the generalization edge from $s_{ca}$ to $s_e$). 
Here, one would have to analyze the possible execution paths from $s_e$ to $s_r$ to make sure that
  that there was definitely no deallocation before the allocation was lost during generalization.
  Since this only occurs in rare cases,
  we do not add such allocations in order to ease the formalization.
\end{itemize}

To formally reason about allocations being removed in generalizations, we introduce the following definition.

\begin{definition}[Predicate $\lostAL$]
Let $s, \overline{s}$ be states such that $s$ has a generalization edge to $\overline{s}$ using
an instantiation $\mu$. Furthermore, let $\alloc{v_1}{v_2} \in \ALL(s)$.
Then $\lostAL(s, \overline{s}, \alloc{v_1}{v_2})$ holds iff there exists no
$\alloc{w_1}{w_2} \in \ALL(\overline{s})$ such that 
$\models \stateformula{s} \implies  v_1 = \mu(w_1) \wedge v_2 = \mu(w_2)$.
\end{definition}

\paragraph{Stack Allocations}

Recall that in the step from the call state $s_c$ to the call abstraction $s_{ca}$,
all but the topmost stack frames of the call state $s_c$ are removed.
However, the stack allocations of the deleted frames are moved to the (only) stack frame
of $s_{ca}$.
This means that
when simulating the execution of $\rfunc$'s call by the path from $s_c$ over $s_{ca}$ and
$s_e$ to $s_r$, 
the topmost stack frame of the return state $s_r$ may contain allocations
that were originally part of the lower stack frames of $s_c$.
(Further call abstractions cannot happen on the path from $s_e$ to $s_r$, since
here we only have to regard execution paths.)

When intersecting  $s_c$ and  $s_r$, stack allocations must be restored
to their correct frames.
As the lower stack frames of $s_c$ were not active during the execution that led to $s_r$, those
stack allocations cannot have been deallocated and they should therefore be added to the respective frames of the
intersection $s_i$. But when
 turning the only stack frame of the return state $s_r$ into the
topmost frame of the 
intersected state $s_i$, 
we remove all of its stack allocations.
This is done to
 guarantee the disjointness of all stack allocations in the intersected state.
As mentioned before, the reason is that $s_r$'s only stack frame may contain allocations
that were moved there from lower stack frames of $s_c$ during the call abstraction from
$s_c$ to $s_{ca}$.
Intersected states are symbolically executed by evaluating the return instruction in their topmost stack frame,
which would remove the allocations in this stack frame anyway.

\paragraph{Points-To Entries}

As with allocations, points-to information from the return state $s_r$ can always
be taken over to the
intersected state $s_i$, but points-to atoms from the call state can only be added to the
intersection $s_i$
if they have not been invalidated.

Hence, we only copy an entry $w_1
\pointsto[\codevar{ty}] w_2$
from $s_c$ to $s_i$ if it is part of an allocation $\alloc{v_1}{v_2}$ that is
 lost during the generalization
 from the call abstraction $s_{ca}$ to the entry state $s_e$.
In other words, $\alloc{v_1}{v_2}$ must contain all addresses from
$w_1$ to  $w_1 + \mathit{size}(\code{ty})-1$ and
$\lostAL(s_\mathit{ca},s_e,
\alloc{v_1}{v_2})$ holds.
This is sound, 
since then the points-to atom  $w_1
\pointsto[\codevar{ty}] w_2$
cannot have 
been modified
during the summarized function execution.
The reason is that
our symbolic execution rules can only
access or modify the content of an address  if the address is known to be in an allocated part of the memory
(otherwise, one would violate memory safety).

Note that it would also be possible to add those $\PT$ entries from $s_c$ to the intersection that
are part of an allocation that
is not removed during the generalization to $s_e$,
provided that it is not modified
during the execution summarized by the intersection edge.
We have implemented this improvement in \aprove{} by 
augmenting allocations with
an additional flag
that indicates whether or not an allocation has been modified.
But to ease readability,  we did not include it in the formalization of
this paper.

\subsubsection{Definition of State Intersections}\label{Definition of State Intersections}

To sum up, the state intersection is defined as follows for a call state $s_c$ and a
corresponding return state $s_r$.

\begin{definition}[State Intersection]
	\label{def:intersection}
	Let $s_c = (\FR_1^c \,\cdot\, \widetilde{\CS}^c, \KB^c, \AL^c, \PT^c,  \CV^c)$ be a call
        state and $s_r = ([(\Pos_1^r,\LV_1^r,\AL_1^r)], \KB^r,\AL^r,\PT^r, \CV^r)$ be a
        return state of the same function $\sfunc$.
        Let $s_\mathit{ca}$ be the call abstraction of $s_c$ and let $s_e$ be an entry
        state that is a generalization of $s_{ca}$.
	Let $\mu: \Vsym(s_e) \to \Vsym(s_c)$ be the instantiation used for the generalization and
	  let  $\renamed: \Vsym(s_r) \to \Vsym$ be a function that maps all symbolic variables of
        $s_r$ to pairwise different fresh ones.
A state $s_i$ is an \emph{intersection} of $s_c$ and $s_r$ iff it has the form 
$((\Pos_1^r,\renamed(\LV_1^r),\emptyset) \cdot \widetilde{\CS}^c,\allowbreak
        \KB^i,\allowbreak \AL^i,\allowbreak \PT^i,\allowbreak \CV^i)$, where we
        have:
	\[\begin{array}{lllll}
	  \KB^i &=  &\renamed(\KB^r) &\cup &\KB^c \cup \{ \mu(v) = \renamed(w) \mid
          v \corrvar w \in \CV^r \} \\
      \AL^i &= &\renamed(\AL^r) &\cup &\{ \alloc{v_1}{v_2} \in \AL^{c}
          \mid \lostAL(s_{\mathit{ca}},s_e, \alloc{v_1}{v_2}) \\
      \PT^i &= &\renamed(\PT^r)\\
        && \multicolumn{3}{l}{\cup \: \{ (w_1 \pointsto[\codevar{ty}] w_2) \in \PT^c
            \mid \lostAL(s_\mathit{ca},s_e,
            \alloc{v_1}{v_2})
            \text{ holds for some }}\\            
          && \multicolumn{3}{r}{\phantom{\cup\;} 
 \alloc{v_1}{v_2} \in \AL^c \text{ where }
          \models \stateformula{s_c} \implies
          v_1 \leq w_1  \land w_1 + \mathit{size}(\code{ty})-1 \leq v_2 \}}
\\
     \CV^i &=& \multicolumn{3}{l}{\phantom{\cup\:}   \{v \corrvar w \mid v \corrvar w \in \CV^c \land w \in \Vsym(s^i)\}}\\
             && \multicolumn{3}{l}{ \cup \:  \{v \corrvar \overline{w} \mid v \corrvar w \in \CV^c \land
              \models \stateformula{s^i} \implies w = \overline{w}\}}
	\end{array}\]
\end{definition}
So the
variable identities $\CV^i$ are built in the same way as for other symbolic execution rules.

In our example,
when creating the intersected state $\GraphFNodeFourteen$ from the call state  $\GraphFNodeTwelve$
and the return state $\GraphFNodeSeven$, we have $\AL^\GraphFNodeSeven = \varnothing$ and
$\PT^\GraphFNodeSeven = \varnothing$. The information from  $\AL^\GraphFNodeTwelve = \{
\alloc{v_{\codevar{p}}}{v_{\codevar{p}}} \}$ and $\PT^\GraphFNodeTwelve = \{v_{\codevar{p}}
\pointsto[\codevar{i8}] v_{\codevar{dec}} \}$ is not taken over to  $\GraphFNodeFourteen$, since
$\alloc{v_{\codevar{p}}}{v_{\codevar{p}}}$ is not removed during the generalization from
$\GraphFNodeThirteen$ to $\GraphFNodeOne$, i.e.,
$\lostAL(\GraphFNodeThirteen,\GraphFNodeOne,
\alloc{v_{\codevar{p}}}{v_{\codevar{p}}})$ does not hold.

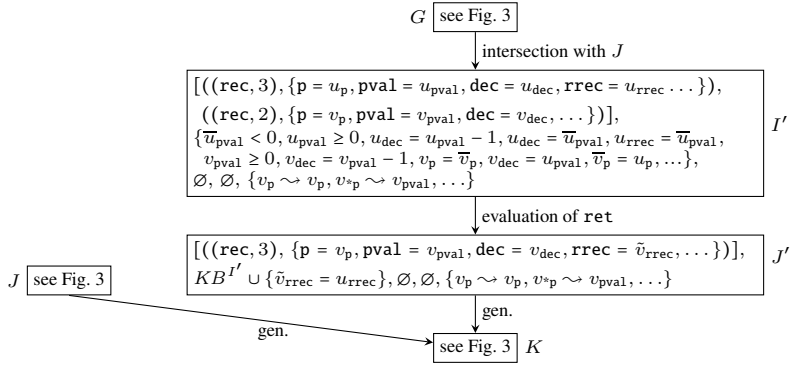
\begin{figure}[t]
\newcommand{\ydist}{0.1cm}
\newcommand{\xdist}{0.5cm}
\newcommand{\setfont}{\scriptsize}
\begin{tikzpicture}[]
\scriptsize
\def\widetwidth{7cm}
\tikzstyle{state}=[
                           font=\scriptsize,
                           draw]

\node[state, label=180:$G$, align=left] (Inter-G) 
{see Fig.~\ref{fig:full-graph-function-f}};

\node[state,label={[xshift=0]2:$I'$},below=0.5cm of Inter-G, text width=8.9cm] (Inter-I-p) 
     { \vspace*{-.23cm}
       \begin{tabbing}
         $[$\=$((\code{rec}, 3),$ \=$\{ \code{p} = u_\codevar{p},
	\code{pval} = u_{\codevar{pval}},
        \code{dec} = u_\codevar{dec},
        \code{rrec} = u_{\codevar{rrec}}
	  \dots \}),$\\
	\>$((\code{rec}, 2),$ \>$\{\code{p} = v_\codevar{p},
	\code{pval} = v_{\codevar{pval}},
	\code{dec} = v_\codevar{dec},
        \dots\})],$\\$
  \{ 
      \overline{u}_\codevar{pval} < 0, u_\codevar{pval} \ge 0,
      u_\codevar{dec} = u_\codevar{pval} - 1,  u_\codevar{dec}= \overline{u}_\codevar{pval},
      u_\codevar{rrec} = \overline{u}_\codevar{pval},$\\$
      \phantom{\{}v_\codevar{pval} \ge 0, v_\codevar{dec} = v_\codevar{pval} - 1, v_\codevar{p} = \overline{v}_{\codevar{p}},
      v_\codevar{dec} = u_\codevar{pval}, \overline{v}_\codevar{p} = u_\codevar{p},
 ... \},$\\$
    \emptyset, \,
    \emptyset,$
    $\{v_{\codevar{p}} \corrvar v_{\codevar{p}}, v_{\codevar{*p}} \corrvar v_{\codevar{pval}}, \ldots\}$
    \end{tabbing}
};

\node[state, label=2:$J'$, align=left, text width=8.9cm,below=0.5cm of Inter-I-p] (Inter-J-p) 
{
  $[((\code{rec}, 3), \, \{ \code{p} = v_\codevar{p},
    \code{pval} = v_\codevar{pval}, \code{dec} = v_\codevar{dec}, \code{rrec} = \tilde{v}_\codevar{rrec}, \dots \})],$\\
$\KB^{I'} \cup \{ \tilde{v}_\codevar{rrec} = u_\codevar{rrec} \},\allowbreak
    \emptyset,
     \emptyset, \{v_{\codevar{p}} \corrvar v_{\codevar{p}}, v_{\codevar{*p}} \corrvar v_{\codevar{pval}}, \ldots\}$
};

\node[state, label=180:$J$, align=left,left=1cm of Inter-J-p.south west, anchor=south east] (Inter-J) 
{see Fig.~\ref{fig:full-graph-function-f}};

\node[state, label=2:$K$, align=left,below=0.5cm of Inter-J-p] (Inter-K) 
{see Fig.~\ref{fig:full-graph-function-f}};

\draw[eval-edge] (Inter-G) -- node[anchor=west] { intersection with $J$} (Inter-I-p);
\draw[eval-edge] (Inter-I-p) -- node[anchor=west] {evaluation of \code{ret}} (Inter-J-p);

\draw[gen-edge] (Inter-J.south) -- node[anchor=west,yshift=-.2cm] {gen.\ } (Inter-K);
\draw[gen-edge] (Inter-J-p.south) -- node[anchor=west] {gen.\ } (Inter-K);

\end{tikzpicture}
\vspace*{-0.2cm}
\caption{\label{fig:state-15-prime} 
Intermediate states during analysis of \code{f}
}
\end{figure}

Afterwards, applying the symbolic execution rule for the $\code{ret}$ instruction
yields the state $\GraphFNodeFifteen$.
Here, the value $v_{\codevar{rrec}}$ of the program variable \code{rrec} is equal to the
result $w_\codevar{pval}$ of \code{f}'s recursive call. Note that $J$
is another return state. Thus, one now 
has to construct the intersection of the call state $\GraphFNodeTwelve$ and
$\GraphFNodeFifteen$.
This yields another intersected state $\GraphFNodeFourteen'$ 
shown in Fig.~\ref{fig:state-15-prime}.
In $\GraphFNodeFourteen'$, we transformed all information taken from $J$
by a renaming $\delta^{I'}$ that replaces 
all symbolic variables $v_\codevar{x}$ by $u_\codevar{x}$ and $w_\codevar{x}$ by
$\overline{u}_\codevar{x}$.
$\KB^{I'}$ also contains the equalities
$v_\codevar{dec} = u_\codevar{pval}$ (as $v_\codevar{*p} \corrvar v_\codevar{pval} \in \CV^J$,
$\mu^H(v_\codevar{*p}) = v_\codevar{dec}$, and 
$\delta^{I'}(v_\codevar{pval}) = u_\codevar{pval}$)
and
$\overline{v}_\codevar{p} = u_\codevar{p}$
(as $v_\codevar{p} \corrvar v_\codevar{p} \in \CV^J$,
$\mu^H(v_\codevar{p}) = \overline{v}_\codevar{p}$, and 
$\delta^{I'}(v_\codevar{p}) = u_\codevar{p}$).

By symbolically evaluating the  \code{ret}
instruction in the topmost stack frame of $I'$, one obtains the state $J'$. Now the value
$\tilde{v}_\codevar{rrec}$ of the program variable \code{rrec} is equal to the result
$u_\codevar{rrec}$ of \code{f}'s recursive call.

In state $J$, we had $\models \KB^{J} \implies v_\codevar{pval} \ge 0 \wedge  v_\codevar{pval} - 1 = w_\codevar{pval} < 0 $,
which can be simplified to $\models \KB^{J} \implies v_\codevar{pval} = 0$.
Analogously, in  $J'$, we have
$\models \KB^{J'}  \implies  {u_\codevar{pval} \ge 0} \wedge u_\codevar{pval} - 1 =\allowbreak  \overline{u}_\codevar{pval} < 0$,
which implies $\models \KB^{J'} \implies u_\codevar{pval} = 0$.
Moreover, we obtain $\models \KB^{J'} \implies v_\codevar{pval} - 1 = v_\codevar{dec} =  u_\codevar{pval}$.
The latter equality holds due to the entry $v_\codevar{*p} \corrvar v_\codevar{pval}$ in $\CV^J$,
which allowed us to add $v_\codevar{dec} = u_\codevar{pval}$ to $\KB^{I'}$.
Together, this implies $\models \KB^{J'} \implies v_\codevar{pval} = 1$.
Intuitively, this reflects the fact that in $J$, the original value at the pointer \code{p} was 0,
whereas in $J'$ the original value was 1.

\subsection{Complete Symbolic Execution Graphs}\label{Complete SEGs}

Note that 
  the single stack frames of both $\GraphFNodeFifteen$ and $\GraphFNodeFifteen'$ 
  are at the same program position
  and their $\LV$-functions have the same domain.
To obtain a finite symbolic execution graph,
we \emph{merge} the return states $\GraphFNodeFifteen$
and $\GraphFNodeFifteen'$ to a single generalized return state.
More precisely, we merge each pair of return states $s_r$ and $\overline{s}_r$ if they are
at the same
program position of a recursive function (or a function in a group of mutually recursive
functions), if the domains of their $\LV$-functions are identical, and if there
exists an entry state $s_e$ that has an execution path to both $s_r$ and
$\overline{s}_r$. If the latter condition is not satisfied, then merging does not have any
advantages, since both return states are
part of independent analyses of the same function.

We presented a heuristic for merging states in \cite{LLVM-JAR}  that is used for 
such similar return states if there is not
yet
a more general state in the SEG that one could draw a generalization edge to. For two states $s$ and
$\overline{s}$, our merging heuristic
generates a new  state $g$ which is a generalization of both $s$ and $\overline{s}$. This heuristic
can be used here to
obtain the state $\GraphFNodeEighteen$, where the heuristic introduces fresh symbolic
variables $m_\codevar{x}$.\footnote{The heuristic's general idea for merging two states $s$ and $\overline{s}$ to a more general state
$g$ is to first extend $\stateformula{s}$ to $\extendedformula{s}$,  which contains
additional constraints implied by $\stateformula{s}$. 
	Then, those formulas of $\extendedformula{s}$ that are also implied by
        $\stateformula{\overline{s}}$ are added to $\KB^{g}$ (where one of course has to
        take the renaming of the variables into account). 
	To yield the state $\GraphFNodeEighteen$, the definition of $\extendedformula{s}$
        from \cite{LLVM-JAR}        
        has to be extended as follows:
	For expressions $t_1 < t_2 \in
        \extendedformula{s}$
        where $\extendedformula{s}$ also contains an inequality with a term $t_3$      
        such that $\models \stateformula{s} \implies t_3 = t_1$, we add  $t_3
        < t_2$ to $\extendedformula{s}$. We proceed analogously for similar cases
(e.g., where   $t_2 < t_1 \in
        \extendedformula{s}$). So in our example, since both 
$w_\codevar{pval} < 0$ and $v_\codevar{rrec} = w_\codevar{pval}$ are contained in
        $\KB^{\GraphFNodeFifteen} \subseteq \extendedformula{\GraphFNodeFifteen}$, we have
 $v_\codevar{rrec} < 0$ in $\extendedformula{\GraphFNodeFifteen}$. For that reason,
        $m_\codevar{rrec} < 0$ is contained in the generalized state $\GraphFNodeEighteen$.}
Of course, our merging heuristic from \cite{LLVM-JAR} 
now has to be extended to handle
the set $\CV$ as well.
If there are entries $v_e \corrvar v_s \in \CV^s$, $v_e \corrvar v_{\overline{s}} \in \CV^{\overline{s}}$,
and a $v \in \Vsym(g)$ such that $\mu^s(v) = v_s$ and $\mu^{\overline{s}}(v) = v_{\overline{s}}$ (where $\mu^s$
and $\mu^{\overline{s}}$ are the instantiations for the generalizations from $s$ to $g$ and from $\overline{s}$
to $g$, respectively), then $\CV^g$ contains  $v_e \corrvar v$. For example, since we have
$v_\codevar{p} \corrvar v_\codevar{p}$ and 
$v_\codevar{*p} \corrvar v_\codevar{pval}$ in both states
$\GraphFNodeFifteen$ and $\GraphFNodeFifteen'$, we add $v_\codevar{p} \corrvar m_\codevar{p}$
and $v_\codevar{*p} \corrvar m_\codevar{pval}$ to $\CV^\GraphFNodeEighteen$.

For return states like
$\GraphFNodeFifteen$  that have outgoing generalization edges,
we do not have to include any intersections in the graph. The reason is that it is enough
to construct an intersection with the generalized return state $\GraphFNodeEighteen$,
since the resulting intersection is more general than an intersection with the more
specific return state $\GraphFNodeFifteen$. Thus, the states
 $\GraphFNodeFourteen'$ and $\GraphFNodeFifteen'$ can be removed from the graph provided
that we construct an intersection of $\GraphFNodeTwelve$ with the generalized return state $\GraphFNodeEighteen$
instead.

The   state $\GraphFNodeEighteen$ contains the knowledge $m_\codevar{pval}
\ge 0$ and $m_\codevar{rrec} < 0$.
It represents all concrete states where the value at \code{p} was originally some
non-negative number $k$ and $k+1$ recursive invocations have finished.
So while the return state $\GraphFNodeSeven$ corresponds to runs of $\code{f}$ that
directly end in $\code{f}$'s non-recursive case, the return state $\GraphFNodeEighteen$ corresponds to runs
of $\code{f}$ with at least one recursive call. 
The return state 
$\GraphFNodeEighteen$ has to be intersected with the call state
$\GraphFNodeTwelve$, yielding state $\GraphFNodeNineteen$.
Here, we used a renaming $\delta^\GraphFNodeNineteen$ with
$\delta^\GraphFNodeNineteen(m_\codevar{p}) = z_\codevar{p}$,
$\delta^\GraphFNodeNineteen(m_\codevar{pval}) = z_\codevar{pval}$, etc.
Since
$v_\codevar{*p} \corrvar m_\codevar{pval} \in \CV^\GraphFNodeEighteen$ and
$v_\codevar{p} \corrvar m_\codevar{p} \in \CV^\GraphFNodeEighteen$,
we have
$v_\codevar{dec} = z_\codevar{pval}  \in \KB^\GraphFNodeNineteen$
(since $\mu^\GraphFNodeThirteen(v_\codevar{*p}) = v_\codevar{dec}$
and $\delta^\GraphFNodeNineteen(m_\codevar{pval}) = z_\codevar{pval}$) and
$\overline{v}_\codevar{p} = z_\codevar{p}  \in \KB^\GraphFNodeNineteen$
(since $\mu^\GraphFNodeThirteen(v_\codevar{p}) = \overline{v}_\codevar{p}$ and
$\delta^\GraphFNodeNineteen(m_\codevar{p}) = z_\codevar{p}$).
Evaluating the return instruction in $\GraphFNodeNineteen$ leads to
its successor $\GraphFNodeTwenty$, which $\GraphFNodeEighteen$ is a generalization of.

This concludes the analysis of the function \code{f}, as its SEG in
Fig.~\ref{fig:full-graph-function-f} is \emph{complete}: 

\begin{definition}[Complete SEG]
\label{def-complete-seg}
A symbolic execution graph is \emph{weakly complete} iff
\begin{enumerate}
\item For all of its leaves $s$ we either have $s = \ERROR$, $\stateformula{s}$ is unsatisfiable,
or $s$ has only one stack frame which is at a \code{ret} instruction.
\item Each call state of some function $\sfunc$ has exactly one call
  abstraction which in turn has an outgoing generalization edge to an entry state
  of $\sfunc$. 
\item For all pairs of return states $s_r$  and call states $s_c$  of some function
  $\sfunc$, the following holds: If $s_r$ has no outgoing generalization edge and
  the entry state of $\sfunc$ following $s_c$ has an execution path to $s_r$, then
  there is an intersection edge from
  $s_c$  to the intersection of $s_c$ and $s_r$.
\end{enumerate}
A symbolic execution graph is \emph{complete} iff it is weakly complete and does not contain $\ERROR$. 
\end{definition}

Note that we do not create intersections with return states that have been generalized to a more general one.
Moreover, we only require intersections of call and return states if the entry state following the
call state has an execution path to the return state.
If this is not case, then the return state belongs to a different, independent analysis of the same function,
starting from a different entry state. Thus, we do not only avoid merging of
return states from independent analyses of the same function, but we also do not
create intersections between call and return states from such independent analyses.

In \cite{LLVM-JAR}, we proved the correctness of our symbolic execution w.r.t.\ the formal
definition of the \LLVM{} semantics from the \textsf{Vellvm} project
\cite{Vellvm}.
Similar to \cite[Thm.\ 10]{LLVM-JAR}, we now show
that every \LLVM evaluation of concrete states can be
simulated by symbolic execution of abstract states.
Let $\cto$ denote \LLVM's evaluation relation on concrete states, i.e., $c \cto \overline{c}$
holds iff $c$ evaluates to $\overline{c}$ by executing one \LLVM instruction.
Similarly, $c \cto \ERROR$
 means that the evaluation step performs an operation that may lead to undefined behavior.
An \LLVM program is \emph{memory safe} for $c \neq \ERROR$  iff there is no evaluation $c \cto^+
\ERROR$,
where $\cto^+$
is the transitive closure of $\cto$.
The following theorem states that for each 
computation of concrete states there is
a corresponding path in the SEG whose
abstract
states represent the concrete states of the computation.

\setcounter{thm-soundness-graph-construction-main-text}{\value{theorem}}

\begin{theorem}[Soundness of the Symbolic Execution Graph]
\label{thm:soundness-graph-construction-main-text}
	Let $\pi = c_0 \ccto c_1 \ccto c_2 \ccto \dots $ be a (finite resp.\ infinite)
        \emph{\LLVM} evaluation of concrete states such that $c_0$ is represented by some 
        state $s_0$ in  a weakly complete SEG $\GG$. 
	Then there exists a (finite resp.\ infinite) sequence of states $s_{0},  s_{1},  s_{2},
        \dotsc$ where $\GG$ has an edge from $s_{j-1}$ to $s_j$ if $j > 0$,
        and  there exist
        $0 = i_0 \leq i_1 \leq  \ldots$ with $c_{i_j} \repby s_j$
        for all $j \geq 0$.
  Moreover, if $\pi$ is infinite then the  corresponding
        sequence of abstract states in $\GG$
        is infinite as well.
        In contrast, if $\pi$ is finite and ends at some concrete state $c$, then the
        sequence of states in $\GG$ ends at some state $s$ with $c 
        \repby s$.        
\end{theorem}

 The proof relies on the fact that our symbolic execution rules correspond
to the actual execution of \LLVM when they are applied to concrete states.
Moreover, terminating executions of 
function calls can be simulated using intersection edges
(for that reason,  some subsequences of concrete states can be ``skipped'' (i.e., not represented by abstract states)
in Thm.\ \ref{thm:soundness-graph-construction-main-text})
and non-terminating 
function calls can be simulated by following a call abstraction edge to the 
entry state of the called function and by continuing the execution from there.

Note that a complete SEG does not contain $\ERROR$. Hence, the program is memory safe for
all concrete states represented
in the SEG.

\setcounter{cor-memory-safety}{\value{theorem}}

\begin{corollary}[Memory Safety of \LLVM Programs]
\label{cor:memory-safety}
 Let $\PP$ be a program with a complete
symbolic execution graph $\GG$. Then $\PP$ is memory safe for all states represented by $\GG$.
\end{corollary}

%% file: reusingfunctiongraphs.tex
\subsection{Modular Re-Use of Symbolic Execution Graphs}\label{Re-Use}
\label{sect:re-using-segs}

In \cite{LLVM-JAR}, whenever an \LLVM{} function \code{g} calls an auxiliary function \code{f},
then during the construction of \code{g}'s symbolic execution graph, one obtained a new
abstract state whose topmost stack frame is at the start of the function
\code{f}. To evaluate this state further, now one had to execute \code{f} symbolically and only after the end of
\code{f}'s execution, one could remove the topmost stack frame and continue the further
execution of \code{g}. So even one had analyzed termination of \code{f} before, one could
not re-use its symbolic execution graph, but one had to perform a new symbolic execution of
\code{f} whenever it is called. This missing modularity had severe drawbacks for the
performance of the approach and moreover, it prevented the analysis of functions with
recursive calls.

In \cref{Abstracting the Call Stack}-\ref{Complete SEGs}, we showed how to
abstract from the call stack by using 
call abstractions and intersections. This does not only allow us to analyze recursive
functions, but it also allows us to re-use previously computed symbolic
execution graphs of auxiliary functions. Thus, it is the key for the modularization of our
approach.

To illustrate this,
we now show how the previously computed symbolic execution graph of \code{f} from \cref{fig:full-graph-function-f}
can be re-used in a modular way
to analyze functions like \code{main} from \cref{sect:domain} which call \code{f}, see \cref{fig:graph-main}.
We assume  that \code{main}'s call of \code{f} is at program position $\Pos_c$ inside of \code{main}'s
\code{while}-loop, yielding a call state $\GraphMainNodeCall$.
Its call abstraction $\GraphMainNodeCallAbstraction$ has a generalization edge to
$\GraphFNodeOne$, the entry state of \code{f}. 

\begin{figure}[t]
\begin{center}
\input{graph_main}
\caption{\label{fig:graph-main} SEG for \code{main} (extract)}
\end{center}
\end{figure}
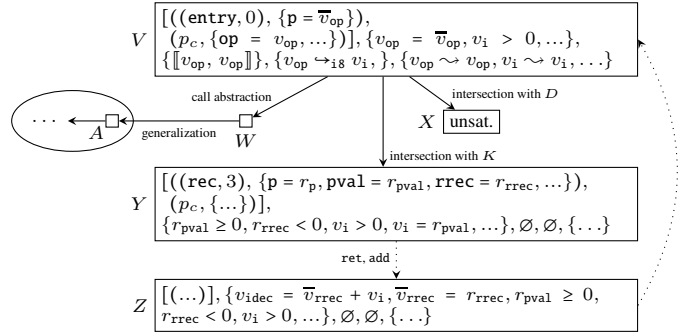

Intersecting
the call state $\GraphMainNodeCall$ with the return state $\GraphFNodeSeven$
of \code{f} yields a state $\GraphMainNodeIntersectionUnsat$, whose corresponding state formula
$\stateformula{\GraphMainNodeIntersectionUnsat}$ is unsatisfiable.
The reason is that in
$\KB^X$ we have $v_\codevar{i} > 0$ (from $\GraphMainNodeCall$), $\delta^X(v_\codevar{pval}) < 0$
(from $\GraphFNodeSeven$, where a renaming $\delta^X$ is applied) 
and $v_\codevar{i} = \delta^X(v_\codevar{pval})$ (since $v_\codevar{*p} \corrvar v_\codevar{pval}  \in \CV^\GraphFNodeSeven$ and $v_\codevar{i}$ is identified with $v_\codevar{*p}$ in the
generalization from $\GraphMainNodeCallAbstraction$ to $\GraphFNodeOne$).
Intuitively, the unsatisfiability of
$\stateformula{\GraphMainNodeIntersectionUnsat}$
is due to the fact that when
\code{f} is called from \code{main}, the value at
\code{p} in \code{f} cannot be negative due to the condition of \code{main}'s
\code{while}-loop and thus, it cannot
immediately trigger the base case of $\code{f}$.

The intersection of the call state $\GraphMainNodeCall$ with the return state
$\GraphFNodeEighteen$ of \code{f} yields the state
$\GraphMainNodeIntersection$.
Here, we again have $v_\codevar{i} > 0$ (from $\GraphMainNodeCall$), but now we also obtain
$r_\codevar{pval} \geq 0$
(from $\GraphFNodeEighteen$, where $m_\codevar{pval}$ is renamed to $r_\codevar{pval}$, i.e.,
$\delta^Y(m_\codevar{pval}) = r_\codevar{pval}$).
Moreover, since $v_\codevar{*p} \corrvar m_\codevar{pval}  \in \CV^\GraphFNodeEighteen$,
in the intersection $\GraphMainNodeIntersection$ we have an equality between
$\mu^\GraphMainNodeCallAbstraction(v_\codevar{*p})$  and $\delta^Y(m_\codevar{pval})$, where
$\mu^\GraphMainNodeCallAbstraction(v_\codevar{*p})$ is $v_\codevar{i}$ and
$\delta^Y(m_\codevar{pval})$ is $r_\codevar{pval}$.
Again, in the intersection we have $\AL^\GraphMainNodeIntersection = \PT^\GraphMainNodeIntersection
= \emptyset$, since $\AL^\GraphFNodeEighteen = \PT^\GraphFNodeEighteen = \emptyset$
and the only allocation in $\GraphMainNodeCall$ is not removed in the generalization
step from $\GraphMainNodeCallAbstraction$ to $\GraphFNodeOne$.
Further evaluation of $\GraphMainNodeIntersection$ yields a state
$\GraphMainNodeAfterStore$.
Here, $v_\codevar{idec}$ is the sum of \code{f}'s return value $\overline{v}_\codevar{rrec}$ and
the previous value $v_\codevar{i}$. 
There is a path from $\GraphMainNodeAfterStore$ back to $\GraphMainNodeCall$ and by
$(r_\codevar{rrec} < 0) \in  \KB^\GraphMainNodeIntersection$ (resulting from the return state
$\GraphFNodeEighteen$), this indicates that \code{i} is decremented in the loop.

%% file: graph_main.tex
\footnotesize
\newcommand{\ydist}{0.3cm}
\newcommand{\xdist}{0.5cm}
\newcommand{\setfont}{\scriptsize}
\begin{tikzpicture}[node distance = \ydist and \xdist]
\scriptsize
\def\widetwidth{7.4cm}

\tikzstyle{state}=[
                           font=\scriptsize,
                           draw]

\node[state, label=180:$\GraphMainNodeCall$, align=left, text width=\widetwidth] (Main-1) 
{
$[((\code{entry}, 0), \, \{\code{p} = \overline{v}_{\codevar{op}}\}),
\; (p_c, \{ \code{op} = v_\codevar{op}, ...\})],$\\
$\{  v_\codevar{op} = \overline{v}_{\codevar{op}}, v_\codevar{i} > 0, ... \},\allowbreak
    \{ \alloc{v_{\codevar{op}}}{v_{\codevar{op}}} \},$\\
    $\{ v_{\codevar{op}} \pointsto[\codevar{i8}] v_{\codevar{i}}  \}, 
     \{v_{\codevar{op}} \corrvar v_{\codevar{op}}, v_{\codevar{i}} \corrvar v_{\codevar{i}}, \ldots\}$
};

\node[state, label=270:$\GraphMainNodeCallAbstraction$, align=left, below=0.5cm of Main-1,xshift=-2cm] (Main-2) 
{
};

\node[state, label={[xshift=2pt, yshift=2pt]210:$\GraphFNodeOne$}, align=left, left=1.6cm of Main-2] (Main-3) 
{
};

\node[state, label=180:$\GraphMainNodeIntersectionUnsat$, align=left, below=0.44cm of Main-1, xshift=1cm] (Main-4) 
{
unsat.};

\node[ align=left, left=0.5cm of Main-3] (Main-3-empty) 
{
$\ldots$
};

\node[state, label=180:$\GraphMainNodeIntersection$, align=left, text width=\widetwidth, below=1.2cm of Main-1] (Main-5) 
{
$[((\code{rec}, 3), \, \{\code{p} = r_{\codevar{p}}, \code{pval} =
r_{\codevar{pval}}, \code{rrec} = r_{\codevar{rrec}}, ...\}),$\\$
\phantom{[}(p_c, \{  ...\})],$\\
    $\{  r_\codevar{pval} \geq 0, r_\codevar{rrec} < 0, v_\codevar{i} > 0, v_\codevar{i} = r_\codevar{pval}, ...  \},\allowbreak
    \emptyset, \allowbreak
    \emptyset, \{\ldots\}$
};

\node[state, label=180:$\GraphMainNodeAfterStore$, align=left, text width=\widetwidth, below=0.5cm of Main-5] (Main-6) 
{
$[(...)], \allowbreak
    \{  v_\codevar{idec} = \overline{v}_\codevar{rrec} + v_\codevar{i}, \overline{v}_\codevar{rrec} = r_\codevar{rrec},\allowbreak r_\codevar{pval} \geq 0, \allowbreak r_\codevar{rrec} < 0,\allowbreak v_\codevar{i} > 0, ...  \},\allowbreak
    \emptyset, \allowbreak
    \emptyset, \{\ldots\}$
};

\draw ($(Main-3)+(-15pt,0)$) ellipse (23pt and 11.5pt); 

\draw[ca-edge] (Main-1) -- node[anchor=east] {\tiny call abstraction $\;\;$} (Main-2);

\draw[ca-edge] (Main-2) -- node[anchor=north] {\tiny generalization} (Main-3);

\draw[eval-edge] (Main-3) --  (Main-3-empty);

\draw[ca-edge] ($(Main-1.south)+(5pt,0)$) -- node[anchor=west,xshift=5pt, 
align=left, font=\tiny] {intersection with $\GraphFNodeSeven$
} (Main-4);

\draw[ca-edge] ($(Main-1.south)+(-5pt,0)$) -- node[anchor=west,very near end, font=\tiny]
{ intersection with $\GraphFNodeEighteen$
}
($(Main-5.north)+(-5pt,0)$); 

\draw[omit-edge] (Main-5) -- node[anchor=east,xshift=-0pt] {\tiny  \code{ret}, \code{add}} (Main-6);

\draw[bend angle=30, bend right, omit-edge] (Main-6.east) to (Main-1.east);

\end{tikzpicture}

%% file: its.tex
\section{From SEGs to ITSs}
\label{sect:its}

Once we have a complete symbolic execution graph for the program under consideration, we
extract
\emph{integer transition systems} (ITSs) from its maximal cycles (i.e., from its
\emph{strongly connected components} (SCCs)\footnote{Here, $\GG$ is considered to be an
  \emph{SCC} if it is a maximal subgraph such that for all nodes $A, A'$ in $\GG$, $\GG$
  contains a non-empty path from $A$ to $A'$. So in contrast to the standard definition of
  SCCs, we also require that there must be a non-empty path from every node to itself.})
and apply existing techniques to prove their termination.
An ITS is a graph whose nodes are abstract states and whose edges are \emph{transitions}.
A transition is labeled with conditions that are required for its application.
We use the set $\Vsym$ to denote symbolic variables before applying a transition, and we
let the set
$\Vsym' = \{v' \mid v \in \Vsym\}$ denote the values of symbolic variables after the application of the
transition.
Note that in our SEGs, for all edge types except generalization edges, the
same variable occurring in two consecutive states denotes the same value.
Hence, in the ITSs resulting from SEGs, $v' = v$ holds for all transitions except those
that are obtained from generalization edges.

We use the same translation of symbolic execution graphs into ITSs that was presented in
\cite{LLVM-JAR}, since all new edge types introduced in this paper can be translated in
the same way as evaluation edges:
A non-generalization edge from $s$ to $\overline{s}$ in the SEG is transformed into a
transition with the condition $v' = v$
for all variables $v \in \Vsym(s)$.
In contrast, a generalization edge
from $s$ to $\overline{s}$
with the instantiation $\mu$ is transformed into a
transition with the condition $v' = \mu(v)$ for all
  $v \in \Vsym(\overline{s})$
to take the renaming of variables by $\mu$ into account.
Moreover, whenever a transition results from an edge from $s$ to $\overline{s}$,
we add $\stateformula{s}$ to the condition of the transition.

The only cycle of the SEG of  $\code{f}$
is from $\GraphFNodeOne$ to $\GraphFNodeThirteen$ back to $\GraphFNodeOne$
(see Fig.~\ref{fig:full-graph-function-f}), which corresponds to the recursive call of $\code{f}$.
The generalization edge from $\GraphFNodeThirteen$ to $\GraphFNodeOne$ results in a
condition $v_\codevar{*p}' = v_{\codevar{dec}}$, denoting that the value at the address $\code{p}$ is
decremented prior to each recursive call.
Due to $\models \stateformula{\GraphFNodeTen} \implies
v_\codevar{*p} \ge 0 \wedge v_\codevar{dec} = v_\codevar{*p} - 1$, existing termination techniques
easily show that the ITS corresponding to this cycle terminates.
This  implies termination for all \LLVM states that are represented in the SEG of
Fig.\ \ref{fig:full-graph-function-f}, i.e., this proves termination of the function
\code{f}.

Our new modular approach does not only allow us to re-use the SEGs for auxiliary functions
like \code{f} when they are called by other functions like \code{main}, but we also
benefit from this
modularity when extracting ITSs from the SCCs of the symbolic execution graph.
In the SEG for \code{main}, we have a path from the call state $V$ to the SEG of
\code{f}, but there is no path back from \code{f}'s SEG to \code{main}'s SEG  (see Fig.~\ref{fig:graph-main}).
Hence, the  SCCs of \code{main}'s graph do not contain any part of \code{f}'s
graph.\footnote{In contrast,  in
our previous technique for termination analysis of \LLVM from \cite{LLVM-JAR}, one
would obtain an SCC which contains both the cycles of \code{f}'s and of \code{main}'s SEG
and thus,  the ITS
corresponding to \code{f}'s SEG would have to be regarded again when proving termination of
\code{main}.}

Consequently, the resulting ITS for  \code{main} does not contain any rules of the ITS
for \code{f}, but just a rule that corresponds to the intersection edge from $V$ to
$Y$. This rule summarizes how $\KB$, $\AL$, and $\PT$ are affected by executing \code{f}.

Hence, if one has shown termination of \code{f} before, then to prove termination of
\code{main},
one just has to consider the only cycle  of \code{main}'s
SEG (from $\GraphMainNodeCall$ over $Y$
to $\GraphMainNodeAfterStore$ and back).
On the path from $\GraphMainNodeAfterStore$ back to $\GraphMainNodeCall$
 there is a generalization edge with an instantiation
$\tilde{\mu}$ such that
$\tilde{\mu}(v_\codevar{i}) = v_\codevar{idec}$
(i.e., the corresponding transition in the ITS has the conditions
$v_\codevar{i}' = v_\codevar{idec}$ and $\stateformula{\GraphMainNodeAfterStore}$).
Since we have $\models \stateformula{\GraphMainNodeAfterStore} \implies
v_\codevar{idec} < v_\codevar{i} \wedge v_\codevar{i} > 0 $, termination of the resulting ITS
is again easy to show by standard termination techniques.

As in \cite[Thm.\ 13]{LLVM-JAR},
our construction ensures that termination of the resulting
ITSs implies termination of the original program:

\setcounter{thm-termination}{\value{theorem}}

\begin{theorem}[Termination]
\label{thm:termination}
Let
$\Prog$ be an \emph{\LLVM} program with a complete symbolic execution graph
$\GG$ and let \(\I_1,\ldots,\I_m\) be the ITSs resulting from the SCCs of \(\GG\).
If all ITSs \(\I_1,\ldots,\I_m\) terminate,
then $\Prog$ also terminates for all concrete states $c$
that are represented by  a state of $\GG$.
\end{theorem}

%% file: evaluation.tex
\section{Implementation, Related Work, and Conclusion}
\label{sect:eval}

We developed a technique 
for automated termination analysis of \sfC{} (resp.\ \LLVM)
programs which models the memory in a byte-precise way.
In this paper, we showed how our technique 
can be improved into a modular approach. In this way, every function is analyzed
individually and its termination does not have to be re-proved anymore when it
is called by another function. This improvement also allows us to extend our approach to
the handling of recursive functions.

We implemented our approach  in our tool
\aprove{} \cite{AProVE-JAR}.
In \cref{Implementation Details} we  present
implementation details which we developed in order to improve the
analysis of large programs. After briefly describing the approaches of the other main tools for termination
analysis of \sfC{} programs at \emph{SV-COMP} in \cref{sect:relwork}, \cref{Experimental
  Evaluation}
gives an experimental
comparison with \aprove{}  based on the tools' performance at \emph{SV-COMP}
and
discusses directions for future work.

\subsection{Implementation Details}\label{Implementation Details}

Our approach is especially suitable for programs where a precise modeling of the variable
and memory contents are needed to prove termination. However, a downside of this high
precision is that it often takes long
to construct symbolic execution graphs, since \aprove{} cannot give any meaningful
answer before this construction is finished.
The more information we try to
keep in the abstract states, the more time is needed in every
symbolic execution step when inferring knowledge for the next state.
This results in a larger runtime than
that of many other tools for termination analysis.
Before developing the improvements of the current paper,
this used to result in many timeouts
when analyzing large programs with many function calls,
even if termination of the functions was not hard to prove
once the graph was constructed. For every function call,
an additional subgraph of the SEG was computed in the non-modular approach of \cite{LLVM-JAR}. This
did not only prohibit the handling of
recursive functions but also an efficient treatment of programs with several calls
of the same function. For example, this is the reason why
\aprove's analysis failed on all programs from
the \emph{product-lines} set, which is a part of the benchmarks in the \emph{Termination} category of
\emph{SV-COMP} since 2017. All terminating programs in this set consist of 2500-3800
lines of \sfC{} code. The corresponding \LLVM{} programs have 4800-7000 lines of code.

However, the novel approach of the current paper to analyze functions modularly is a big
step towards scalability.
Moreover, we developed several new heuristics to improve \aprove's performance on large
programs further. In this way,
\aprove's ability to
analyze large programs has increased significantly from year to year, see \cref{Experimental Evaluation}.

In the following, we outline the most crucial heuristics that have been implemented in \aprove{} until
\emph{SV-COMP 2019} in order to improve the handling of large programs.

\paragraph{Adapting the Strategy for Merging}
In \cite{LLVM-JAR}, we presented a strategy to decide when to merge abstract states.
There, merging was used to ensure that programs with loops still yield a finite
SEG.
However, merging can also be seen as a means of reducing the complexity of symbolic execution.
Merging two branches of the SEG and continuing symbolic execution from only the merged state
onwards can reduce the remaining number of required abstract states
significantly.

Since branching instructions lead to an exponential blowup of the
state space, for programs with a particularly high number of such instructions,
we use a more aggressive merging strategy.
It weakens some conditions on when states can be merged and then forces merging
of states that satisfy these weaker conditions.
Thus, we trade precision of the analysis for performance, by trying to obtain
SEGs with fewer states and fewer entries in their components.

When using the aggressive merging strategy, we change the conditions on when
states can be merged as follows:

\begin{itemize}
    \item
      Our original strategy for merging in \cite{LLVM-JAR} required that
        two states $s$ and $\overline{s}$ can only be merged if there is a path from $s$ to $\overline{s}$ in the
        symbolic execution graph.
        The reason was that the intention of merging is to \emph{guess} during an
        infinite path how this path eventually evolves in such a way that we keep all
        knowledge that is valid along this path (e.g., in each iteration of a loop) but remove
        all knowledge that only holds for a segment of this path (e.g., in a single
        iteration). For states of different paths, we did not see an advantage of
        merging these states and possibly losing information that is crucial to prove
        termination for the individual paths.

        However, for excessively branching functions,
        we want to force merging of different branches of their subgraph,
        even if there is no path connecting the involved states.
                Therefore, for those functions we drop the requirement that  there must
                always be a path between merged
        states.
    \item
         Normally,
        our merging heuristic requires merging candidates to have the same program
        variables in the $\LV$ functions of their corresponding stack frames.

        For example, this ensures that one does not merge states $s$ and $s_1$
        whose program position is at the beginning of a loop, where
 $s$ has not entered the loop yet whereas
        $s_1$ has executed the first iteration of the loop.
        This is because usually, $s_1$ contains extra program variables introduced in the body of
        the loop, and can therefore not be merged with $s$.
        Instead, we only merge $s_1$ with a successor $s_2$ that has iterated the loop
        body twice and has the same set of program variables.
        Indeed, it is preferable to merge only $s_1$ and $s_2$ rather than $s$ and $s_1$, because
        this
        results in more
        information preserved in the resulting generalized state (and this
        information can be crucial in order to prove termination of the loop).

        However, if the program is very large, then
        for other states $s$ and $\overline{s}$ that are not connected by a path in the
SEG, we lift the restriction that merging is only possible if the domains of the $\LV$
functions coincide.
Instead, we then allow to merge abstract states with different program variables
by intersecting their sets of program variables.

Again, this may result in a
loss of precision.
So if there are
variables which are only defined in $s$, but not in $\overline{s}$ and
thus, also not in the
state resulting from merging $s$ and $\overline{s}$,
then the merged state might lack some knowledge about the connection of the values of the
current program variables to the program variables at other positions.
However, the change to this more liberal merging heuristic
does not affect the applicability of our symbolic execution
rules. In other words,
it is still ensured that all
program variables are defined that are needed to evaluate the remaining instructions of
the program.
The reason is that the compilation of \sfC{} programs only results in \emph{well-formed}
        \LLVM programs,
        where it is guaranteed that in all possible executions,
        the instruction defining a variable dominates (i.e., precedes) any instruction
        using it.
        In particular, if there are
        different abstract states at the same program position
        in the SEG, then
        only those program variables can be accessed during subsequent executions that
        were defined on all incoming paths to this position.\footnote{The only \LLVM{} instruction
        that may use variables that have not been defined on all paths to the current
        position is the \code{phi} instruction. However, in our symbolic execution,
        this instruction is evaluated in combination with branching instructions and is
        never the position of an abstract state in the SEG, see \cite{LLVM-JAR}.}
\end{itemize}

\paragraph{Enforcing Unique Entry and Exit of Functions}
       In large programs,
      for each function $\rfunc$, we enforce that there is only a single SEG
      by merging all of its entry states to a unique one. Of course, this can mean that an
      auxiliary function $\rfunc$ may have to be analyzed again if the entry state of its current
      SEG is not general enough to cover a new call of $\rfunc$ in some other
      function. But the effect of enforcing a unique entry state for $\rfunc$
      is that the analysis becomes
        slightly more general each time, until we (hopefully) reach a version that is
        general enough for future uses.
        Although this prohibits
        specialized analyses for individual function calls in different contexts, this
        results in positive effects for symbolic execution of large programs since the
        components of the entry state contain fewer
        entries, which speeds up symbolic execution considerably.
  
        In \cref{Complete SEGs}, we remarked
    that similar return states of
      recursive functions have to be merged to obtain a finite SEG, analogous to
      the merging of states involved in loops.
      For functions that are not recursive,
      this is not necessary.
      However, for large programs, we try to minimize the
        number of return states. For this purpose, we merge all 
        return states at the same program position if their sets of
        defined program variables are identical.
        This reduces the number of pairs of call and return states
        for which we have to construct an intersection.

\paragraph{Removal of Unreachable Information from States}
To increase the performance of symbolic execution, 
we use additional heuristics to detect if certain information in
a state is most likely unnecessary and could be removed.

To this end, we determine for each symbolic variable in an abstract state $s$
whether it is \emph{reachable}.
A variable is reachable if it occurs in the range
  of any of the state's
$\LV$ functions. If a reachable variable occurs as a bound of an entry from $\ALL(s)$,
or in an entry from $\KB$, all other variables in
the same entry are marked as reachable, too.
If the variable $v_1$  of an entry  $(v_1 \pointsto[\codevar{ty}] v_2) \in \PT$ is
reachable and lies within an allocation
with a reachable bound, then $v_2$  becomes reachable, too.
Based on this, we extend the notion of reachability from variables to atoms in abstract states.
We  call entries 
from $\KB$, $\ALL(s)$, and $\PT$
reachable if all their variables are reachable.
Moreover, an entry $v \corrvar w$ from $\CV$ is considered to be reachable if $w$ is reachable.

To reduce the amount of information in the abstract states, we delete all unreachable entries from call abstraction
states.
This is useful, because 
many entries of the call abstraction may only have been relevant for the
lower stack frames that are no longer present. Nevertheless, 
removing unreachable entries might lose information (e.g., if $\PT^s$
has the entries $v \pointsto w_1$ and $v \pointsto w_2$
where $v$ is unreachable but $w_1, w_2$ are reachable, then $\stateformula{s}$ contains $w_1 = w_2$, whereas this
information is lost when deleting these entries from $\PT^s$).
Therefore, for all other states besides call abstractions, we do not remove all
unreachable entries, but we use a
contrived heuristic that decides 
which of the unreachable entries to delete.

\input{relatedwork}

\subsection{Experimental Evaluation and Future Work}
\label{Experimental Evaluation}

The focus of our approach is to analyze programs whose termination
depends on relations between addresses and memory contents, where the analysis requires
explicit low-level pointer arithmetic. \aprove{}'s  successful participation at
\emph{SV-COMP} and at the
\emph{Termination Competition}\footnote{\url{https://www.termination-portal.org/wiki/Termination_Competition}}
shows the applicability of our approach.

A command-line version of \aprove{} can be obtained from \cite{AproveWebsite}.
After installing all dependencies as described on this website, \aprove{} is invoked by the command
\begin{center}
	\code{java -ea -jar aprove.jar -m wst example.c}
\end{center}
to prove termination of the program \code{example.c}. Alternatively, \aprove{} can be accessed
via the web interface on the same website.
To run one of the versions submitted to \emph{SV-COMP}, the corresponding archive can be
downloaded from the competition website. Here, many of the dependencies are already included
in the archive. For example, for the version of 2019, only the \tool{Java Runtime
  Environment},
the \tool{Clang} compiler, and  \tool{Mono}  \cite{Mono} have to be installed.

\begin{figure}[t]
\hspace{0.4cm}
\begin{subfigure}[c]{0.36\textwidth}
\begin{tikzpicture}
\scriptsize
\begin{axis}[
    width=4.4cm,
    ymin=210,
    ymax=960,
    enlargelimits=0.4,
    bar width=5mm, y=0.0392mm,
    symbolic x coords={2017, 2018, 2019},
    xtick=data,
    x tick label style={rotate=45,anchor=east},
    nodes near coords align={vertical},
    legend style={draw=none,fill=none},
    legend columns=1,
    legend cell align={left},
    legend pos=north west,
    legend entries={\tool{UAutomizer},
                    \tool{CPA-Seq},
                    \tool{AProVE}},
    ]
\addlegendimage{only marks,line width=0.6pt,mark=o,mark size=2.2pt}
\addlegendimage{only marks,line width=0.6pt,mark=square,mark size=2.2pt}
\addlegendimage{only marks,line width=0.6pt,mark=x,mark size=2.2pt}
\addplot[ybar, nodes near coords, fill=black!10]
    coordinates {(2017,506) (2018,837) (2019,836)};
\addplot[line width=0.6pt,mark=x,mark size=2.2pt]
    coordinates {(2017,123) (2018,491) (2019,573)};
\addplot[line width=0.6pt,mark=o,mark size=2.2pt]
    coordinates {(2017,442) (2018,739) (2019,739)};
\addplot[line width=0.6pt,mark=square,mark size=2.2pt]
    coordinates {(2017,255) (2018,525) (2019,591)};
\end{axis}
\end{tikzpicture}
\vspace{-0.1cm}
\subcaption{\label{fig:eval-other}Programs of \emph{Termination-Other}}
\end{subfigure}
\hspace{0.7cm}
\begin{subfigure}[c]{0.5\textwidth}
\begin{tikzpicture}
\scriptsize
\begin{axis}[
    width=6.6cm,
    ymin=4,
    ymax=55,
    enlargelimits=0.15,
    bar width=5mm, y=0.8mm,
    symbolic x coords={2014, 2015, 2016, 2017, 2018, 2019},
    xtick=data,
    x tick label style={rotate=45,anchor=east},
    nodes near coords align={vertical},
    legend style={draw=none,fill=none},
    legend columns=1,
    legend cell align={left},
    legend pos=north west,
    legend entries={\tool{AProVE},
                    \tool{UAutomizer},
                    \tool{CPA-Seq},
                    \tool{SeaHorn},
                    \tool{HIPTNT+},
                    \tool{T2}},
    ]
\addlegendimage{only marks,line width=0.6pt,mark=x,mark size=2.2pt}
\addlegendimage{only marks,line width=0.6pt,mark=o,mark size=2.2pt}
\addlegendimage{only marks,line width=0.6pt,mark=square,mark size=2.2pt}
\addlegendimage{only marks,line width=0.6pt,mark=triangle,mark size=2.2pt}
\addlegendimage{only marks,line width=0.6pt,mark=star,mark size=2.2pt}
\addlegendimage{only marks,line width=0.6pt,mark=diamond,mark size=2.2pt}
\addplot[ybar, nodes near coords, fill=black!10]
    coordinates {(2014,22) (2015,28) (2016,35) (2017,55) (2018,55) (2019,55)};
\addplot[line width=0.6pt,mark=x,mark size=2.2pt]
    coordinates {(2014,0) (2015,28) (2016,32) (2017,33) (2018,40) (2019,39)};
\addplot[line width=0.6pt,mark=o,mark size=2.2pt]
    coordinates {(2014,8) (2015,22) (2016,24) (2017,37) (2018,36) (2019,37)};
\addplot[line width=0.6pt,mark=square,mark size=2.2pt]
    coordinates {(2017,0) (2018,0) (2019,13)};
\addplot[line width=0.6pt,mark=triangle,mark size=2.2pt]
    coordinates {(2016,0)};
\addplot[line width=0.6pt,mark=star,mark size=2.2pt]
    coordinates {(2015,23)};
\addplot[line width=0.6pt,mark=diamond,mark size=2.2pt]
    coordinates {(2014,17)};
\end{axis}
\end{tikzpicture}
\vspace{-0.1cm}
\subcaption{\label{fig:eval-rec}Recursive programs of other subcategories}
\end{subfigure}
\caption{\label{fig:eval}Number of termination proofs for leading tools in \emph{SV-COMP}}
\end{figure}
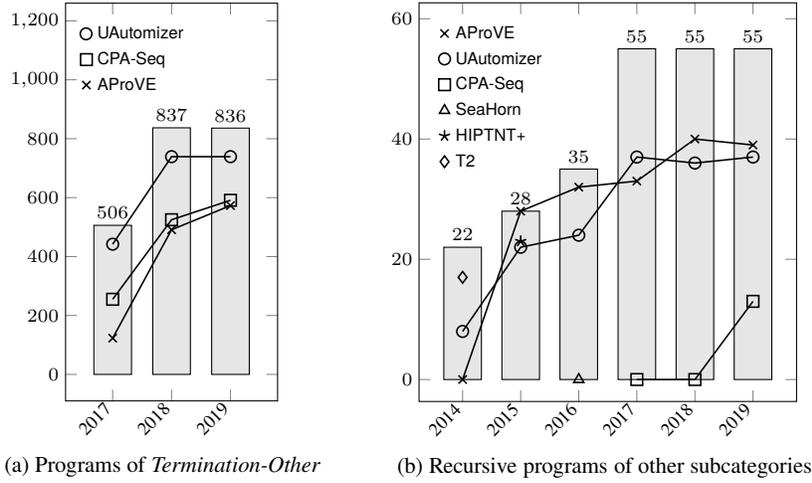

In the following, we evaluate the power of the new contributions of the paper.
To this end, we use the results that  \aprove{} and the other tools achieved at \emph{SV-COMP}.

\cref{fig:eval-other} shows the number of programs where termination was proved
for the three leading tools of the \emph{Termination} category of \emph{SV-COMP} in
\aprove{}'s weakest subcategory \emph{Termination-Other},
  which was introduced in 2017.
The bars in \cref{fig:eval-other} indicate the total number of terminating programs.
This subcategory mainly consists of large programs with significantly more function
calls and branching instructions than there are in the programs in the remaining two
subcategories. In particular, \emph{Termination-Other} includes
the \emph{product-lines} set, which contains 263 terminating
programs.
In 2017, \aprove{} already
performed well on smaller recursive programs, but this approach was not yet
generalized and optimized to use a modular analysis for non-recursive functions.
In the following two years, \aprove{} substantially reduced the relative gap to
the other leading tools for these kinds of examples.

\cref{fig:eval-rec} shows the number of recursive programs in the remaining two
subcategories of \emph{SV-COMP} where termination was proved.
Here, we give the numbers of successful proofs for the three leading tools of the
\emph{Termination} category per year. Again, the bars indicate the total number of
terminating recursive programs. Note that for most of the years, the set of programs
is a true superset of the set of programs of the previous year and the newly added
programs tend to be harder to analyze.
We see that first support to handle recursion
was already very successfully implemented in the \aprove{} version of 2015. In the
following years, this technique was further improved so that for most of the years,
\aprove{} was able to prove termination for more of these programs than the other
tools.

As mentioned, we could not submit \aprove{} to \emph{SV-COMP}  in
2020 and 2021 due to personal reasons, but we participated again in 2022 and 2025. The
three leading tools of 
the \emph{Termination} category of \emph{SV-COMP} 2020 were \tool{UAutomizer},
\tool{CPA-Seq}, and \tool{2LS}. However, \tool{UAutomizer} and \tool{CPA-Seq} did
not find more termination proofs for the programs in \cref{fig:eval-other}
and \cref{fig:eval-rec} than in 2019. \tool{2LS} was able to prove termination for
nearly as many programs as \tool{CPA-Seq} in \emph{Termination-Other}, but did not
find any termination proofs for the recursive programs in other subcategories.

Note that if we include \emph{non}-terminating recursive programs, \tool{UltimateAutomizer} is able
to give (non-)termination proofs for more recursive programs than \aprove{}.
The reason is that  although \aprove{} implements different approaches for disproving
termination, its focus is still on proving termination. The approach of
over-approximating all program runs using an abstraction that is suitable for analyzing large
programs often does not allow for an equivalent graph transformation where
non-termination of the resulting ITSs would imply non-termination of the original
program.

Apart from improving \aprove's capabilities for non-termination proofs, 
in future work we plan to extend our approach to handle recursive data structures.
Here, the main challenge is to create heap invariants that reason about the shape of
data structures and that abstract from their exact properties, but still
contain sufficient knowledge about the memory contents needed for the termination proof.
Similar to the approach in the current paper, this will require methods to remove and to restore
knowledge about allocations in the abstract states in order to validate memory safety.
Furthermore, these tasks have to combined with the handling of byte-precise pointer
arithmetic.

%% file: relatedwork.tex
\subsection{Related Work}
\label{sect:relwork}

The general approach of \aprove{} is closely related to abstract interpretation \cite{AbstractInt}.
In contrast to many other abstract interpretation approaches, however,
our abstract states may include arbitrary arithmetic terms (e.g., they can contain any
arithmetic expression 
arising from the conditions in the program).
Therefore, our symbolic execution 
starts with a rather precise abstraction, which is then coarsened during generalization
steps and call abstraction steps. This can be seen as a fixpoint computation to generate
an over-approximation of all possible program runs.

Our work is inspired by our earlier approach for modular termination analysis of recursive
\tool{Java Bytecode} programs \cite{RTA11}. However, since \cite{RTA11} 
handles \tool{Java}, it cannot analyze memory safety,
explicit allocation and deallocation of memory, and pointer arithmetic. Thus, the current paper
shows how to adapt such an approach for modular symbolic execution of possibly recursive
programs
to a byte-precise
modeling of the memory, as required for the analysis of languages like \sfC{}
or \LLVM.

Moreover, there are several further differences between the current approach and the
technique of \cite{RTA11} which also result in improved modularity. 
Recall that in the current paper,
when analyzing termination of a function \code{main},
we connect call states like $V$ (where \code{main} calls an auxiliary function \code{f})
with intersection states like $Y$ (which results from intersecting the call state
$V$ with the
return state $K$ of \code{f}). Moreover, there are paths from the
call states in \code{main}'s SEG to the SEG of \code{f}. However, there is no edge back from
\code{f}'s SEG to the SEG of \code{main}. Hence, the SEG of \code{f} is not part of the
cycles of \code{main}'s SEG.

As explained in \cref{sect:its},
this means that if one has proved termination of the
auxiliary function \code{f} before, then the ITSs for
\code{f} do not have to be regarded anymore
when proving termination of \code{main}. In contrast, this modularity is lacking in
\cite{RTA11}, because there,
instead of edges from
the call states in \code{main}'s SEG to the intersection states, there would be
edges from the
return states of the auxiliary function
\code{f} to the intersection states in \code{main}'s SEG.
(So in the graph of Fig.\ \ref{fig:graph-main}, instead of the edge from $V$ to $Y$, there would be an
edge from $K$ to $Y$.)
Hence, there the SEG of \code{f}
would become part of cycles in the SEG of \code{main},
i.e., there would be one SCC that contains both the cycles of \code{f}'s and \code{main}'s
SEG. Thus,
 the ITSs
corresponding to \code{f}'s SEG would have to be regarded again when proving termination of
\code{main}.

There exist many approaches and tools for proving and disproving termination of \sfC{} programs, e.g., 
besides our own tool \aprove, the leading termination analysis tools at \emph{SV-COMP} 2014-2020 were
\tool{UltimateAutomizer} \cite{Ultimate},
  \tool{CPA-Seq} (based on \tool{CPAchecker}
\cite{CPA}), \tool{HIPTNT+} \cite{HIPTNT+}, \tool{SeaHorn} \cite{SeaHorn},
\tool{T2} \cite{T2}, and \tool{2LS} \cite{2LS}.
In the following, we give a brief overview of other termination analysis approaches,
in particular for handling modularity and recursion.

All of the tools mentioned above apply abstractions to reduce the state space when
analyzing (non-)termination. While our approach is based on a symbolic execution of the
program on abstract states,
\tool{UltimateAutomizer} uses an automata-based approach, whose key idea is to build
Büchi automata that accept all non-terminating traces of the program. Then, an
emptiness check either proves termination or yields an infinite trace that serves
as a (potentially spurious) counterexample for termination.
If spurious, a proof for its infeasibility is constructed using an inductive sequence
of interpolants from the error trace.
This proof is then generalized in order to exclude as many unfeasible traces as
possible. For an interprocedural analysis, so-called nested word automata are used,
which model the nesting of functions and use nested interpolants \cite{NestedInterpolants}
to exclude spurious traces. In this way, \tool{UltimateAutomizer} also handles recursion.

Counterexample-guided abstraction refinement is also used by
\tool{CPAchecker} but in a different setting. Here, an abstract reachability tree is
constructed, which unfolds the control flow graph. The edges of the tree correspond
to instructions of the program. The abstraction starts at a coarse level and is refined
whenever a spurious counterexample is found. To re-use effects of
functions that have already been analyzed before, \tool{CPAchecker} uses block
abstraction memoization, computing separate abstract reachability trees for
individual function bodies, if they are called. Whenever the same function is called
again, the function tree can be re-used if the
function's locally relevant variables are the same in the context of the current
abstract state.
Similar to \tool{UltimateAutomizer}, this approach has been extended to recursion using nested
interpolation for recursive function calls \cite{CPA-Rec}. While \aprove's strength is the
handling of programs whose termination depends on explicit heap operations, 
\tool{CPAchecker} is particularly powerful for large programs.

\tool{SeaHorn} incrementally synthesizes a ranking function candidate by asking a safety
verifier for counterexamples to non-termination. As long as terminating executions
are found that do not yet adhere to the candidate function, it is refined. Ultimately,
the candidate is either validated as an actual ranking function or non-termination is
implied.
To treat functions modularly, \tool{SeaHorn} constructs summaries for functions and
re-uses computed information.
To our knowledge, however, there is no support for recursive functions yet.

\tool{HIPTNT+} analyzes termination of the underlying program 
on a
per-method basis to obtain a modular analysis. Similar to our approach,
\tool{HIPTNT+} uses separation logic to express properties of the heap.
Each method is annotated with
a specification using predicates that is incrementally refined by case analyses. 
In this way, summaries of (non-)termination characteristics in the specification 
are derived and can be re-used every time a function is called within another
function.

\tool{T2} invokes an extended version of \tool{llvm2kittel} \cite{llvm2kittel} to translate
\sfC{} programs into ITSs. Then, termination of these ITSs is analyzed using techniques
that are also implemented in \aprove's back-end. 
While \aprove{} always tries to prove termination of \emph{all} runs of an ITS,
\tool{T2}
supports the termination analysis for ITSs where all runs begin with dedicated start
terms. For that reason, \tool{T2}
can also prove non-termination of ITSs (and therefore, \aprove{} uses \tool{T2} instead of
its own ITS-back-end  when trying to prove non-termination of \sfC{} programs). 
On the other hand, \tool{T2} does not
model the heap. Instead, it treats read accesses as loading non-deterministic values and
simply ignores write accesses. 

\tool{2LS} focuses on non-recursive programs with several functions.
It proves termination by an over-approximating forward analysis using templates
over bitvectors to synthesize linear lexicographic ranking functions. In order to handle
heap-allocated data structures, it uses a template domain for shape analysis.
Interprocedural summarization enables a modular analysis of
large programs that do not contain recursive functions.

%% file: appendix.tex
\setcounter{section}{0}
\renewcommand{\thesection}{\Alph{section}}
\section{Separation Logic Semantics of Abstract States}
\label{app:semantics-sl}

In order to formalize which concrete states are represented by an abstract state $s$, we
introduced a \emph{separation logic}  formula $\stateformula{s}_{\SL}$ in \cite{LLVM-JAR}.
It extends $\stateformula{s}$ by further information about the memory, in order to
define which concrete states are represented by an (abstract) state.

 First, we define
the semantics of the fragment of separation logic used. In this fragment, first-order
logic formulas are extended by ``$\pointsto$'' for information from $\PT$.
We employ the usual semantics of the ``$*$'' operator, i.e., $\varphi_1 * \varphi_2$ means that
$\varphi_1$ and $\varphi_2$ hold for different parts of the memory. 

We use \emph{interpretations} $(\slassignment,\slmemory)$ to determine the semantics of
separation logic. Let  $\Idsi = \{ \code{x}_i \mid \code{x} \in \Ids, i \in \N_{> 0} \}$
be the set of all indexed program variables that we use to represent stack
\underline{fr}ames.
The function $\slassignment : \Idsi \to \Z$  assigns values to the
program variables, augmented with a stack index. The function 
$\slmemory : \N_{> 0} \partialfunctionmap \{0,\ldots,2^8-1\}$ describes the
memory contents at allocated addresses as unsigned bytes.
In the following, we also consider possibly non-concrete instantiations $\sigma : \Vsym \to
\mathcal{T}(\Vsym)$, where $\mathcal{T}(\Vsym)$ are all arithmetic terms containing only
variables from $\Vsym$.

\begin{definition}[Semantics of Separation Logic] 
\label{def:semantics-of-sl}
Let $\slassignment : \Idsi \to \Z$, $\slmemory : \N_{> 0} \partialfunctionmap \{0,\ldots,2^8-1\}$,
	and let $\varphi$ be a formula. 
Let $\slassignment(\varphi)$ result from replacing all $\code{x}_i$ in $\varphi$ by the
value $\slassignment(\code{x}_i)$. 
Note that by construction, local variables $\code{x}_i$ are never quantified in our formulas.
Then we define $(\slassignment,\slmemory) \models \varphi$ iff $\slmemory \models \slassignment(\varphi)$.

We now define $\slmemory \models \psi$ for formulas $\psi$ that may contain symbolic variables from $\Vsym$.
As usual, all free variables $v_1,\ldots,v_n$ in $\psi$ are implicitly universally
quantified, i.e., $\slmemory \models \psi$ iff  
$\slmemory \models \forall v_1,\ldots, v_n. \, \psi$.

The semantics of arithmetic operations and predicates as well as of first-order
connectives and quantifiers are as usual. 
In particular, we define $\slmemory \models \forall v. \, \psi$ iff $\slmemory \models
\sigma(\psi)$ holds for all instantiations $\sigma$ where $\sigma(v) \in \Z$ 
	and $\sigma(w) = w$ for all $w \in \Vsym \setminus \{v\}$.

The semantics of $\pointstosl$ and $*$ for variable-free formulas are as follows:
For $n_1,n_2 \in \Z$, let 
$\slmemory \models n_1 \pointstosl n_2$ hold iff $\slmemory(n_1) = n_2$.\footnote{We use 
  ``$\pointstosl$'' instead of ``$\mapsto$'' in separation logic, since 
  $\slmemory \models n_1 \mapsto n_2$ would imply that $\slmemory(n)$ is undefined for all 
  $n \neq n_1$. This would be inconvenient in our formalization, since $\PT$ 
  usually only contains information about a \emph{part} of the allocated 
  memory.
}

The semantics of $*$ is defined as usual in separation logic: For two partial 
functions $\slmemory_1,\slmemory_2 : \N_{>0} \partialfunctionmap \Z$, we write
$\slmemory_1 \bot \slmemory_2$ to indicate that  
the domains of $\slmemory_1$ and $\slmemory_2$ are disjoint. If $\slmemory_1 \bot \slmemory_2$, then
$\slmemory_1 \uplus \slmemory_2$ denotes the 
union of $\slmemory_1$ and $\slmemory_2$.
Now $\slmemory \models \varphi_1 * \varphi_2$ holds iff there 
exist $\slmemory_1 \bot \slmemory_2$ such that $\slmemory = \slmemory_1 \uplus \slmemory_2$ where 
$\slmemory_1 \models \varphi_1$ and $\slmemory_2 \models \varphi_2$.
We define the empty separating conjunction to be $\mathit{true}$, i.e.,
$\bigast\nolimits_{\varphi \in \AL} \, \stateformula{\varphi}_\SL \; = \; \mathit{true}$
if $\AL = \varnothing$.
\end{definition}

We now define the formula $\stateformula{s}_{\SL}$ for a state $s$.
In $\stateformula{s}_\SL$, the elements of $\AL$ are combined with the \emph{separating conjunction} ``$*$'' to
express that different allocated memory blocks are disjoint. In contrast, the elements of
$\PT$ are combined by the ordinary conjunction ``$\wedge$''. 
This is due to the fact that $\PT$ may contain entries $v_1 \pointsto[\codevar{ty_1}]
v_2$, $w_1 \pointsto[\codevar{ty_2}] w_2$ referring to overlapping parts of the
memory. Similarly, we also combine the two formulas resulting from $\AL$ and $\PT$ by
``$\wedge$'', as both express different properties of the same addresses. 
Recall that
we identify \emph{sets} of first-order formulas $\{\varphi_1,..., \varphi_n\}$
with their conjunction $\varphi_1 \wedge ... \wedge \varphi_n$ and $\CS$ with the
set resp.\ with the conjunction of the equations
$\bigcup_{1 \leq i \leq n} \{ \code{x}_i = \LV_i(\code{x}) \mid \code{x} \in \Ids, \LV_i(\code{x}) \text{ is defined}\}$.
As in  \cref{sect:seg}, for any type \codevar{ty}, $\sizeOf(\codevar{ty})$ denotes the size of \codevar{ty} in
bytes.

\begin{definition}[$\SL$ Formulas for States]
\label{def:StateFormula}
For $v_1,v_2 \in \Vsym$, let  $\stateformula{\alloc{v_1}{v_2}}_\SL =(\forall x. \exists
y. \; (v_1 \leq x \leq v_2) \implies (x \pointstosl y))$. 
In order to reflect the two's complement representation, for any \LLVM{} type $\codevarx{ty}$ we define
\mbox{\small $\stateformula{v_1 \hookrightarrow_{\codevarx{ty}} v_2}_\SL =$}
\vspace*{-2mm}
\[
\stateformula{v_1 \hookrightarrow_{\mathit{size}(\codevar{ty})} v_3}_\SL
\; \wedge \;
(v_2 \geq 0 \, \Rightarrow \, v_3 = v_2) \;
\wedge \;
(v_2 < 0 \, \Rightarrow \, v_3 = v_2 + 2^{8 \cdot \mathit{size}(\codevarx{ty})}),
\]
where $v_3 \in \Vsym$ is fresh. 
We assume a little-endian data layout (where least significant bytes are stored in the lowest address).
Hence, we let $\stateformula{v_1 \hookrightarrow_0 v_3}_\SL = \mathit{true}$ and
	 $\stateformula{v_1 \hookrightarrow_{n+1} v_3}_\SL = (v_1 \hookrightarrow (v_3 \;
\mathrm{mod} \; 2^8)) \; \wedge \; \stateformula{\,(v_1+1) \hookrightarrow_{n}
  (v_3\;\mathrm{div}\;2^8)\,}_\SL$. 

A state $s = (\CS,\KB,\AL,\PT)$  is then represented in separation logic by
\[\stateformula{s}_\SL = \stateformula{s} \; \wedge \; \CS \;  \wedge \;
(\bigast\nolimits_{\varphi \in \AL^*(s)} \; \; \stateformula{\varphi}_\SL)
\; \wedge \;
(\bigwedge\nolimits_{\varphi \in \PT} \; \; \stateformula{\varphi}_\SL).\]
\end{definition}

For any abstract state $s$ we have $\models \stateformula{s}_\SL \implies
\stateformula{s}$, i.e., $\stateformula{s}$ is a weakened version of $
\stateformula{s}_\SL$. 
As mentioned, we use $\stateformula{s}$ for the construction of the symbolic execution
graph, enabling standard first-order SMT solving to be used for all reasoning required in
this construction. The separation logic formula $\stateformula{s}_\SL$ is only needed to
define when a concrete state $c$ is \emph{represented} by an abstract state $s$. As
stated in Def.\ \ref{def:representation-same-stack-size} this is the case if 
 $(\slassignment^c,\slmemory^c)$ is a model of
$\sigma(\stateformula{s}_\SL)$ and for each allocation of $s$ there exists a corresponding allocation
in $c$ of the same size. Here, from every concrete state $c$ one can extract an
interpretation $(\slassignment^c,\slmemory^c)$ as follows.

\begin{definition}[Interpretations $\slassignment^c$, $\slmemory^c$]
Let $c \neq \ERROR$ be a concrete state.
For every $\codevar{x}_i \in \Idsi$ where $\codevar{x} \in \domain(\LVi^{c})$, let
$\slassignment^c(\codevar{x}_i) = n$ for the number $n \in \Z$ with $\models
\stateformula{c} \implies \LVi^c(\codevar{x}) = n$.

For $n \in \N_{>0}$, the  function $\slmemory^{c}(n)$ is defined iff there exists a $(w_1
\pointsto[\codevar{i8}] w_2) \in \PT^c$  such that $\models  \stateformula{c}  \implies  w_1
= n$.  
Let $\models \stateformula{c} \implies  w_2 = k$ for  $k \in [-2^7,2^7-1]$.
Then we have $\slmemory^c(n) = k$ if $k \geq 0$ and $\slmemory^c(n) = k + 2^8$ if $k < 0$.
\end{definition}

%% file: proofs.tex
\section{Proofs}
\label{app:proofs}

This appendix contains all proofs for the results of the paper.


\setcounter{auxctr}{\value{theorem}}
\setcounter{theorem}{\value{thm-soundness-graph-construction-main-text}}

\begin{theorem}[Soundness of the Symbolic Execution Graph]
Let $\pi = c_0 \ccto c_1 \ccto c_2 \ccto \dots $ be a (finite resp.\ infinite)
\emph{\LLVM} evaluation of concrete states such that $c_0$ is represented by some 
        state $s_0$ in  a weakly complete SEG $\GG$. 
	Then there exists a (finite resp.\ infinite) sequence of states $s_{0},  s_{1},  s_{2},
        \dotsc$ where $\GG$ has an edge from $s_{j-1}$ to $s_j$ if $j > 0$,
        and  there exist $0 = i_0 \leq i_1 \leq  \ldots$ with $c_{i_j} \repby s_j$
        for all $j \geq 0$.
       Moreover, if $\pi$ is infinite then the  corresponding
        sequence of abstract states in $\GG$
        is infinite as well.
        In contrast, if $\pi$ is finite and ends at some concrete state $c$, then the
        sequence of states in $\GG$ ends at some state $s$ with $c 
        \repby s$.       
\end{theorem}
\begin{proof}
  The
 corresponding theorem in \cite{LLVM-JAR} did not reason about paths
but about single concrete evaluation steps. It stated that for a concrete state $c$
that is represented by an abstract state $s$ in $\GG$, $c \cto \nxt{c}$
implies that there is a path from $s$ to an abstract state $\nxt{s}$ in $\GG$
such that $\nxt{c}$ is represented by $\nxt{s}$. Intuitively, each concrete
evaluation step is simulated by an evaluation edge during symbolic execution, while
generalization and refinement edges do not correspond to a concrete evaluation step.
Therefore, we argued that if $s$ has an outgoing evaluation edge,
then its direct successor
$\nxt{s}$ represents $\nxt{c}$.
In contrast, if $s$ has an outgoing generalization edge,
then the generalized state also represents $c$, and if $s$ has outgoing refinement
edges, then one of the direct successors of $s$ represents $c$. In the latter case,
the next step in the graph is an evaluation which yields a state $\nxt{s}$ that represents
$\nxt{c}$. In case of a
generalization, there may be a refinement step before $\nxt{s}$ is computed by
evaluating an instruction. This is illustrated in Fig.\ \ref{OldProofFig}. 

\newcommand{\nodedist}{1.5cm}
\newcommand{\vertdist}{-1.5cm}
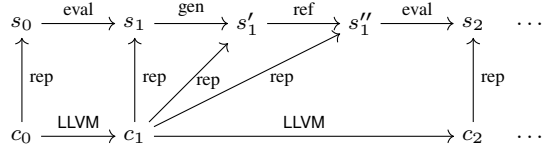
\begin{figure}[t]
\begin{center}
  \begin{tikzpicture}
    \node[] at (0,0) (s0) {$s_0$};
    \node[] at ($(s0)+(\nodedist,0)$) (s1) {$s_1$};
    \node[] at ($(s1)+(\nodedist,0)$) (s1') {$s_1'$};
    \node[] at ($(s1')+(\nodedist,0)$) (s1'') {$s_1''$};
    \node[] at ($(s1'')+(\nodedist,0)$) (s2) {$s_2$};
    \node[] at ($(s2)+(\nodedist/2,0)$) (sdots) {$\ldots$};
    
    \node[] at ($(s0)+(0,\vertdist)$) (c0) {$c_0$};
    \node[] at ($(c0)+(\nodedist,0)$) (c1) {$c_1$};
    \node[] at ($(c1)+(3*\nodedist,0)$) (c2) {$c_2$};
    \node[] at ($(c2)+(\nodedist/2,0)$) (cdots) {$\ldots$};
    
    \path[->] (s0) edge node[above] {\scriptsize{eval}} (s1);
    \path[->] (s1) edge node[above] {\scriptsize{gen}} (s1');
    \path[->] (s1') edge node[above] {\scriptsize{ref}} (s1'');
    \path[->] (s1'') edge node[above] {\scriptsize{eval}} (s2);
    
    \path[->] (c0) edge node[above] {\scriptsize{\LLVM}} (c1);
    \path[->] (c1) edge node[above] {\scriptsize{\LLVM}} (c2);
    
    \path[->] (c0) edge node[right] {\scriptsize{rep}} (s0);
    \path[->] (c1) edge node[right] {\scriptsize{rep}} (s1);
    \path[->] (c1) edge node[right] {\scriptsize{rep}} (s1');
    \path[->] (c1) edge node[right=0.2cm] {\scriptsize{rep}} (s1'');
    \path[->] (c2) edge node[right] {\scriptsize{rep}} (s2);
  \end{tikzpicture}
\end{center}
\caption{\label{OldProofFig}Relation between evaluation in \LLVM{} and paths in the SEG in \cite{LLVM-JAR}}
\end{figure}

In the present paper,
soundness of the evaluation rules, the generalization rule, and the refinement rule
follows from the proof in \cite{LLVM-JAR}. There are only two modifications that
we have to consider. First, we have the new state component $\CV$. However, this component
does not have any impact on the formula representation of states or on the representation
relation, and therefore it does not change the proof.
Second, we have the notion of \emph{weak} representation in our new approach and thus,
also in Thm.\ \ref{thm:soundness-graph-construction-main-text}. 
However, it is easy to see that this does not affect the proof:
\begin{itemize}
  \item For all evaluation rules except the \code{call} and the \code{ret} instruction,
    symbolic execution is only affected by the lower stack frames due to the allocations
    of those frames and the corresponding entries in $\PT$.
    However, \emph{which} frame an allocation belongs to
    has no effect on the symbolic execution.
    Furthermore, for $\PT$ entries, the states do not even contain the
    information on their corresponding stack frames. Therefore, for all instructions except \code{call} and \code{ret},
    applying our symbolic execution rules to a state and then creating its context
    abstraction of size $k$ results in the same result as first creating the context abstraction
    of size $k$ of the original state and then applying the symbolic evaluation rules to
    the context abstraction.
  \item Symbolically evaluating the \code{call} instruction on an abstract state $s$
    creates a new topmost stack frame corresponding to the new
     concrete stack frame that is created when evaluating \code{call} on the corresponding
     concrete state $c$. Again, the stack frame below the newly created frame
    is the only one that has an impact on the individual state components.
  \item The \code{ret} instruction pops the first stack frame. Thus, the second stack frame becomes
    the new topmost frame. Since the corresponding symbolic execution rule requires the second stack frame
    to be present in the abstract state, possibly missing stack frames due to context abstraction
    do not have an impact on the execution result.
\end{itemize}

Soundness of call abstraction follows from the fact that the call abstraction $s_{ca}$ of
an abstract state $s_c$ is more general than $s_c$, i.e., we do not have any additional
knowledge in $s_{ca}$ but instead we may lose knowledge from $s_c$ by abstracting from all but the
topmost stack frame. Therefore, it is trivial that any concrete state that is weakly
represented by $s_c$ is also weakly represented by $s_{ca}$.

Finally, we have to prove soundness of intersections. This is a special case since
intersection edges are the only edges that represent more than one concrete evaluation
step. The corresponding concrete steps are, however, represented by the path from the
call state $s_c$ to the return state $s_r$ that is used to create the intersection $s_i$.
This is illustrated in Fig.\ \ref{NewProofFig}.

\begin{figure}[t]
\begin{center}
  \begin{tikzpicture}
    \node[] at (0,0) (s) {$s$};
    \node[] at ($(s)+(\nodedist,0)$) (sc) {$s_c$};
    \node[] at ($(sc)+(\nodedist,0)$) (sca) {$s_{ca}$};
    \node[] at ($(sca)+(\nodedist,0)$) (se) {$s_e$};
    \node[] at ($(se)+(\nodedist,0)$) (sr) {$s_r$};
    \node[] at ($(sr)+(\nodedist,0)$) (si) {$s_i$};
    \node[] at ($(si)+(\nodedist/2,0)$) (sdots) {$\ldots$};
    \node[] at ($(s)-(\nodedist/2,0)$) (sdots) {$\ldots$};
    
    \node[] at ($(s0)+(0,\vertdist)$) (c) {$c$};
    \node[] at ($(c)+(\nodedist,0)$) (cc) {$c_c$};
    \node[] at ($(cc)+(3*\nodedist,0)$) (cr) {$c_r$};
    \node[] at ($(cr)+(\nodedist/2,0)$) (cdots) {$\ldots$};
    \node[] at ($(c)-(\nodedist/2,0)$) (cdots) {$\ldots$};
    
    \path[->] (s) edge node[above] {\scriptsize{\texttt{call}}} (sc);
    \path[->] (sc) edge node[above] {\scriptsize{call abs.}} (sca);
    \path[->] (sca) edge node[above] {\scriptsize{gen}} (se);
    \path[->] (se) edge[dotted] (sr);
    \draw[->,rounded corners=8] (sc) -- ($(sc)+(0,0.6)$) -- ($(si)+(0,0.6)$) -- (si);
    \node[] at ($(3*\nodedist,0.82)$) (text) {\scriptsize{intersection}};
    
    \path[->] (c) edge node[above] {\scriptsize{\LLVM}} (cc);
    \path[->] (cc) edge[dotted] node[above] {\scriptsize{\LLVM}} (cr);
    
    \path[->] (c) edge node[right] {\scriptsize{$\repby$}} (s);
    \path[->] (cc) edge node[right] {\scriptsize{$\repby$}} (sc);
    \path[->] (cc) edge node[right] {\scriptsize{$\repby$}} (sca);
    \path[->] (cc) edge node[right=0.2cm] {\scriptsize{$\repby$}} (se);
    \path[->] (cr) edge node[right] {\scriptsize{$\repby$}} (sr);
    \path[->] (cr) edge node[right] {\scriptsize{$\repby$}} (si);
  \end{tikzpicture}
\end{center}
\caption{\label{NewProofFig}Relation between evaluation in \LLVM{} and paths in the SEG for intersections}
\end{figure}
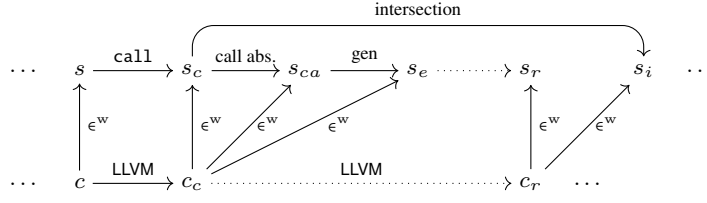

Hence, we now prove that if $c_c \repby s_c$ for a call state $s_c$,
the execution of the function in $c_c$'s topmost stack frame terminates in $c_r$,
 and  $c_r \repby s_r$
for a corresponding return state $s_r$, then we also have $c_r \repby s_i$ for the intersection
$s_i$ of $s_c$ and $s_r$. In the following, let $\conttop{c}_r$ be the context abstraction of size
$|s_i|$ of $c_r$.
To show that  $c_r \repby s_i$ holds, we prove that  $\conttop{c}_r$ is represented by
$s_i$. To this end, we have to check the requirements imposed by Def.\ \ref{def:representation-same-stack-size}. 

Since $c_r \repby s_r$,  the program position
and the domains of the local variables correspond to each other
in the topmost stack frame of $c_r$ and $s_r$. Therefore, they also correspond to each other 
in $\conttop{c}_r$ and $s_i$, since the program position and the domains of the local
variables are equal in $s_r$ and in the topmost stack frame of $s_i$.

All lower stack frames do not change between the concrete call state $c_c$  and the
concrete return state $c_r$ of the same
function since the topmost stack frame is never returned during this part of the
evaluation.
Therefore, due to $c_c \repby s_c$ we
have that all lower stack frames of $\conttop{c}_r$ (which are also lower stack frames of
$c_c$)
have the same program positions and the same domains of the local variables as
the
lower stack frames of $s_i$ (which are also lower stack frames of $s_c$).

For the third condition of Def.\ \ref{def:representation-same-stack-size}, 
since the allocation list of $s_i$'s topmost stack frame is empty by Def. \ref{def:intersection}, 
we do not require any
corresponding allocations in the topmost stack frame of $\conttop{c}_r$.
The stack allocations in the lower stack frames of $s_i$ are the same as the allocations
in the lower stack frames of $s_c$. Hence,  $c_c \repby s_c$ again implies that these
lower stack frames are also represented in $\conttop{c}_r$ (note that the context
abstraction can only increase the number of stack allocations in the stack frames).

Hence, to prove that  $\conttop{c}_r$ is represented by
$s_i$, it remains to show that the second condition of
Def.\ \ref{def:representation-same-stack-size} holds. So we have to show that
\begin{equation}
  \label{wr:model} \parbox{6.5cm}{$(\slassignment^{\conttop{c}_r},\slmemory^{\conttop{c}_r})$ is a model of
        $\sigma(\stateformula{s_i}_\SL)$\\for some concrete instantiation $\sigma : \Vsym \to
      \Z$.}
\end{equation}

To prove \eqref{wr:model}, we have to show that $\slmemory^{\conttop{c}_r}$ is a model of all of
the following subformulas.
\begin{enumerate}[(a)]
  \item $\slassignment^{\conttop{c}_r}(\sigma(\CS^{s_i}))$,
  \item $\sigma(\bigast\nolimits_{\varphi \in \AL^*(s_i)} \, \stateformula{\varphi}_\SL)$,
  \item $\sigma(\bigwedge \nolimits_{\varphi \in \PT^{s_i}} \, \stateformula{\varphi}_\SL)$,
  \item $\sigma(\KB^{s_i})$
  \item $\sigma(\{ 1 \leq v_1 \land v_1 \leq v_2 \mid \alloc{v_1}{v_2} \in \AL^*(s_i) \})$
  \item $\sigma(\{ v_2 < w_1 \lor w_2 < v_1 \mid \alloc{v_1}{v_2}, \alloc{w_1}{w_2} \in \AL^*(s_i), \; (v_1,v_2) \neq (w_1,w_2)\})$
  \item $\sigma(\{1 \leq v_1 \mid (v_1 \hookrightarrow_{\codevar{ty}} v_2) \in \PT^{s_i} \})$
  \item $\sigma(\{v_2 = w_2 \mid (v_1 \hookrightarrow_{\codevar{ty}} v_2), 
   (w_1 \hookrightarrow_{\codevar{ty}} w_2) \in \PT^{s_i} \mbox{ and } \models \, 
   \stateformula{s_i} \implies v_1 = w_1\})$
  \item $\sigma(\{v_1
  \neq w_1 \mid (v_1 \hookrightarrow_{\codevar{ty}} v_2),  
   (w_1 \hookrightarrow_{\codevar{ty}} w_2) \in \PT \mbox{ and } \models \, 
  \stateformula{s_i} \implies v_2 \neq w_2\})$
\end{enumerate}
We first define how to choose $\sigma$ and then show why $\slmemory^{\conttop{c}_r}$
is a model of the individual subformulas.
Since $c_r \repby s_r$, there exists an instantiation $\sigma_r$ that assigns a concrete value
to each symbolic variable in $s_r$ and thereby yields the context abstraction
$\conttop{c}'_r$
of size $|s_r| = 1$
of $c_r$
(i.e., $(\slassignment^{\conttop{c}'_r}, \slmemory^{\conttop{c}'_r}) \models \sigma_r(\stateformula{s_i}_\SL)$).
Similarly, since $c_c \repby s_c$, there exists an instantiation $\sigma_c$ with the
same property for $s_c$ and $c_c$. Then, we choose
$$\sigma = (\sigma_r \circ \renamed^{-1}) \circ \sigma_c\text{,}$$
where $\renamed$ is the function that renames symbolic variables from $s_r$ to create $s_i$.
Note that the domains of $(\sigma_r \circ \renamed^{-1})$ and $\sigma_c$ are disjoint since the range of
$\renamed$ only contains fresh variables.
\begin{enumerate}[(a)]
    \item We have $\CS^{s_i} = (\Pos_1^{s_r},\renamed(\LV_1^{s_r}),\emptyset) \cdot
      \widetilde{\CS}^{s_c}$,
where $\widetilde{\CS}^{s_c}$ is the call stack of $s_c$ without its topmost frame.       
For the topmost stack frame
of  $\CS^{s_i}$, we have the same assignment of program variables as in $s_r$
and we have $\sigma = 
\sigma_r \circ \delta^{-1}$ for variables in the range of $\delta$.
So since $c_r \repby s_r$, for every program variable $\code{x} \in \Ids$ where
$\LV_1^{s_r}$ is defined, we have $\slassignment^{\conttop{c}'_r}(\code{x}_1) =
\sigma_r(\LV^{s_r}_1(\code{x}))$. Thus, we also get
$\slassignment^{\conttop{c}_r}(\code{x}_1) =
\slassignment^{\conttop{c}'_r}(\code{x}_1) =
\sigma_r(\LV^{s_r}_1(\code{x})) = 
(\sigma_r \circ \delta^{-1}) (\delta(\LV^{s_r}_1(\code{x})))
= \sigma(\delta(\LV^{s_r}_1(\code{x})))$.

Similarly,  $c_c \repby s_c$ implies that for the corresponding context abstraction
$\conttop{c}_c$ of $c_c$ we have $\slassignment^{\conttop{c}_c}(\code{x}_i) =
\sigma_c(\LV^{s_c}_i(\code{x}))$ for $i \geq 2$. As the lower stack frames of $c_c$ are
not modified during the evaluation from $c_c$ to $c_r$,
we have
$\slassignment^{\conttop{c}_r}(\code{x}_i) =
\slassignment^{\conttop{c}_c}(\code{x}_i) =
\sigma_c(\LV^{s_c}_i(\code{x})) = 
\sigma(\LV^{s_c}_i(\code{x}))$ for $i \geq 2$.
      \item We have to show that if a concrete address in $\conttop{c}_r$ corresponds to an allocation
        from $\AL^*(s_i)$, then it is mapped to a value by $\slmemory^{\conttop{c}_r}$.
        In the topmost stack frame of $s_i$, there are no allocations. For allocations
        $\alloc{v_1}{v_2}$ from lower stack
        frames of $s_i$, the claim holds 
        since they are taken from $s_c$. Hence, $c_c \repby s_c$ implies
        $\slmemory^{\conttop{c}_c} \models \stateformula{\sigma_c(\alloc{v_1}{v_2})}_{SL}$
(where $\sigma_c(\alloc{v_1}{v_2}) = \sigma(\alloc{v_1}{v_2})$),
        and thus
        also    $\slmemory^{\conttop{c}_r} \models \stateformula{\sigma(\alloc{v_1}{v_2})}_{SL}$
as these stack frames are not
modified  during the evaluation from $c_c$ to $c_r$ and by the definition of the context
abstraction, all of these lower stack frames are still present in
$\conttop{c}_r$.\footnote{For that reason, we have $\slmemory^{\conttop{c}} =
  \slmemory^{c}$ for any context abstraction $\conttop{c}$ of any concrete state $c$.} 

Now we consider the allocations on the heap (i.e., from $\AL^{s_i}$). 
        Since $c_r \repby s_r$, all addresses within an allocation $\sigma_r(\alloc{v_1}{v_2})$ with
        $\alloc{v_1}{v_2} \in \AL^{s_r}$ are mapped to a value by $\slmemory^{\conttop{c}_r}$. 
        Hence, all addresses within an allocation $\sigma(\alloc{v_1}{v_2}) =
        \sigma_r(\renamed^{-1}(\alloc{v_1}{v_2}))$ with
        $\alloc{v_1}{v_2} \in \renamed(\AL^{s_r})$ are mapped to a value by
        $\slmemory^{\conttop{c}_r}$.
        
        Finally, since $c_c \repby s_c$, all addresses within an allocation
        $\sigma(\alloc{v_1}{v_2}) = \sigma_c(\alloc{v_1}{v_2})\linebreak[3]
        \in \AL^{s_c}$ are mapped to a  value by $\slmemory^{\conttop{c}_c}$. 
 For all addresses of those allocations in $\AL^{s_c}$ that are lost during
 the generalization from the call abstraction $s_{ca}$ to the entry state $s_e$
(i.e., where $\lostAL(s_{\mathit{ca}},s_e, \alloc{v_1}{v_2})$ holds),
 we know that they are not accessed
        (and modified) in the path to $s_r$ (else, this would yield the error state $\ERROR$).
        Therefore, since $c_r \repby s_r$, in $\slmemory^{\conttop{c}_r}$ these addresses are mapped to the
        same values.\label{item:alloc}
    \item Similar to \ref{item:alloc}, since $c_r \repby s_r$, for all entries
        $(v_1 \hookrightarrow_{\codevarx{ty}} v_2) \in \PT^{s_r}$, $\slmemory^{\conttop{c}_r}$ is a model of
      $\sigma_r(\stateformula{v_1 \hookrightarrow_{\codevarx{ty}} v_2}_\SL)$.
     So it is also a model of
        $\sigma_r(\renamed^{-1}(\stateformula{v_1 \hookrightarrow_{\codevarx{ty}} v_2}_\SL))$ for
      $(v_1 \hookrightarrow_{\codevarx{ty}} v_2)\in \renamed(\PT^{s_r})$
(where \mbox{\small $\sigma_r(\renamed^{-1}(\stateformula{v_1 \hookrightarrow_{\codevarx{ty}}
        v_2}_\SL)) = \sigma(\stateformula{v_1 \hookrightarrow_{\codevarx{ty}}
        v_2}_\SL)$}).
      
      Moreover, as argued in \ref{item:alloc},
 if an address
        corresponds to an allocation in $\AL^{s_c}$ that is lost during the generalization from the
        call abstraction $s_\mathit{ca}$ to the entry state $s_e$, then it is mapped to
        the same value by $\slmemory^{\conttop{c}_c}$ and  $\slmemory^{\conttop{c}_r}$.
        Hence, for all entries $(w_1 \hookrightarrow_{\codevarx{ty}} w_2) \in \PT^{s_c}$ where
        $\lostAL(s_\mathit{ca},s_e,\alloc{v_1}{v_2})$ holds for an allocation $\alloc{v_1}{v_2}$ that
        contains the address $w_1$, $\slmemory^{\conttop{c}_r}$ is a model of
        $\sigma_c(\stateformula{w_1 \hookrightarrow_{\codevarx{ty}} w_2}_\SL)$, i.e.,
of $\sigma(\stateformula{w_1 \hookrightarrow_{\codevarx{ty}} w_2}_\SL)$.\label{item:pointsto}
    \item With $c_r \repby s_r$ we know that $\sigma_r(\KB^{s_r})$ holds.
        Therefore, $\sigma_r(\renamed^{-1}(\renamed(\KB^{s_r})))$ (and hence $\sigma(\renamed(\KB^{s_r}))$)
        holds as well.
        Similarly, with $c_c \repby s_c$ we know that $\sigma_c(\KB^{s_c})$ and hence
        $\sigma(\KB^{s_c})$ holds, too.
        
        For each $\mu(v) = \renamed(w)$ in the third subset of $\KB^{s_i}$, note that
        $\sigma(\mu(v) = \renamed(w))$ is equal to
$\sigma_c(\mu(v)) = \sigma_r(\renamed^{-1}(\renamed(w)))$.
        Intuitively, $\sigma_c(\mu(v)) = \sigma_r(w)$ 
 holds for every $v \corrvar w \in \CV^{s_r}$ since in each symbolic execution
        step, we only add an entry to the component $\CV$ if during this step, the respective values are equal (and thus,
        in the corresponding concrete states, these symbolic variables have to be instantiated by the same
        values). For each entry $v \corrvar w \in \CV^{s_r}$, $v$ is a variable of $s_e$,
        and $\mu(v)$ is the corresponding variable in $s_c$. Thus, $\sigma_c(\mu(v))$ is equal to
        $\sigma_r(w)$.
    \item Since $\sigma_r(1 \leq v_1 \land v_1 \leq v_2)$ holds for all allocations
        $\alloc{v_1}{v_2} \in \AL^*(s_r)$, we also have
        $\sigma_r(\renamed^{-1}(1 \leq v_1 \land v_1 \leq v_2))$ for all
        $\alloc{v_1}{v_2} \in \renamed(\AL^*(s_r))$.
        Similarly, for all allocations $\alloc{v_1}{v_2} \in \AL^*(s_c)$,
        $\sigma_c(1 \leq v_1 \land v_1 \leq v_2)$ holds.
        Therefore, we have $\sigma(1 \leq v_1 \land v_1 \leq v_2)$ for all allocations of $s_i$.\label{item:alloc2}
    \item Since this condition holds for all pairs of allocations in $s_r$ resp. $s_c$,
      with the  reasoning as for \ref{item:alloc2}
            it also holds for all pairs of allocations in $s_i$ that originate from the
            same state.
            
        It remains to show for all pairs $\alloc{v_1}{v_2}, \alloc{w_1}{w_2} \in \AL^*(s_i)$ where
        $\alloc{v_1}{v_2} \in \renamed(\AL^*(s_r))$ and $\alloc{w_1}{w_2} \in \AL^*(s_c)$,
        that these allocations are
        disjoint. For stack allocations, this is trivial since the topmost stack frame of $s_i$ does
        not contain any allocations and the lower stack frames only contain allocations from $s_c$.

        Heap allocations are only added from $\AL^{s_c}$ if they have been removed in the generalization
        from the call abstraction $s_{\mathit{ca}}$ to the entry state $s_e$.
        In the concrete evaluation path from $c_c$ to $c_r$, allocation of already allocated areas is only
        possible if in the meantime, the area was freed.
        However, if \code{free} was invoked on an allocated area that is  lost during generalization, we would reach the
        error state $\ERROR$ during symbolic execution. Therefore, all allocations in $s_i$ that originate
        from $s_r$ are disjoint from those allocations in $s_c$ where
        $\lostAL(s_\mathit{ca},s_e,\alloc{w_1}{w_2})$ holds.\label{item:disjoint}
    \item We can follow the same line of reasoning as for \ref{item:alloc2}.
    \item For $(v_1 \hookrightarrow_{\codevar{ty}} v_2), (w_1 \hookrightarrow_{\codevar{ty}} w_2)
        \in \PT^{s_i}$, with \ref{item:pointsto} we have that $\slmemory^{\conttop{c}_r}$ is a
        model of $\sigma(\stateformula{v_1 \hookrightarrow_{\codevarx{ty}} v_2}_\SL \land
        \stateformula{w_1 \hookrightarrow_{\codevarx{ty}} w_2}_\SL)$.

        Recall that  $\slmemory^{\conttop{c}_r}$ is a model of
 $\sigma(\varphi)$ for all $\varphi \in \stateformula{s_i}$ that correspond to the cases
        (d)-(g). Let $\models \stateformula{s_i} \implies v_1 =
        w_1$ hold. If $v_1 = w_1$ is already implied by the subformulas  $\varphi \in \stateformula{s_i}$
from the cases (d)-(g), then  $\slmemory^{\conttop{c}_r}$ is also a model of $\sigma(v_1 =
w_1)$. Otherwise,         
        since $\stateformula{s_i}$ is the smallest set of formulas satisfying
        Def.\ \ref{def:StateFOLFormula}, one can use an inductive argument to show that       
        $\slmemory^{\conttop{c}_r}$ is also a model of $\sigma(v_1 = w_1)$. Thus,
    we have
          $\sigma(v_1) = \sigma(w_1) = n$ for some $n \in \Z$. Hence,
        $\slmemory^{\conttop{c}_r}$ is a model of
        $\stateformula{n \hookrightarrow_{\codevarx{ty}} \sigma(v_2)}_\SL \land
        \stateformula{n \hookrightarrow_{\codevarx{ty}} \sigma(w_2)}_\SL$, which implies
        $\sigma(v_2) = \sigma(w_2)$.\label{item:pointsto2}
    \item We can follow the same line of reasoning as for \ref{item:pointsto2}.
\end{enumerate}

Now we show that if the concrete \LLVM{} evaluation path $\pi$ is infinite,
then the corresponding sequence in $\GG$ is also infinite.
As stated above, a concrete evaluation step is represented by evaluation edges in the graph.
If there is an edge from $s_j$ to $s_{j+1}$ such that $c \repby s_j$ and $c \repby s_{j+1}$,
then this edge must be a call abstraction edge, a generalization edge, or a refinement edge,
for which we have the following application conditions:
\begin{itemize}
	\item A call abstraction is only performed after evaluation of a \code{call} instruction.
	\item A state may only be generalized if it has an incoming evaluation or call abstraction edge.
	\item Refinement is never performed on a state with an incoming refinement
          edge.
\end{itemize}
Therefore, the
longest possible sequence $s_j, s_{j+1}, s_{j+2}, \ldots$ in $\GG$ with $c \repby s_j$,
$c \repby s_{j+1}$, $c \repby s_{j+2}$, etc. has length 4, where $s_j$ and $s_{j+1}$ are connected
by a call abstraction edge, $s_{j+1}$ is generalized to $s_{j+2}$, and $s_{j+3}$ is a refinement
of $s_{j+2}$.

Hence, if the concrete \LLVM{} evaluation path $\pi$ is infinite, then
this can only be simulated
by an infinite symbolic execution $s_0, s_1, s_2, \ldots$ in $\GG$.
Here, 
each concrete \LLVM{} evaluation
step is represented by an evaluation edge in $\GG$, with only one exception: if a called
auxiliary function $\rfunc$  is
entered (in a state $c_e$) and returned (in a state $c_r$), then this path is summarized in the symbolic
execution graph by an intersection edge from a call state $s_c$ to an intersection state
$s_i$. Therefore, if we have an infinite number of concrete evaluation steps, then we also
have an infinite number of symbolic execution steps in the corresponding path in $\GG$.

On the other hand, if the concrete \LLVM{} evaluation path $\pi$ is finite and ends in a concrete
state
$c$, then one can simulate $\pi$
by a path in $\GG$ that ends in a state $s$ that weakly represents $c$. The reason is again
that  each concrete \LLVM{} evaluation
step is represented by an evaluation edge in $\GG$, with the exception of called
auxiliary functions $\rfunc$  that are
entered (in a state $c_e$) and returned (in a state $c_r$). Again, these paths are
summarized in the SEG  by an intersection edge from  $s_c$ to
$s_i$.
However, if the final state $c$ of $\pi$ is in the middle
of a call of an auxiliary function $\rfunc$, then the corresponding path in $\GG$ does not follow the
intersection edge, but it follows the call abstraction edge from the call state
$s_c$ to the call abstraction $s_{ca}$, and further via the generalization
edge to an entry state $s_e$ of $\rfunc$, and then stops in the middle of the path from
$\rfunc$'s entry state $s_e$ to its return state $s_r$.
\end{proof}


\setcounter{theorem}{\value{cor-memory-safety}}

\begin{corollary}[Memory Safety of \LLVM Programs]
 Let $\PP$ be a program with a complete
symbolic execution graph $\GG$. Then $\PP$ is memory safe for all states
represented by $\GG$.
\end{corollary}
\begin{proof}
  If
    $c_0$ is represented by a state $s_0$ in the SEG $\GG$, then $c_0
\cto^+ \ERROR$ implies that $\ERROR$ is the last state in a finite computation and 
 by Thm.\ \ref{thm:soundness-graph-construction-main-text}, 
there is a path from  $s_0$ to $\ERROR$ in $\GG$, which
contradicts the prerequisite that $\GG$ is complete. 
\end{proof}


\setcounter{theorem}{\value{thm-termination}}

\begin{theorem}[Termination]
Let
$\Prog$ be an \emph{\LLVM} program with a complete symbolic execution graph
$\GG$ and let  \(\I_1,\ldots,\I_m\) be the ITSs resulting from the SCCs of \(\GG\).
If all ITSs \(\I_1,\ldots,\I_m\) terminate,
then $\Prog$ also terminates for all concrete states $c$
that are represented
  by  a state of $\GG$.
\end{theorem}

\begin{proof}
Let $\pi = c_0 \cto c_1 \cto c_2 \cto \dots $ be an infinite evaluation sequence of
concrete states such
that $c_0$ is represented by some state $s_0$ in $\GG$. 
By  \cref{thm:soundness-graph-construction-main-text} there exists an infinite sequence
of states $s_{0},  s_{1},  s_{2}, \dotsc$ where $\GG$ has an edge from $s_{j-1}$ to $s_j$
if $j > 0$, and  there exist $0 = i_0 \leq i_1 \leq  \ldots$ with $c_{i_j} \repby s_j$
for all $j \geq 0$.
For any $i_j$, let $\sigma_{i_j}$ be the concrete instantiation with
$(\slassignment^{\conttop{c}_{i_j}},\slmemory^{\conttop{c}_{i_j}}) \models
\sigma_{i_j}(\stateformula{s_j}_\SL)$ for the context abstraction $\conttop{c}_{i_j}$ of
$c_{i_j}$ with $|\conttop{c}_{i_j}| = |s_j|$.

Clearly, termination of the ITSs $\I_1,\ldots,\I_m$ is equivalent to termination of their
union $\I = \I_1 \cup \ldots \cup \I_m$. Since  $\GG$ has an edge from $s_{j}$
to $s_{j+1}$ for all $j$, $\I$ also has a transition from  $s_{j}$
to $s_{j+1}$ with some condition $\CON_j$. We now show that for all $j \geq 0$ we have
\begin{equation}
  \label{ITSnonterm} \models (\sigma_{i_j} \cup \sigma'_{i_{j+1}})(\CON).
  \end{equation}
Here, for any instantiation $\sigma$, let $\sigma'$ be the corresponding instantiation of
the post-variables $\Vsym'$, i.e., $\sigma'(v')$ is defined to be $\sigma(v)$. 
Then \eqref{ITSnonterm} implies that there is an infinite evaluation with the ITS $\I$,
i.e., that $\I$ is not terminating.

To prove  \eqref{ITSnonterm},
we perform a case analysis based on the type of the edge between $s_{j}$ and $s_{j+1}$ in $\GG$.
\begin{itemize}
\item \underline{Generalization Edge:}
In this case, by construction $\I$ has a transition from $s_{j}$ to $s_{j+1}$ with the
condition $\CON = \stateformula{s_{j}} \cup \{v' = \mu(v) \mid v \in \Vsym(s_{j+1})\}$.
Recall that $(\slassignment^{\conttop{c}_{i_j}},\slmemory^{\conttop{c}_{i_j}}) \models
\sigma_{i_j}(\stateformula{s_j}_\SL)$.
By $\stateformula{s_j} \subseteq \stateformulaSL{s_j}$ 
and the fact that there are no occurrences of program variables or $\hookrightarrow$ in 
$\stateformula{s_{j}}$, we obtain
$\models \sigma_{i_j}(\stateformula{s_{j}})$.

Moreover, since the edge from $s_{j}$ to $s_{j+1}$ is a generalization edge, we have
$\sigma_{i_{j+1}}(v) = \sigma_{i_j}(\mu(v))$ for all $v \in \Vsym(s_{j+1})$.
We therefore have $\models (\sigma_{i_j} \cup 
\sigma_{i_{j+1}}')(\{v' = \mu(v) \mid v \in \Vsym(s_{j+1})\})$.
Together, we obtain $\models (\sigma_{i_j} \cup 
\sigma_{i_{j+1}}')(\CON)$, i.e., \eqref{ITSnonterm} holds.
\smallskip
\item \underline{All Other Edge Types:}
By construction $\I$ has a transition from $s_{j}$ to $s_{j+1}$ with the
condition $\CON = \stateformula{s_{j}} \cup \{v' = v \mid v \in \Vsym(s_j)\}$.
Using the same reasoning as for generalization edges, we get
$\models \sigma_{i_j}(\stateformula{s_{j}})$.

Since the edge from $s_{j}$ to $s_{j+1}$ is not a generalization edge,
we have $\sigma_{i_{j+1}}(v) = \sigma_{i_j}(v)$ for all $v \in \Vsym(s_j)$.
We therefore obtain $\models (\sigma_{i_j} \cup 
\sigma_{i_{j+1}}')( \{v' = v \mid v \in \Vsym(s_j)\})$.
Together, we have $\models (\sigma_{i_j} \cup \sigma_{i_{j+1}}')(\CON)$,
i.e., \eqref{ITSnonterm} holds.
\end{itemize}
\end{proof}

\setcounter{theorem}{\value{auxctr}}